\documentclass{article}
\usepackage{url}
\usepackage{color}
\usepackage{bm}
\usepackage{amsmath}
\usepackage{amssymb}\usepackage{amsthm}
\usepackage{amsfonts}
\usepackage{amscd}
\begin{document}
\theoremstyle{plain}
\newtheorem*{ithm}{Theorem}
\newtheorem*{idefn}{Definition}
\newtheorem{thm}{Theorem}[section]
\newtheorem{lem}[thm]{Lemma}
\newtheorem{dlem}[thm]{Lemma/Definition}
\newtheorem{prop}[thm]{Proposition}
\newtheorem{set}[thm]{Setting}
\newtheorem{cor}[thm]{Corollary}
\newtheorem*{icor}{Corollary}
\theoremstyle{definition}
\newtheorem{assum}[thm]{Assumption}
\newtheorem{notation}[thm]{Notation}
\newtheorem{defn}[thm]{Definition}
\newtheorem{clm}[thm]{Claim}
\newtheorem{ex}[thm]{Example}
\theoremstyle{remark}
\newtheorem{rem}[thm]{Remark}
\numberwithin{equation}{section}

\newcommand{\unit}{\mathbb I}
\newcommand{\ali}[1]{{\mathfrak A}_{[ #1 ,\infty)}}
\newcommand{\alm}[1]{{\mathfrak A}_{(-\infty, #1 ]}}
\newcommand{\nn}[1]{\lV #1 \rV}
\newcommand{\br}{{\mathbb R}}
\newcommand{\dm}{{\rm dom}\mu}
\newcommand{\inn}{({\rm {inner}})}
\newcommand{\Ad}{\mathop{\mathrm{Ad}}\nolimits}
\newcommand{\Proj}{\mathop{\mathrm{Proj}}\nolimits}
\newcommand{\RRe}{\mathop{\mathrm{Re}}\nolimits}
\newcommand{\RIm}{\mathop{\mathrm{Im}}\nolimits}
\newcommand{\Wo}{\mathop{\mathrm{Wo}}\nolimits}
\newcommand{\Prim}{\mathop{\mathrm{Prim}_1}\nolimits}
\newcommand{\Primz}{\mathop{\mathrm{Prim}}\nolimits}
\newcommand{\ClassA}{\mathop{\mathrm{ClassA}}\nolimits}
\newcommand{\Class}{\mathop{\mathrm{Class}}\nolimits}
\newcommand{\diam}{\mathop{\mathrm{diam}}\nolimits}
\def\qed{{\unskip\nobreak\hfil\penalty50
\hskip2em\hbox{}\nobreak\hfil$\square$
\parfillskip=0pt \finalhyphendemerits=0\par}\medskip}
\def\proof{\trivlist \item[\hskip \labelsep{\bf Proof.\ }]}
\def\endproof{\null\hfill\qed\endtrivlist\noindent}
\def\proofof[#1]{\trivlist \item[\hskip \labelsep{\bf Proof of #1.\ }]}
\def\endproofof{\null\hfill\qed\endtrivlist\noindent}

\newcommand{\qe}{\sim_{\mathop{q.e.}}}
\newcommand{\bgu}[2]{\beta_{#1}^{U#2}}
\newcommand{\SDC}{{\mathfrak A}_{\mathop{\mathrm SDC}}}
\newcommand{\sdc}{\SDC}
\newcommand{\mkA}{{\mathfrak A}}
\newcommand{\mkB}{{\mathfrak B}}
\newcommand{\mkC}{{\mathfrak C}}
\newcommand{\mkD}{{\mathfrak D}}
\newcommand{\mkh}{{\mathfrak h}}
\newcommand{\mkk}{{\mathfrak K}}
\newcommand{\mkK}{{\mathfrak K}}
\newcommand{\varphii}{\varphi}
\newcommand{\sdcizs}{\sdc\lmk \mkk_1\oplus \mkk_0\oplus \mkk_3,  \mkC_1\oplus \mkC_0\oplus \mkC_3\rmk}
\newcommand{\sdciz}{\sdc\lmk \mkk_1\oplus \mkk_0,  \mkC_1\oplus \mkC_0\rmk}
\newcommand{\sdcin}{\sdc\lmk \mkk_1\oplus \mkk_2,  \mkC_1\oplus \mkC_2\rmk}
\newcommand{\sdcis}{\sdc\lmk \mkk_1\oplus \mkk_3,  \mkC_1\oplus \mkC_3\rmk}
\newcommand{\sdckcz}[1]{\sdc^{(0)}\lmk\mkk_{#1},\mkC_{#1}\rmk}
\newcommand{\sdckc}[1]{\sdc\lmk\mkk_{#1},\mkC_{#1}\rmk}
\newcommand{\bh}[1]{\caB\lmk\caH_{#1}\rmk}
\newcommand{\pg}{{\mathfrak S}(\bbZ^2)}
\newcommand{\oo}{{\boldsymbol\varphii}}
\newcommand{\caA}{{\mathcal A}}
\newcommand{\caB}{{\mathcal B}}
\newcommand{\caC}{{\mathcal C}}
\newcommand{\caD}{{\mathcal D}}
\newcommand{\caE}{{\mathcal E}}
\newcommand{\caF}{{\mathcal F}}
\newcommand{\caG}{{\mathcal G}}
\newcommand{\caH}{{\mathcal H}}
\newcommand{\caI}{{\mathcal I}}
\newcommand{\caJ}{{\mathcal J}}
\newcommand{\caK}{{\mathcal K}}
\newcommand{\caL}{{\mathcal L}}
\newcommand{\caM}{{\mathcal M}}
\newcommand{\caN}{{\mathcal N}}
\newcommand{\caO}{{\mathcal O}}
\newcommand{\caP}{{\mathcal P}}
\newcommand{\caQ}{{\mathcal Q}}
\newcommand{\caR}{{\mathcal R}}
\newcommand{\caS}{{\mathcal S}}
\newcommand{\caT}{{\mathcal T}}
\newcommand{\caU}{{\mathcal U}}
\newcommand{\caV}{{\mathcal V}}
\newcommand{\caW}{{\mathcal W}}
\newcommand{\caX}{{\mathcal X}}
\newcommand{\caY}{{\mathcal Y}}
\newcommand{\caZ}{{\mathcal Z}}
\newcommand{\bbA}{{\mathbb A}}
\newcommand{\bbB}{{\mathbb B}}
\newcommand{\bbC}{{\mathbb C}}
\newcommand{\bbD}{{\mathbb D}}
\newcommand{\bbE}{{\mathbb E}}
\newcommand{\bbF}{{\mathbb F}}
\newcommand{\bbG}{{\mathbb G}}
\newcommand{\bbH}{{\mathbb H}}
\newcommand{\bbI}{{\mathbb I}}
\newcommand{\bbJ}{{\mathbb J}}
\newcommand{\bbK}{{\mathbb K}}
\newcommand{\bbL}{{\mathbb L}}
\newcommand{\bbM}{{\mathbb M}}
\newcommand{\bbN}{{\mathbb N}}
\newcommand{\bbO}{{\mathbb O}}
\newcommand{\bbP}{{\mathbb P}}
\newcommand{\bbQ}{{\mathbb Q}}
\newcommand{\bbR}{{\mathbb R}}
\newcommand{\bbS}{{\mathbb S}}
\newcommand{\bbT}{{\mathbb T}}
\newcommand{\bbU}{{\mathbb U}}
\newcommand{\bbV}{{\mathbb V}}
\newcommand{\bbW}{{\mathbb W}}
\newcommand{\bbX}{{\mathbb X}}
\newcommand{\bbY}{{\mathbb Y}}
\newcommand{\bbZ}{{\mathbb Z}}
\newcommand{\str}{^*}
\newcommand{\lv}{\left \vert}
\newcommand{\rv}{\right \vert}
\newcommand{\lV}{\left \Vert}
\newcommand{\rV}{\right \Vert}
\newcommand{\la}{\left \langle}
\newcommand{\ra}{\right \rangle}
\newcommand{\ltm}{\left \{}
\newcommand{\rtm}{\right \}}
\newcommand{\lcm}{\left [}
\newcommand{\rcm}{\right ]}
\newcommand{\ket}[1]{\lv #1 \ra}
\newcommand{\bra}[1]{\la #1 \rv}
\newcommand{\lmk}{\left (}
\newcommand{\rmk}{\right )}
\newcommand{\al}{{\mathcal A}}
\newcommand{\md}{M_d({\mathbb C})}
\newcommand{\eaut}{\mathop{\mathrm{EAut}}\nolimits}
\newcommand{\qaut}{\mathop{\mathrm{QAut}}\nolimits}
\newcommand{\sqaut}{\mathop{\mathrm{SQAut}}\nolimits}
\newcommand{\gsqaut}{\mathop{\mathrm{GSQAut}}\nolimits}
\newcommand{\QLS}{\mathop{\mathcal{SL}}\nolimits}
\newcommand{\haut}{\mathop{\mathrm{HAut}}\nolimits}
\newcommand{\guaut}{\mathop{\mathrm{GUQAut}}\nolimits}
\newcommand{\IG}{\mathop{\mathrm{IG}}\nolimits}
\newcommand{\IP}{\mathop{\mathrm{IP}}\nolimits}
\newcommand{\ainn}{\mathop{\mathrm{AInn}}\nolimits}
\newcommand{\id}{\mathop{\mathrm{id}}\nolimits}
\newcommand{\Tr}{\mathop{\mathrm{Tr}}\nolimits}
\newcommand{\co}{\mathop{\mathrm{co}}\nolimits}
\newcommand{\sym}{\mathop{\mathrm{Sym}}\nolimits}
\newcommand{\sgn}{\mathop{\mathrm{sgn}}\nolimits}
\newcommand{\Ran}{\mathop{\mathrm{Ran}}\nolimits}
\newcommand{\Ker}{\mathop{\mathrm{Ker}}\nolimits}
\newcommand{\Aut}{\mathop{\mathrm{Aut}}\nolimits}
\newcommand{\QAut}{\mathop{\mathrm{QAut}}\nolimits}
\newcommand{\EAut}{\mathop{\mathrm{EAut}}\nolimits}
\newcommand{\SPT}{\mathop{\mathrm{SPT}}\nolimits}
\newcommand{\Inn}{\mathop{\mathrm{Inn}}\nolimits}
\newcommand{\spn}{\mathop{\mathrm{span}}\nolimits}
\newcommand{\Mat}{\mathop{\mathrm{M}}\nolimits}
\newcommand{\Dia}{\mathop{\mathrm{D}}\nolimits}
\newcommand{\UT}{\mathop{\mathrm{UT}}\nolimits}
\newcommand{\DT}{\mathop{\mathrm{DT}}\nolimits}
\newcommand{\GL}{\mathop{\mathrm{GL}}\nolimits}
\newcommand{\spa}{\mathop{\mathrm{span}}\nolimits}
\newcommand{\supp}{\mathop{\mathrm{supp}}\nolimits}
\newcommand{\rank}{\mathop{\mathrm{rank}}\nolimits}
\newcommand{\idd}{\mathop{\mathrm{id}}\nolimits}
\newcommand{\ran}{\mathop{\mathrm{Ran}}\nolimits}
\newcommand{\dr}{ \mathop{\mathrm{d}_{{\mathbb R}^k}}\nolimits} 
\newcommand{\dc}{ \mathop{\mathrm{d}_{\cc}}\nolimits} \newcommand{\drr}{ \mathop{\mathrm{d}_{\rr}}\nolimits} 
\newcommand{\auz}{\Aut^{(0)}}
\newcommand{\zin}{\mathbb{Z}}
\newcommand{\rr}{\mathbb{R}}
\newcommand{\cc}{\mathbb{C}}
\newcommand{\ww}{\mathbb{W}}
\newcommand{\nan}{\mathbb{N}}\newcommand{\bb}{\mathbb{B}}
\newcommand{\aaa}{\mathbb{A}}\newcommand{\ee}{\mathbb{E}}
\newcommand{\pp}{\mathbb{P}}
\newcommand{\wks}{\mathop{\mathrm{wk^*-}}\nolimits}
\newcommand{\Hom}{\mathop{\mathrm{Hom}}\nolimits}
\newcommand{\mk}{{\Mat_k}}
\newcommand{\mnz}{\Mat_{n_0}}
\newcommand{\mn}{\Mat_{n}}
\newcommand{\dist}{\mathrm{d}}
\newcommand{\braket}[2]{\left\langle#1,#2\right\rangle}
\newcommand{\ketbra}[2]{\left\vert #1\right \rangle \left\langle #2\right\vert}
\newcommand{\abs}[1]{\left\vert#1\right\vert}
\newtheorem{nota}{Notation}[section]
\def\qed{{\unskip\nobreak\hfil\penalty50
\hskip2em\hbox{}\nobreak\hfil$\square$
\parfillskip=0pt \finalhyphendemerits=0\par}\medskip}
\def\proof{\trivlist \item[\hskip \labelsep{\bf Proof.\ }]}
\def\endproof{\null\hfill\qed\endtrivlist\noindent}
\def\proofof[#1]{\trivlist \item[\hskip \labelsep{\bf Proof of #1.\ }]}
\def\endproofof{\null\hfill\qed\endtrivlist\noindent}
\newcommand{\pa}[1]{\pi_{#1}\alpha_{#1}}
\newcommand{\pat}{\lmk \pa L\hat\otimes \pa R\rmk}
\newcommand{\ao}[1]{a_\omega(#1)}
\newcommand{\eg}[2]{(-1)^{\ao{#1}}#2}
\newcommand{\unic}{\unit_{\bbC^2}}
\newcommand{\unih}[1]{\unit_{\caH_{#1}}}
\newcommand{\ZZ}{\bbZ_2\times\bbZ_2}
\newcommand{\SSS}{\mathcal{S}}
\newcommand{\cs}{S}
\newcommand{\ct}{t}
\newcommand{\hS}{S}
\newcommand{\vv}{{\boldsymbol v}}
\newcommand{\ala}{a}
\newcommand{\bet}{b}
\newcommand{\gam}{c}
\newcommand{\alphas}{\alpha}
\newcommand{\alphai}{\alpha^{(\sigma_{1})}}
\newcommand{\alphan}{\alpha^{(\sigma_{2})}}
\newcommand{\betas}{\beta}
\newcommand{\betai}{\beta^{(\sigma_{1})}}
\newcommand{\betan}{\beta^{(\sigma_{2})}}
\newcommand{\alphass}{\alpha^{{(\sigma)}}}
\newcommand{\uu}{V}
\newcommand{\vp}{\varsigma}
\newcommand{\vpr}{R}
\newcommand{\tg}{\tau_{\Gamma}}
\newcommand{\sgg}{\Sigma_{\Lambda}}
\newcommand{\nh}{1}
\newcommand{\rk}{2,a}
\newcommand{\nii}{1,a}
\newcommand{\nhh}{3,a}
\newcommand{\sjt}{2}
\newcommand{\sjtg}{2}
\newcommand{\bcg}{\caB(\caH_{\alpha})\otimes  C^{*}(\Sigma_{\Lambda})}
\newcommand{\cacn}[2]{\caA_{C_{#1}\cap H_{#2}^{{c_{#2}^{(2)}}}}}
\newcommand{\caci}[2]{\caA_{C_{#1}\cap H_{#2}^{{c_{#2}^{(1)}}}}}
\newcommand{\cac}[2]{\caA_{C_{#1}\cap H_{#2}^{c_{#2}}}}
\newcommand{\cacc}[3]{\caA_{\lmk C_{#1}\cap H_{#2}^{c^{(#3)}_{#2}} \rmk^c\cap H_{#2} }}
\newcommand{\Uo}{{\mathrm U}(1)}
\newcommand{\css}[2]{c_{#1}^{#2}}
\newcommand{\uss}[2]{u_{#1}^{#2}}
\newcommand{\kss}[2]{\kappa_{#1}^{#2}}
\newcommand{\mss}[3]{m_{#1#2}^{#3}}

\newcommand{\PD}{\mathcal{PD}(G)}
\newcommand{\PDz}{\mathcal{PD}_0(G)}
\newcommand{\ptp}{\lmk\pi_L\hat\otimes\pi_R\rmk}
\newcommand{\bebe}[3]{\bgu{#2}#1\lmk\eta_{#2#1}^\epsilon\rmk^{-1}
\bgu{#3}#1\lmk \lmk \eta_{#3#1}^{\eg{#2}\epsilon}\rmk_{-\ao{#2}\epsilon}\rmk^{-1}
\eta_{#2#3,#1}^\epsilon \lmk \bgu{#2#3}#1\rmk^{-1}}
\newcommand{\etas}[3]{\eta_{#1#2}^{#3}}
\newcommand{\os}[1]{\omega_{#1}}
\newcommand{\bebey}[4]{\bgu{#2}#1\lmk\eta_{#2#1}^{\epsilon #4}\rmk^{-1}
\bgu{#3}#1\lmk \lmk \eta_{#3#1}^{\eg{#2}{\epsilon #4}}\rmk_{-\ao{#2}{\epsilon }}\rmk^{-1}
\eta_{#2#3,#1}^{\epsilon #4} \lmk \bgu{#2#3}#1\rmk^{-1}}
\newcommand{\tdnh}{\tilde\delta_{28}}
\newcommand{\dnh}{\delta_{\nhh}}
\newcommand{\eijz}[2]{e_{#1#2}^{(\Lambda_0)}}
\newcommand{\tdsj}{\tilde \delta_{30}}
\newcommand{\tgsj}{\tilde\caG_{30}}
\newcommand{\dn}{\delta_{\nii}}
\newcommand{\tdnhs}{\tilde\delta_{28}}
\newcommand{\tdsjs}{\tilde \delta_{30}}
\newcommand{\sigman}{i}
\newcommand{\Rn}{2}
\newcommand{\Ln}{1}
\newcommand{\dro}{\delta_{\rk}}
\newcommand{\che}[1]{#1}

\title{An invariant of symmetry protected topological phases with on-site finite group symmetry for two-dimensional Fermion systems}

\author{Yoshiko Ogata \thanks{ Graduate School of Mathematical Sciences
The University of Tokyo, Komaba, Tokyo, 153-8914, Japan
Supported in part by
the Grants-in-Aid for
Scientific Research, JSPS.}}
\maketitle
\begin{abstract}
We consider SPT-phases with on-site finite group $G$ symmetry for two-dimensional Fermion systems.
We derive an invariant of the classification.
\end{abstract}

\section{Introduction}\label{introsec}
The notion of symmetry protected topological (SPT) phases was introduced by Gu and Wen [GW].
It is defined as follows:
we consider the set of all Hamiltonians with some symmetry, 
which have a unique gapped ground state in the bulk, and can be smoothly deformed into
a common trivial gapped Hamiltonian without closing the gap.
We say 
two such Hamiltonians are equivalent, if they can be smoothly deformed into
each other, without breaking the symmetry.
We call an equivalence class of this classification, a
symmetry protected topological (SPT) phase.
In \cite{1dFermi} we derived $\bbZ_2\times H^1(G, \Uo)\times H^2(G,\Uo)$-valued
invariant of one-dimensional Fermionic SPT-phases.
In this paper, we derive an invariant of SPT-phases  in two-dimensional Fermionic systems.
 
We start by summarizing standard setup of Fermionic systems on the two dimensional lattice $\bbZ^{2}$ \cite{BR1,BR2, EK}. We will use freely the basic notation in section \ref{notasec}.
We also use the facts and notation
of graded $C^*$-algebra from \cite{bla} and \cite{1dFermi}.

Let us first recall the definition of the self-dual CAR-algebra introduced by Araki \cite{sdc}.
(See also \cite{EK} Chapter 6.)
\begin{defn}
For a Hilbert space $\mkK$ with a complex conjugation $\mkC$ (i.e., anti-unitary such that $\mkC=\mkC^*$),
self-dual-CAR-algebra $\SDC(\mkK,\mkC)$
over $(\mkK,\mkC)$
is defined as  the universal enveloping $C^*$-algebra
generated by $\{B(f)\mid f\in \mkK\}$
such that
\begin{align}
\begin{split}
\mkK\ni f\mapsto B(f),\quad \text{linear}\\
\left\{ B(f), B(g)\right\}=\braket{f}{g}\unit,\\
B(f)^*=B(\mkC f),\quad f,g\in \mkK.
\end{split}
\end{align}
\end{defn}
If $u\in\caU(\mkK)$ satisfies $u\mkC=\mkC u$, then 
there exists an automorphism $\Xi_u\in \Aut(\SDC(\mkK,\mkC))$
such that $\Xi_u\lmk B(f)\rmk=B\lmk uf\rmk$, $f\in\mkK$.
In particular, for $u=-\unit_\mkK$, 
$\Theta_\mkK:=\Xi_{-\unit}$ defines an automorphism on $\SDC(\mkK,\mkC)$
satisfying $\Theta_\mkK^2=\id$. Hence it defines a grading over $\SDC(\mkK,\mkC)$.
In general, for a graded $C^*$-algebra $\mkB$, we denote by $\mkB^{(0)}$
its even part and by $\mkB^{(1)}$
its odd part. Elements in $\mkB^{(0)}$ or $\mkB^{(1)}$ are
said to be homogeneous.
For a homogeneous element $b\in \mkB$,
we denote by $\partial b$ the grading of $b$.
With this grading, $\SDC(\mkK,\mkC)$ has an odd self-adjoint unitary
(consider $B(f)$ for $f\in \mkK$ with $\mkC f=f$, $\lV f\rV^2=2$.)
We say a state $\omega$ on $\SDC(\mkK,\mkC)$ is homogeneous if it is invariant under $\Theta_\mkk$.
When an automorphism $\alpha$ on $\Aut(\SDC(\mkK,\mkC))$ commutes with $\Theta_\mkk$,
we say that $\alpha$ is graded. We denote by $\Aut^{(0)}\lmk\SDC(\mkK,\mkC)\rmk$
the set of all graded automorphisms on $\Aut(\SDC(\mkK,\mkC))$.

A basis projection for $(\mkK,\mkC)$ is an orthogonal projection
on $\mkK$ such that $p+\mkC p\mkC=\unit_\mkK$.
Basis projection exists if $\mkK$ is even or infinite dimensional.
If $p$ is a basis projection for $(\mkK,\mkC)$,
the self-dual CAR-algebra $\SDC(\mkK,\mkC)$ is $*$-isomorphic to the CAR-algebra
$\mkA_{\rm CAR}(p\mkK)$ over $p\mkK$
via a $*$-isomorphism $\gamma_{p} : \SDC(\mkK,\mkC)\to \mkA_{\rm CAR}(p\mkK)$
such that
\begin{align}\label{gamp}
\gamma_{p}\lmk B(f)\rmk=a^*(pf)+a(p\mkC f),\quad f\in \mkK.
\end{align}
Here, $a^*(g), a(g)$, $g\in p\mkK$ denotes the creation and anihilation operators of the CAR-algebra $\mkA_{\rm CAR}(p\mkK)$.
With the Fock state $\omega$ over $\mkA_{\rm CAR}(p\mkK)$ (i.e., the state with $\omega(a^*(g)a(g))=0$ for any
$g\in p\mkK$), via the $*$-isomorphism above,
we can define a state $\omega_p:=\omega\gamma_p$ over $\SDC(\mkK,\mkC)$.

For Hilbert spaces $\mkK_1,$ $\mkK_2$ with complex conjugations $\mkC_1,\mkC_2$,
there is a $*$-isomorphism
$\gamma_{12}: \SDC(\mkK_1,\mkC_1)\hat\otimes \SDC(\mkK_2,\mkC_2)
\to \SDC\lmk \mkK_1\oplus \mkK_2,\mkC_1\oplus \mkC_2\rmk$
such that
\begin{align}
\gamma_{12} \lmk B_1(f_1)\hat\otimes \unit +\unit\hat\otimes B_2(f_2)\rmk
=B(f_1\oplus f_2),\quad f_1\in\mkK_1,\; f_2\in\mkK_2.
\end{align}
(Recall the graded tensor product $\hat\otimes$ from \cite{bla} section 14.)
Here we denoted the generators of $\SDC(\mkK_i,\mkC_i)$ by $B_i(f_i)$, $f_i\in\mkK_i$.
We identify $\SDC(\mkK_1,\mkC_1)\hat\otimes \SDC(\mkK_2,\mkC_2)$ and
$ \SDC\lmk \mkK_1\oplus \mkK_2,\mkC_1\oplus \mkC_2\rmk$ via this isomorphism
throughout this paper, without writing
$\gamma_{12}$ explicitly.
In particular, for homogeneous states $\varphi_1,\varphi_2$ on  $\SDC(\mkK_1,\mkC_1)$,
$ \SDC(\mkK_2,\mkC_2)$,
we denote the state $\lmk \varphi_1\hat \otimes\varphi_2\rmk\circ\gamma_{12}^{-1}$ 
simply by $\varphi_1\hat \otimes\varphi_2$.
(Recall that $\hat\varphi_1\hat\otimes\hat\varphi_2$ is a state on
$\SDC(\mkK_1,\mkC_1)\hat\otimes \SDC(\mkK_2,\mkC_2)$
such that $(\hat\varphi_1\hat\otimes\hat\varphi_2)(a\hat\otimes b)=\varphi_1(a)\varphi_2(b)$.
 )
For graded automorphisms $\alpha_1\in \Aut^{(0)}\lmk \SDC(\mkK_1,\mkC_1)\rmk$,
$\alpha_2\in \Aut^{(0)}\lmk \SDC(\mkK_2,\mkC_2)\rmk$,
we denote the automorphism $\gamma_{12}\lmk\alpha_1\hat\otimes \alpha_2\rmk\gamma_{12}^{-1}$
simply by $\alpha_1\hat\otimes \alpha_2$,

Throughout this paper, we fix some $d\in2 \nan$.
For each $k\in\bbZ$, we set
\begin{align}
\begin{split}
H_{U}^{k}:=\bbZ\times \bbZ_{\ge k},\quad
H_{D}^{k}:=\bbZ\times \bbZ_{\le k},\quad
H_{L}^{k}:=\bbZ_{\le k}\times\bbZ,\quad
H_{R}^{k}:=\bbZ_{\ge k}\times \bbZ.
\end{split}
\end{align}
In particular,
left, right, upper, lower half planes are denoted by $H_L:=H_L^{-1}$, $H_R:=H_R^{0}$,
$H^U:=H_U^{0}$, $H_D:=H_D^{-1}$.
We set $\mkh:=l^2(\bbZ^2)$.
Let $\{\delta_{(x,y)} \mid (x,y)\in \bbZ^{2}\}$ be the standard basis of $l^{2}(\bbZ^{2})$,
and $\mathfrak C$, the complex conjugation on $l^{2}(\bbZ^{2})$
with respect to it.

For each $X\subset \bbZ^{2}$, we set
\begin{align}
\hat X:=
\left\{
(dx+j,y)\mid (x,y)\in X,\; j=0,\ldots,d-1
\right\}, 
\quad \mathfrak h_{X}:=l^{2}(\hat X)\subset l^{2}(\bbZ^{2}).
\end{align}
We denote the restriction $\mathfrak C\vert_{\mathfrak h_{X}}$
of $\mathfrak C$ on $\mathfrak h_{X}$
by $\mathfrak C_{X}$, 
and set
\begin{align}
\caA_{X}:=\mathfrak A_{\rm SDC}\lmk \mathfrak h_{X}, \mathfrak C_{X}\rmk,\quad
\tilde \caA_X:=\SDC\lmk\ l^2(X), \mkC\vert_{l^2(X)}\rmk.
\end{align}
In particular, we set $\al:=\al_{\bbZ^2}$.
We denote by $\Theta_X:=\Theta_{\mathfrak h_{X}}$ the grading automorphism
on $\caA_X$.
For $X=H_L,H_R$, we also set $\Theta_L:=\Theta_{H_L}$,
$\Theta_R:=\Theta_{H_R}$.

Now we define a reference state.
For each $(x,y)\in\bbZ^2$,
$\{\delta_{(z,y)}\mid z=dx,dx+1,\ldots,dx+d-1\}$
is a CONS of $\mkh_{(x,y)}:=\mkh_{\{(x,y)\}}$.
One can decompose $\mkh_{(x,y)}$ as
\begin{align}
\mkh_{(x,y)}=
\bigoplus_{j=0}^{\frac d2-1}\mkh_{(x,y)}^{(j)}
\end{align}
with its $2$-dimensional subspaces
\begin{align}
\mkh_{(x,y)}^{(j)}:=
\bbC-\spa\left\{
\delta_{(dx+2j,y)}, \delta_{(dx+2j+1,y)}
\right\},\quad j=0,\ldots,\frac d2-1.
\end{align}
Because each $\mkh_{(x,y)}^{(j)}$ is invariant under $\mkC$,
its restriction $\mkC^{(j)}_{(x,y)}:=\mkC\vert_{\mkh_{(x,y)}^{(j)}}$ gives
a complex conjugation 
on $\mkh_{(x,y)}^{(j)}$.

Let $p^{(j)}_{(x,y)}$
be the orthogonal projection on $\mkh_{(x,y)}^{(j)}$
 onto 
the one-dimensional subspace
spanned by $\delta_{(dx+2j,y)}+i\delta_{(dx+2j+1,y)}$.
By the definition, we see
that $\mkC^{(j)}_{(x,y)}$ is a basis projection for
$(\mkh_{(x,y)}^{(j)},\mkC^{(j)}_{(x,y)})$, i.e.,
\begin{align}\label{pjpa}
\mkC^{(j)}_{(x,y)} p^{(j)}_{(x,y)} \mkC^{(j)}_{(x,y)}
=\unit_{\mkh_{(x,y)}^{(j)}}-p^{(j)}_{(x,y)}.
\end{align}
Set projections $p_{(x,y)}$, $(x,y)\in\bbZ^2$,
$p_X$, $X\subset\bbZ^2$ on
$\mkh_{(x,y)}$, $\mkh_X$ respectively by
\begin{align}\label{pxx}
p_{(x,y)}:=\bigoplus_{j=0}^{\frac d2-1} p_{(x,y)}^{(j)} ,\quad
p_X:=\bigoplus_{(x,y)\in X} p_{(x,y)}.
\end{align}
By (\ref{pjpa}),
we see that $p_X$ is a basis projection for $(\mkh_X,\mkC_X)$
\begin{align}
\mkC_X p_X\mkC_X=\unit_{\mkh_X}-p_X.
\end{align}
In particular, $p:=p_{\bbZ^2}$ is a basis projection
for $(\mkh,\mkC)$.
From this basis projection $p$, we can 
construct a Fock state $\omega^{(0)}:=\omega_p$
on $\al$.
Set $g_{j,(x,y)}:=\delta_{(dx+2j,y)}+i\delta_{(dx+2j+1,y)}$,
for $(x,y)\in\bbZ^2$ and $j=0,\ldots,\frac d2-1$.
With $\gamma_p$ in (\ref{gamp}),
we note that
\begin{align}\label{abp}
\gamma_p^{-1}
\lmk a^*(p g_{j,(x,y)}) a(p g_{j,(x,y)}) \rmk
=B(g_{j,(x,y)}) B(g_{j,(x,y)})^*.
\end{align}

This $\omega^{(0)}$ is our reference state.
We also define Fock states on $\al_{H_L}$, $\al_{H_R}$
out of
basis projections $p_{H_L}$, $p_{H_R}$ on
$\mkh_{H_L}$, $\mkh_{H_R}$,
$\omega_L:=\omega_{p_{H_L}}$, $\omega_R:=\omega_{p_{H_R}}$. 
From the structure we
see that 
\begin{align}\label{zlrten}
\omega^{(0)}=\omega_{L}\hat\otimes \omega_{R}.
\end{align}

Throughout this paper, we fix a finite group $G$
and a unitary representation $U$
on $\bbC^d$,
commuting with the complex conjugation 
with respect to the standard basis of $\bbC^d$.
Identifying the standard basis of $\mkh_{(x,y)}$
$\{ \delta_{(dx+k,y)}\mid k=0,\ldots,d-1\}$
with that of $\bbC^d$,
we write the copy of $U$
on $\mathfrak h_{(x,y)}$ as $U_{(x,y)}$.
Note that $U_{(x,y)}(g)$, $g\in G$ and $\mkC_{(x,y)}$
commute,
hence $U_X(g):=\bigoplus_{(x,y)\in X} U_{(x,y)}(g)$, $g\in G$
and $\mkC_X$
commute, for each $X\subset\bbZ^2$.
As a result, for each $X\subset \bbZ^2$
and $g\in G$ there is an automorphism $\beta_g^X:=\Xi_{U_X(g)}$ such that
\begin{align}
\beta_g^X\lmk B(f)\rmk
=B\lmk U_X(g) f\rmk,\quad f\in \mkh_X.
\end{align}
We set
\begin{align}
\begin{split}
\beta_g:=\beta_g^{\bbZ^2},\quad
\beta_g^U:=\beta_g^{H_U^1},\quad\beta_g^D:=\beta_g^{H_D^1},\quad
\beta_g^{UR}:=\beta_f^{H_U^1\cap H_R},\quad \beta_g^{UL}:=\beta_f^{H_U^1\cap H_L}
\quad g\in G.
\end{split}
\end{align}


A mathematical model on two dimensional Fermionic system is fully specified by its
even  interaction $\Phi$.
We denote the set of all finite subsets of $\bbZ^2$ by ${\mathfrak S}_{\bbZ^{2}}$.
A uniformly bounded even interaction on $\caA$ 
is a map $\Phi: {\mathfrak S}_{\bbZ^{2}}\to \caA^{(0)}$ such that
\begin{align}
\Phi(X)=\Phi(X)^*\in \caA_{X}^{(0)},\quad X\in {\mathfrak S}_{\bbZ^{2}},
\end{align}
and 
\begin{align}
\sup_{X\in {\mathfrak S}_{\bbZ^{2}}}\lV \Phi(X)\rV<\infty.
\end{align}
It is of finite range with interaction length less than or equal to $R\in\nan$ if 
$\Phi(X)=0$ for any $X\in {\mathfrak S}_{\bbZ^{2}}$
whose diameter is larger than $R$.
An on-site interaction, i.e., an interaction with
$\Phi(X)=0$ unless $X$ consists of a single point, is said to be trivial.
An even interaction $\Phi$ is $\beta$-invariant if $\beta_g(\Phi(X))=\Phi(X)$
for any $X\in {\mathfrak S}_{\bbZ^{2}}$.
%
For a uniformly bounded and finite range even interaction $\Phi$ and $\Lambda\in {\mathfrak S}_{\bbZ^{2}}$
define the local Hamiltonian
\begin{align}
\lmk H_\Phi\rmk_\Lambda
:=\sum_{X\subset\Lambda} \Phi(X),
\end{align}
and denote the dynamics
\begin{align}
\tau^{(\Lambda)\Phi}_t (A):=e^{it\lmk H_\Phi\rmk_\Lambda}Ae^{-it\lmk H_\Phi\rmk_\Lambda},
\quad t\in \bbR,\quad A\in\caA.
\end{align}
By the uniform boundedness and finite rangeness of $\Phi$, 
 for each $A\in\caA$, the following limit exists
\begin{align}
\lim_{\Lambda\to\bbZ^{2}} \tau^{(\Lambda),\Phi}_t\lmk
A\rmk=:
\tau^{\Phi}_t\lmk A\rmk,\quad t\in\bbR,
\end{align}
and defines the dynamics $\tau^{\Phi}$ on $\caA$.
%
For a uniformly bounded and finite range even interaction $\Phi$,
a state $\varphi$ on $\caA$ is called a \mbox{$\tau^{\Phi}$-ground} state
if the inequality
$
-i\,\varphi(A^*{\delta_{\Phi}}(A))\ge 0
$
holds
for any element $A$ in the domain $\caD({\delta_{\Phi}})$ of the generator ${\delta_\Phi}$.
Let $\varphi$ be a $\tau^\Phi$-ground state, with a GNS triple $(\caH_\varphi,\pi_\varphi,\Omega_\varphi)$.
Then there exists a unique positive operator $H_{\varphi,\Phi}$ on $\caH_\varphi$ such that
$e^{itH_{\varphi,\Phi}}\pi_\varphi(A)\Omega_\varphi=\pi_\varphi(\tau^t_\Phi(A))\Omega_\varphi$,
for all $A\in\caA$ and $t\in\mathbb R$.
We call this $H_{\varphi,\Phi}$ the bulk Hamiltonian associated with $\varphi$.
\begin{defn}
We say that an interaction $\Phi$ has a unique gapped ground state if 
(i)~the $\tau^\Phi$-ground state, which we denote as $\omega_{\Phi}$, is unique, and 
(ii)~there exists a $\gamma>0$ such that
$\sigma(H_{\omega_{\Phi},\Phi})\setminus\{0\}\subset [\gamma,\infty)$, where  $\sigma(H_{\omega_{\Phi},\Phi})$ is the spectrum of $H_{\omega_{\Phi},\Phi}$.
We denote by $\caP_{UG}
$ the set of all uniformly bounded finite range even
interactions, with unique gapped ground state.
We denote by $\caP_{UG\beta}$ the set of all uniformly bounded finite range 
{\it $\beta$-invariant} even interactions, with unique gapped ground state.
\end{defn}
Set $\Phi_p :\mathfrak S_{\bbZ^2}\to \caA$ by
\begin{align}\label{ppdef}
\Phi_p(\{(x,y)\}):=\sum_{j=1}^{\frac d2 -1}B(g_{j,(x,y)}) B(g_{j,(x,y)})^*
-B(g_{j,(x,y)})^* B(g_{j,(x,y)}),\quad
(x,y)\in \bbZ^2
\end{align}
and $\Phi_p(X)=0$ otherwise.
Then we see that $\Phi_p$ is a uniformly bounded on-site
interaction with a unique gapped ground state
$\omega^{(0)}:=\omega_p$ \cite{sdc}.

In this paper we consider a classification problem of a 
subset of $\caP_{UG\beta}$.
To describe this subset we need 
to explain the classification problem of unique gapped ground state phases, without symmetry.
For $\Gamma\subset\bbZ^{2}$, we denote by  $\Pi_{\Gamma}:\caA\to \caA_{\Gamma}$ the conditional expectation with respect to the trace state (see Theorem 4.7 \cite{aramori}).
Note from the proof of Theorem 4.7 \cite{aramori}
that $\Pi_\Gamma(a\hat\otimes b)=0$
for any $a\in \caA_\Gamma$ and $b\in\caA_{\Gamma^c}^{(1)}$.
Let $f:(0,\infty)\to (0,\infty)$ be a continuous decreasing function 
with $\lim_{t\to\infty}f(t)=0$.
For each $A\in\caA$, let
\begin{align}\label{dzeta}
\lV A\rV_f:=\lV A\rV
+ \sup_{N\in \nan}\lmk\frac{\lV
A-\Pi_{\Lambda_N}
(A)
\rV}
{f(N)}
\rmk.
\end{align}
We denote by $\caD_f$ the set of all $A\in\caA$ such that
$\lV A\rV_f<\infty$.
Here $\Lambda_N:=[-N,N]^{\times 2}\cap \bbZ^2$.

The classification of unique gapped ground state phases $\caP_{UG}
$ without symmetry is the following.
\begin{defn}\label{classsym}
Two interactions $\Phi_0,\Phi_1\in\caP_{UG}
$ are equivalent if there is
a path of even  interactions
 $\Phi : [0,1]\to \caP_{UG}
$
satisfying the following:
\begin{enumerate}
\item $\Phi(0)=\Phi_0$ and $\Phi(1)=\Phi_1$.
\item
For each $X\in{\mathfrak S}_{\bbZ^2}$, the map
$[0,1]\ni s\to \Phi(X;s)\in\caA_{X}^{(0)}$ is  $C^1$.
We denote by $\dot{\Phi}(X;s)$ 
the corresponding derivatives.
The interaction obtained by differentiation is denoted by $\dot\Phi(s)$, for each $s\in[0,1]$.
\item
There is a number $R\in\nan$
such that $X \in {\mathfrak S}_{\bbZ^2}$ and $\diam{X}\ge R$ imply $\Phi(X;s)=0$, for all $s\in[0,1]$.
\item
Interactions are bounded as follows
\begin{align}
C_b^{\Phi}:=\sup_{s\in[0,1]}\sup_{X\in {\mathfrak S}_{\bbZ^2}}
\lmk
\lV
\Phi\lmk X;s\rmk
\rV+|X|\lV
\dot{\Phi} \lmk X;s\rmk
\rV
\rmk<\infty.
\end{align}
\item
Setting
\begin{align}
b(\varepsilon):=\sup_{Z\in{\mathfrak S}_{\bbZ^2}}
\sup_{s,s_0 \in[0,1],0<| s-s_0|<\varepsilon}
\lV
\frac{\Phi(Z;s)-\Phi(Z;s_0)}{s-s_0}-\dot{\Phi}(Z;s_0)
\rV
\end{align}
for each $\varepsilon>0$, we have
$\lim_{\varepsilon\to 0} b(\varepsilon)=0$. 
\item
There exists a $\gamma>0$ such that
$\sigma(H_{\omega_{\Phi(s)},\Phi(s)})\setminus\{0\}\subset [\gamma,\infty)$ for
all $s\in[0,1]$, where  $\sigma(H_{\omega_{\Phi(s)},\Phi(s)})$ is the spectrum of $H_{\omega_{\Phi(s)},\Phi(s)}$.
\item
There exists an $0<\eta<1$ satisfying the following:
Set $\zeta(t):=e^{-t^{ \eta}}$.
Then for each $A\in D_\zeta$, 
$\omega_{\Phi(s)}(A)$ is differentiable with respect to $s$, and there is a constant
$C_\zeta$ such that:
\begin{align}\label{dcon}
\lv
\frac{d}{ds}\omega_{\Phi(s)}(A)
\rv
\le C_\zeta\lV A\rV_\zeta,
\end{align}
for any $A\in D_\zeta$.(Recall (\ref{dzeta})).
\end{enumerate}
We write $\Phi_0\sim\Phi_1$ if
$\Phi_0$ and $\Phi_1$ are equivalent. 
If $\Phi_0,\Phi_1\in\caP_{UG\beta}$ and
if we can take the path in $\caP_{UG\beta}$, i.e.,
 so that $\beta_g\lmk \Phi(X;s)\rmk=\Phi(X;s)$, $g\in G$ for all $s\in[0,1]$,
then we say $\Phi_0$ and $\Phi_1$ are $\beta$-equivalent
and write  $\Phi_0\sim_\beta\Phi_1$.
\end{defn}
The object we classify in this paper is the following:
\begin{defn}
We denote by $\caP_{SL\beta}$ the set of all 
$\Phi\in\caP_{UG\beta}$
such that $\Phi\sim\Phi_p$
with $\Phi_p\in \caP_{UG}$ defined in (\ref{ppdef}).
Connected components of 
 $\caP_{SL\beta}$ with respect to
$\sim_\beta$ are the
symmetry protected topological (SPT)-phases.
\end{defn}
In this paper, we introduce an invariant of 
this classification.

\section{Main result}\label{mainsec}
Our invariant is
given by
\begin{align}
\lmk C^3(G,\Uo\oplus\Uo)\rmk
\times \lmk H^2(G, \bbZ_2\oplus\bbZ_2)\rmk
\times \lmk H^1(G, \bbZ_2)\rmk,
\end{align}
devided by some equivalence relation.
First let us specify it.

For $A:=\bbZ_2, \Uo$, we define a $\bbZ_2$-action on $A\oplus A$
by 
\begin{align}
\bbZ_2\times \lmk A\oplus A\rmk\ni (a,x)
\mapsto 
 x^{a}
:=\begin{pmatrix}
0&1\\1&0
\end{pmatrix}^a x\in A\oplus A.
\end{align}
We associate $A\oplus A$ the point-wise multiplication , i.e.,
for $x=(x_+,x_-), y=(y_+,y_-)\in A\oplus A$, we set $x\cdot y:=(x_+y_+, x_-y_-)$.
For $x=(x_+,x_-)\in \bbZ_2\oplus \bbZ_2$, we also
set $(-1)^x:=((-1)^{x_+}, (-1)^{x_-})\in \Uo\oplus \Uo$.
For $x\in C^1(G,A\oplus A)$, $y\in C^2(G,A\oplus A)$, $z\in C^3(G,A\oplus A)$ and $a\in H^1(G,\bbZ_2)$,
we set
\begin{align}
\begin{split}
&d_a^1 x(g,h):=\frac{\lmk x^{a(g)}(h)\rmk \cdot x(g)}{x(gh)},\\
&d_a^2 y(g,h,k):=\frac{\lmk y^{a(g)}(h,k)\rmk\cdot y(g,hk)} {y(gh,k)\cdot y(g,h)},\\
&d_a^3 z(g,h,k,f):=\frac{\lmk  \lmk z^{a(g)}(h,k,f)\rmk\rmk \cdot z(g, hk, f)\cdot z(g,h,k)}{z(gh,k, f)\cdot z(g,h,kf)}.
\end{split}
\end{align}
We use the following Lemma below. The proof is straightforward from the definition.
\begin{lem}\label{lem50}
For any $m,x\in C^1(G,\bbZ_2\oplus \bbZ_2)$, $a\in H^1(G,\bbZ)$,
setting $\tilde \sigma(g,h):=(-1)^{x(g)\cdot m^{a(g)}(h)}$,
we have
\begin{align}
(-1)^{d^1_a x(g,h)\cdot m^{a(gh)}(k) +x(g)\cdot d_a^1m^{a(g)(h,k)}}
=d^2_a\tilde \sigma(g,h,k).
\end{align}
\end{lem}

We denote by $\widetilde{\PD}$ the pentad $(c, \kappa_R,\kappa_L,b,a)$
of 
\begin{align}
c\in C^3\lmk G, \Uo\oplus \Uo\rmk, \;
\kappa_R\in C^2\lmk G, \bbZ_2\oplus \bbZ_2\rmk,\;
\kappa_L\in C^2\lmk G, \bbZ_2\oplus \bbZ_2\rmk,\;
b\in C^1\lmk G, \bbZ_2\oplus \bbZ_2\rmk,\;
a\in H^1(G,\bbZ_2)
\end{align}
such that
\begin{align}
d_a^1 b(g,h)
=\kappa_L(g,h)+\kappa_R(g,h), \label{b49}\\
d_a^2 \kappa_R(g,h,k)
=0,\label{kr49}\\
d_a^2 \kappa_L(g,h,k)
=0,\label{kl49}\\
d_a^3 c(g,h,k,f)
=(-1)^{\kappa_L(g,h)\cdot \lmk \kappa_R^{a(gh)}(k,f)\rmk}.\label{c49}
\end{align}
\begin{lem}
On $\widetilde\PD$,
set 
\[
(c^{(1)}, \kappa_R^{(1)},\kappa_L^{(1)},b^{(1)},a^{(1)})\sim_{\PD}(c^{(2)}, \kappa_R^{(2)},\kappa_L^{(2)},b^{(2)},a^{(2)})
\]
 for $(c^{(1)}, \kappa_R^{(1)},\kappa_L^{(1)},b^{(1)},a^{(1)}), (c^{(2)}, \kappa_R^{(2)},\kappa_L^{(2)},b^{(2)},a^{(2)})\in \widetilde\PD$
if
the following hold.
\begin{description}
\item[(i)]$a^{(1)}(g)=a^{(2)}(g)=: a(g)$ for any $g\in G$, and
\item[(ii)] there exist an $m\in C^1(G, \bbZ_2\oplus \bbZ_2)$
and a $\sigma\in C^2(G, \Uo\oplus \Uo)$
such that
\begin{align}
&\kss R {(2)}(g,h)=d_a^1 m(g,h)+\kss R {(1)}(g,h),\\
&\kss L {(2)}(g,h)=d_a^1 b^{(2)}(g,h)-d_a^1 b^{(1)}(g,h)-d_a^1 m(g,h)+\kss L{(1)}(g,h),\\
&c^{(2)}(g,h,k)
=\lmk-1\rmk^{\kss L {(1)}(g,h) \cdot  m^{a(gh)}(k)}\lmk -1\rmk^{\lmk b^{(2)}(g)- b^{(1)}(g)-m(g)\rmk
\cdot \lmk \kss R {(2)}\rmk^{a(g)} (h,k)} d_a^2\sigma(g,h,k) c^{(1)}(g,h,k).
\end{align}
\end{description}
Then this $\sim_{\PD}$ is an equivalence relation.
\end{lem}
\begin{defn}
We denote the equivalence classes by $\PD$.
We also denote by $[(c, \kappa_R,\kappa_L,b,a)]_{\PD}$
the equivalence class containing $(c, \kappa_R,\kappa_L,b,a)\in\widetilde \PD$.
\end{defn}
\begin{proof}We show only the transitivity.
The proof for the symmetry is analogous.
Suppose that
\begin{align*}
(c^{(1)}, \kappa_R^{(1)},\kappa_L^{(1)},b^{(1)},a^{(1)})\sim_{\PD}(c^{(2)}, \kappa_R^{(2)},\kappa_L^{(2)},b^{(2)},a^{(2)})
\end{align*}
with $m\in C^1(G, \bbZ_2\oplus \bbZ_2)$
and a $\sigma\in C^2(G, \Uo\oplus \Uo)$, and
\begin{align*}
(c^{(2)}, \kappa_R^{(2)},\kappa_L^{(2)},b^{(2)},a^{(2)})\sim_{\PD}(c^{(3)}, \kappa_R^{(3)},\kappa_L^{(3)},b^{(3)},a^{(3)})
\end{align*}
with
$l\in C^1(G, \bbZ_2\oplus \bbZ_2)$
and a $\sigma'\in C^2(G, \Uo\oplus \Uo)$.
Then we have $a^{(1)}(g)=a^{(2)}(g)=a^{(3)}(g)=:a(g)$
and
\begin{align}
&\kss R {(3)}(g,h)=d_a^1 (l+m)(g,h)+\kss R {(1)}(g,h),\\
&\kss L {(3)} (g,h)
=d^1_a\lmk b^{(3)}-b^{(1)}-(l+m)\rmk(g,h)
+\kss L {(1)} (g,h).
\end{align}
Note that 
\begin{align}
\begin{split}
&\lmk-1\rmk^{\kss L {(1)}(g,h) \cdot  m^{a(gh)}(k)}\lmk -1\rmk^{\lmk b^{(2)}(g)- b^{(1)}(g)-m(g)\rmk
\cdot \lmk \kss R {(2)}\rmk^{a(g)} (h,k)}\\
&\lmk-1\rmk^{\kss L {(2)}(g,h) \cdot  l^{a(gh)}(k)}\lmk -1\rmk^{\lmk b^{(3)}(g)- b^{(2)}(g)-l(g)\rmk
\cdot \lmk \kss R {(3)}\rmk^{a(g)} (h,k)}\\
&=
\lmk-1\rmk^{\kss L {(1)}(g,h) \cdot  (l+m)^{a(gh)}(k)}\lmk -1\rmk^{\lmk b^{(3)}(g)- b^{(1)}(g)-(l+m)(g)\rmk
\cdot \lmk \kss R {(3)}\rmk^{a(g)} (h,k)}\\
&(-1)^{d^1_a(b^{(2)}-b^{(1)}-m)(g,h)\cdot  l^{a(gh)}(k)+ (b^{(2)}-b^{(1)}-m)(g)\cdot d_a^1 l^{a(g)}(h,k)}\\
&=\lmk-1\rmk^{\kss L {(1)}(g,h) \cdot (l+m)^{a(gh)}(k)}\lmk -1\rmk^{\lmk b^{(3)}(g)- b^{(1)}(g)-(l+m)(g)\rmk
\cdot \lmk \kss R {(3)} \rmk^{a(g)}(h,k)}\\
 &d^2_a\sigma''(g,h,k)
\end{split}
\end{align}
with $\sigma''(g,h):= (-1)^{(b^{(2)}-b^{(1)}-m)(g)\cdot l^{a(g)} (h)}$.
Here we used Lemma \ref{lem50} in the last equation.
Setting $\tilde\sigma := \sigma\sigma' \sigma''$,
we get 
\begin{align*}
(c^{(1)}, \kappa_R^{(1)},\kappa_L^{(1)},b^{(1)},a^{(1)})\sim_{\PD}(c^{(3)}, \kappa_R^{(3)},\kappa_L^{(3)},b^{(3)},a^{(3)})
\end{align*}
with $l+m$ and $\tilde\sigma$.
\end{proof}
\begin{lem}\label{lem52}
For any $(c,\kappa_L,\kappa_R, b,a)\in \widetilde \PD$,
set 
\begin{align}
\tilde c(g,h,k):=(-1)^{b(g)\cdot \kappa_R^{a(g)}(h,k)} c(g,h,k).
\end{align}
Then we have $(\tilde c,\kappa_R,\kappa_R, 0,a)\in \widetilde\PD$
and $(c,\kappa_L,\kappa_R, b,a)\sim_{\PD} (\tilde c,\kappa_R,\kappa_R, 0,a)$.
\end{lem}
\begin{proof}
Setting
\begin{align}
x(g,h,k):=(-1)^{b(g)\cdot \kappa_R^{a(g)}(h,k)},
\end{align}
one can check 
\begin{align}
d_a^3 x(g,h,k,f)=
(-1)^{d^1_a b(g,h)\cdot \kappa_R^{a(gh)} (k,f)},
\end{align}
using the fact that $d^2_a\kappa_R=0$.
From this, we have
\begin{align}
d^3_a\tilde c(g,h,k,f)=d_a^3 x(g,h,k,f)d_a^3 c(g,h,k,f)
=(-1)^{\lmk d^1_a b(g,h)+\kappa_L(g,h)\rmk\cdot \kappa_R^{a(gh)} (k,f)}
=(-1)^{\kappa_R(g,h)\cdot \kappa_R^{a(gh)}(k,f)}.
\end{align}
Combining this and
\begin{align}
d^1_a 0(g,h)=0=\kappa_R(g,h)+\kappa_R(g,h),\quad
d^2_a\kappa_R(g,h,k)=0
\end{align}
in $\bbZ_2$, we obtain $(\tilde c,\kappa_R,\kappa_R, 0,a)\in\widetilde \PD$.
Checking $(c,\kappa_L,\kappa_R, b,a)\sim_{\PD} (\tilde c,\kappa_R,\kappa_R, 0,a)$. is immediate.
\end{proof}
Hence to consider $\PD$, it suffices to think of the following.
\begin{defn}
We set
\begin{align}
\widetilde \PDz=
\left\{
(c,\kappa,a)
\mid (c,\kappa,\kappa,0,a)\in\widetilde \PD
\right\}.
\end{align}
We introduce the equivalence relation
$\sim_{\PDz}$ on $\widetilde \PDz$
by
\begin{align}
(c^{(1)},\kappa^{(1)},a^{(1)})\sim_{\PDz}(c^{(2)},\kappa^{(2)},a^{(2)})\quad\text{if}\quad
(c^{(1)},\kappa^{(1)},\kappa^{(1)},0,a^{(1)})\sim_{\PD}(c^{(2)},\kappa^{(2)},\kappa^{(2)},0,a^{(2)}).
\end{align}
We denote the equivalence classes by $\PDz$.
We also denote by $[(c, \kappa ,a)]_{\PDz}$
the equivalence class containing $(c, \kappa, a)\in\widetilde \PDz$.
\end{defn}
The main theorem of this paper is the following.
\begin{thm}\label{mainthm}
There is a $\PDz$-valued index on $\caP_{SL\beta}$, which is an invariant of the classification
 $\sim_\beta$ of $\caP_{SL\beta}$.
\end{thm}
The proof follows the strategy of \cite{2dSPT} for the quantum spin system case.(See reviews and videos in \cite{IAMP}\cite{CDM}\cite{ICM}.)
We consider the action of $\beta_g^U$
on our ground state $\omega_\Phi$ of $\Phi$ in SPT-phase i.e.,
$\omega_\Phi\circ\beta_g^U$.
Due to the fact that $\Phi$ is in the SPT-phase,
$\omega_{\Phi}$ can be written as $\omega_{\Phi}=\omega^{(0)}\alpha$
with some quasi-local automorphism $\alpha$, which does not create
the long range entanglement.
From this and that $\omega_{\Phi}$ is $\beta_g$-invariant,
we see that the effective excitation caused by $\beta_g^U$
on $\omega_{\Phi}$ is localized around the $x$-axis.
In particular, $\omega^{(0)}\circ\alpha\beta_g^U\alpha^{-1}$
satisfies the split property with respect to the cut $H_L-H_R$.
Namely, it is quasi-equivalent to a state of the form $\varphi_L\hat\otimes \varphi_R$,
with homogeneous states $\varphi_L,\varphi_R$ on $\caA_{H_L}$
$\caA_{H_R}$.
The difference from the quantum spin case \cite{2dSPT}
and our Fermionic case is that in Fermionic case when a state satisfies the split property,
there are two possibilities.
If the restriction of $\omega^{(0)}\circ\alpha\beta_g^U\alpha^{-1}$ to $\caA_{H_R}$ is a factor state,
then it is equivalent to the state of the form $\omega^{(0)}\circ\lmk \eta_{gL}\hat\otimes\eta_{gR}\rmk$,
with some graded automorphisms $\eta_{gL}$, $\eta_{gR}$
 on $\caA_{H_L}$, $\caA_{H_R}$ localized around $x$-axis.
If the restriction is not a factor state, then it is equivalent to the state of the form
$\omega^{(1)}\circ\lmk \eta_{gL}\hat\otimes\eta_{gR}\rmk$, with some other
reference state $\omega^{(1)}$.
This state $\omega^{(1)}$ is of the form $\omega^{(0)}\circ\tau$,
with some automorphism $\tau$ on the $x$-axis representing the space translation on the $x$-axis.
This dichotomy gives us the index taking value in $H^1(G,\bbZ_2)$.
Some combination of automorphisms $\eta_{gL}$ and $\eta_{gR}$ is implementable by a unitary in the GNS representation of $\omega_{\Phi}$
and it
give us some indices in $C^3(G,\Uo\oplus \Uo)$(see Lemma \ref{lem39})
, just like in the quantum spin case. The $C^2(G,\Uo\oplus \Uo)$ part represents 
if this unitary is even or odd.
Of course they are not independent to each other that
result in the complication of our index in $\PDz$.
Without the doubled structure like $\Uo\oplus \Uo$,
result coincides with the one predicted in \cite{BM} \cite{WG}.
For the moment, we do not know how to remove it. It comes from the fact that 
$\omega\circ\tau$ and $\omega\circ\tau^{-1}$
are not equivalent in general.
It might be possible to remove it under some symmetry.

The rest of the paper is organized as follows.
 In the analysis, the automorphic equivalence gets important.
 It is explained in section \ref{autosec}.
 The $\eta_{g,L}, \eta_{g,R}$ above are given from the classification of
 pure states satisfying the split property, analogous to that of \cite{GS}.
 Its Fermionic version is given in section \ref{splitsec}.
The $\PDz$-valued index is derived in section \ref{indexsec}.
In section \ref{stabilitysec}, we show it is actually the invariant of our classification.
Section \ref{sdautosec} gives some convenient property of $\sdckc{}$ which we use in the analysis.
Notations from Appendix \ref{notasec} are used freely.
Some facts about graded von Neumann algebras are collected/ proven in Appendix \ref{gravn}.

\section{Automorphic Equivalence}\label{autosec}
A very important fact about gapped ground state phases
is the automorphic equivalence \cite{bmns}~\cite{NSY},
~\cite{MO},
which started as Hastings adiabatic Lemma \cite{HW}.

First we introduce a class of paths of interactions.
A norm-continuous interaction on $\caA$ defined on an interval $[0,1]$
is a map
$\Phi:{\mathfrak S}_{\bbZ^2}\times [0,1]\to \caA$ such that
\begin{description}
\item[(i)]
for any $t\in[0,1]$, $\Phi(\cdot, t):{\mathfrak S}_{\bbZ^2}\to \caA$
is an even interaction, and
\item[(ii)]
for any $Z\in{\mathfrak S}_{\bbZ^2}$, the map
$\Phi(Z,\cdot ):[0,1]\to \caA_{Z}$
is piecewise norm-continuous.
\end{description}
Let $F$ be an $F$-function on $({\bbZ^2},\dist)$.
(See \cite{NSY} or \cite{2dSPT} Appendix C for the definition of $F$-functions.)
We denote by $\hat \caB_{F}([0,1])$ the set of all
norm continuous interactions $\Phi$ on $\caA$ defined on an interval $[0,1]$ such that
\begin{align}
\lV\lv \Phi\rV\rv_F:=
\sup_{x,y\in{\bbZ^2}}\frac{1}{F\lmk {\dist}(x,y)\rmk}\sum_{Z\in{\mathfrak S}_{\bbZ^2}, Z\ni x,y}
\sup_{t\in [0,1]}\lmk \lV\Phi(Z;t)\rV\rmk<\infty.
\end{align}
From $\Psi\in \hat\caB_F([0,1])$, we can construct
a path of automorphisms $\tau^\Psi_{t,s}$ 
 following
the same argument as quantum spin case
~\cite{BSP, NSY17,1dFermi}.
We can consider analogous path of interactions on
$\Gamma\subset \bbZ^2$, which we denote by $\hat\caB_{F,\Gamma}([0,1])$

For a subset $\Gamma\subset \bbZ^2$, we set
\begin{align}
\begin{split}
&\QAut(\caA_\Gamma):=
\left\{
\alpha\mid
\alpha=\tau_{s,t}^\Psi,\;
\text{for some}\;
\Psi\in \hat\caB_{F,\Gamma}([0,1]),\; F: F\text{-function}\; s,t\in[0,1]
\right\},\\
&\QAut_\beta(\caA_\Gamma):=
\left\{
\alpha\mid
\alpha=\tau_{s,t}^\Psi,\;
\text{for some}\;\beta\text{-invariant}
\Psi\in \hat\caB_{F,\Gamma}([0,1]),\; \; F: F\text{-function}\; s,t\in[0,1]
\right\}.
\end{split}
\end{align}
They form subgroups of $\auz(\caA_\Gamma)$.
Note from the proof of Theorem 4.7 \cite{aramori}
that for the conditional expectation $\Pi_\Gamma$, we have
$\Pi_\Gamma(a\hat\otimes b)=0$
for any $a\in \caA_\Gamma$ and $b\in\caA_{\Gamma^c}^{(1)}$.
Using this fact, just by following the argument in \cite{MO},
we can show the following Fermionic version of automorphic equivalence.
\begin{thm}\label{mo}
Let $\Phi_0,\Phi_1\in\caP_{UG}$ and $\omega_{\Phi_0}$, $\omega_{\Phi_1}$
be their unique gapped ground states.
 Suppose that
$\Phi_0\sim\Phi_1$ holds, via a path $\Phi : [0,1]\to \caP_{UG}$.
Then there exists some 
$\alpha\in\QAut(\caA)$ such that
$\omega_{\Phi_{1}}=\omega_{\Phi_0}\circ\alpha$.
If 
$\Phi_0,\Phi_1\in \caP_{UG\beta}$ and 
$\Phi\sim_\beta\Phi_0$,
we may take $\alpha$ from
$\QAut_\beta(\caA)$.
\end{thm}

Because of the theorem, $\QAut(\caA)$ plays an important role for us, and the ground states
we consider belong to the following set.
\begin{align}
\SPT:=
\left\{
\omega^{(0)}\circ\alpha\mid 
\alpha\in \QAut(\caA),\quad \omega^{(0)}\circ\alpha\circ\beta_g
=\omega^{(0)}\circ\alpha.
\right\}.
\end{align}
For each $\omega\in \SPT$, by definition, the set
\begin{align}
\EAut(\omega)
=
\left\{
\alpha\in \QAut(\caA)\mid
\omega=\omega^{(0)}\circ\alpha
\right\}
\end{align}
is non-empty.

From the fact they are given by local interactions, automorphisms in $\QAut(\caA)$
satisfy nice properties. We list up such properties for the rest of this section.
Most of them can be proven using the same argument as that
of quantum spin case \cite{2dSPT}
combined with the property of conditional expectations $\Pi_\Gamma$ 
mentioned above, and we omit the proof.

For $0<\theta<\frac{\pi}{2}$, we set
\begin{align}
C_{\theta}:=\left\{
(x,y)\mid |y|\le \tan\theta\cdot |x|
\right\}.
\end{align}
For $0<\theta_1<\theta_2\le \frac \pi 2$, we use a notation
$\caC_{(\theta_1,\theta_2]}:=C_{\theta_2}\setminus C_{\theta_1}$ and
$\caC_{[0,\theta_1]}:=C_{\theta_1}$.

We also set
\begin{align}
c_L:=-5,\quad c_R:=5.
\end{align}

An automorphism $\alpha\in \QAut(\caA)$ satisfies a factorization property.
It basically says that we can split $\alpha$ into two along any cut of the system
modulo some error terms localized around the boundary.
For example, if we cut the system along the $y$-axis, we obtain the following:
for any $0<\theta<\frac \pi 2$,
there are 
\begin{align}
\alpha_L\in \QAut\lmk \caA_{H_L}\rmk,\quad \alpha_R\in \QAut\lmk \caA_{H_R}\rmk,\quad
\Upsilon\in\QAut\lmk\caA_{\lmk C_\theta\rmk^c}\rmk
\end{align}
decomposing $\alpha$ as
\begin{align}\label{afactoc}
\alpha=\inn\circ\lmk\alpha_L\hat \otimes\alpha_R\rmk\circ\Upsilon.
\end{align}
If we cut the system along the $x$-axis, we have the following:
there are 
\begin{align}\label{xfac}
\alpha_U\in \QAut(\caA_{H_U^1}),
\quad \alpha_D\in \QAut(\caA_{H_D^{-1}}),\quad
\Xi_L\in \QAut\lmk\caA_{C_\theta\cap H_L}\rmk,\quad
\Xi_R\in \QAut\lmk\caA_{C_\theta\cap H_R}\rmk
\end{align}
such that 
\begin{align}
\alpha
=\lmk \alpha_U\hat\otimes \alpha_D\rmk
\circ\lmk \Xi_{L}\hat\otimes \Xi_R\rmk\inn.
\end{align}
Hence for $\alpha\in \QAut(\caA)$,
and $0<\theta<\frac\pi 2$, the following sets are non-empty.
\begin{align}
\begin{split}
&\caD^V(\alpha,\theta)
:=
\left\{
(\alpha_L,\alpha_R,\Upsilon)\middle|
\begin{gathered}
\alpha_L\in \QAut(\caA_{H_L}),\quad \alpha_R\in \QAut(\caA_{H_R}),\quad
\Upsilon\in \QAut\lmk\caA_{C_\theta^c}\rmk\\
\text{such that} \quad \alpha
=\lmk \alpha_L\hat\otimes \alpha_R\rmk\Upsilon\circ\inn
\end{gathered}
\right\},\\
&\caD^H(\alpha,\theta)
:=
\left\{
(\alpha_U,\alpha_D,\Xi_L,\Xi_R)\middle|
\begin{gathered}
\alpha_U\in \QAut(\caA_{H_U^1}),
\quad \alpha_D\in \QAut(\caA_{H_D^{-1}}),\\
\Xi_L\in \QAut\lmk\caA_{C_\theta\cap H_L}\rmk
\quad
\Xi_R\in \QAut\lmk\caA_{C_\theta\cap H_R}\rmk\\
\text{such that} \quad \alpha
=\lmk \alpha_U\hat\otimes \alpha_D\rmk
\circ\lmk \Xi_{L}\hat\otimes \Xi_R\rmk\inn
\end{gathered}
\right\}.
\end{split}
\end{align}

We can consider finer factorization:
for each\begin{align}\label{thetas1}
0<\theta_{0.8}<\theta_1<\theta_{1.2}<\theta_{1.8}<\theta_2<\theta_{2.2}<
\theta_{2.8}<\theta_3<\theta_{3.2}<\frac\pi 2,
\end{align}
 $\alpha\in \QAut(\caA)$ can be decomposed as \begin{align}\label{sqaut}
&\alpha=\inn\circ\lmk
\alpha_{[0,\theta_1]}\otimes\alpha_{(\theta_1,\theta_2]}
\otimes \alpha_{(\theta_2,\theta_3]}\otimes
\alpha_{(\theta_3,\frac\pi 2]}
\rmk
\circ
\lmk
\alpha_{(\theta_{0.8}, \theta_{1.2}]}\otimes
\alpha_{(\theta_{1.8},\theta_{2.2}]}
\otimes \alpha_{(\theta_{2.8},\theta_{3.2}]}
\rmk
\end{align}
with  \begin{align}\label{sqaut2}
  \begin{split}
&  \alpha_X:=\widehat\bigotimes_{\sigma=L,R,\zeta=D,U} \alpha_{X,\sigma,\zeta},\quad
 \alpha_{[0,\theta_1]}:=\widehat \bigotimes_{\sigma=L,R}\alpha_{[0,\theta_{1}],\sigma},\quad
 \alpha_{(\theta_3,\frac\pi 2]}:=\widehat
 \bigotimes_{\zeta=D,U}  \alpha_{(\theta_3,\frac\pi 2],\zeta}\\
 &\alpha_{X,\sigma,\zeta}\in 
 \QAut\lmk\caA_{C_{X}\cap H_\sigma^{c_\sigma}\cap H_\zeta}\rmk,\quad
 \alpha_{X,\sigma}:=\widehat \bigotimes_{\zeta=U,D}\alpha_{X,\sigma,\zeta},\quad
\alpha_{X,\zeta}:=\widehat \bigotimes_{\sigma=L,R}\alpha_{X,\sigma,\zeta}\\
&\alpha_{[0,\theta_{1}],\sigma}\in 
\QAut\lmk\caA_{{C_{[0,\theta_{1}]}}\cap H_\sigma^{c_\sigma}}\rmk,\quad
 \alpha_{(\theta_3,\frac\pi 2],\zeta}\in \QAut\lmk\caA_{C_{(\theta_3,\frac\pi 2]}\cap H_\zeta}\rmk, 
  \end{split} 
  \end{align}
 for
 \begin{align}\label{sqaut3}
 X=(\theta_1,\theta_2], (\theta_2,\theta_3],
 (\theta_{0.8},\theta_{1.2}],
 (\theta_{1.8},\theta_{2.2}], 
(\theta_{2.8},\theta_{3.2}],\quad \sigma=L,R,\quad \zeta=D,U.
 \end{align}
In order to define our index, we introduce an automorphism localized along $x$-axis.
Let $v_\tau$ be a unitary on $\mkh=l^2(\bbZ^2)$
such that 
\begin{align}
v_\tau\delta_{(x,y)}
:=\left\{
\begin{gathered}
\delta_{(x,y)},\quad y\neq 0\\
\delta_{(x+1,0)},\quad y=0
\end{gathered}.\right.
\end{align}
Note that
$v_\tau$ commutes with the complex conjugation
$\mkC$.
Therefore, it defines an automorphism $\tau:=\Xi_{v_\tau}$
on $\al$
such that
\begin{align}
\tau\lmk B(f)\rmk:=B\lmk v_\tau f\rmk,
\quad f\in \mkh.
\end{align}

Decompositions like above allow us to derive the following Lemma immediately.
\begin{lem}\label{lem29}
For any $\alpha\in \QAut\lmk \caA_{\bbZ^2}\rmk$, the following hold.
\begin{description}
\item[(i)]
For any $0<\varphi<\frac\pi 2$ and $a\in \bbZ$, 
there are automorphisms $\zeta_{\sigma}\in \auz\lmk\caA_{C_\varphi\cap H_\sigma^{c_\sigma}}\rmk$,
$\sigma=L,R$
such that
\begin{align}
\alpha\tau^a\alpha^{-1}=\tau^a\lmk \zeta_L\hat\otimes \zeta_R\rmk\circ\inn.
\end{align}
\item[(ii)]
For any $0<\varphi'<\varphi''<\frac\pi 2$
and $X_\sigma\in\auz\lmk\caA_{C_{\varphi'}\cap H_\sigma^{c_\sigma}}\rmk$, $\sigma=L,R$
there are automorphisms
$\tilde X_\sigma\in\auz\lmk\caA_{C_{\varphi''}\cap H_\sigma^{c_\sigma}}\rmk$, $\sigma=L,R$
 such that
\begin{align}
\alpha\lmk X_L\hat\otimes X_R\rmk\alpha^{-1}
=\lmk \tilde X_L\hat\otimes \tilde X_R\rmk\circ\inn.
\end{align}
\end{description}
\end{lem}
\begin{defn}\label{lem30}
For any $\Lambda\subset\bbZ^2$ and $X\in\auz\lmk\tilde \caA_{\Lambda}\rmk$ and
$a\in \bbZ$,
we denote by
$X_a\in \auz\lmk
\tilde \caA_{(\Lambda_0+a\bm e_0)\cup \Lambda_1}
\rmk$
the automorphism
such that 
$\tau_a X\tau_a^{-1}= X_a$.
Here we decomposed $\Lambda$ as
\begin{align}
\Lambda=\Lambda_0\dot\cup \Lambda_1
:=\lmk \Lambda\cap \lmk \bbZ\times \{0\}\rmk\rmk
\dot\cup \lmk \Lambda\cap \lmk\bbZ^2\setminus \lmk \bbZ\times\{0\}\rmk\rmk\rmk,
\end{align}
and set $\bm e_0:=(1,0)\in\bbZ^2$.
\end{defn}
From Lemma \ref{lem29} (i)
and $\omega^{(0)}\circ \tau^{2a}=\omega^{(0)}$,
we obtain the following:
\begin{lem}\label{lem38}
For any $\alpha\in\QAut(\caA_{\bbZ^2})$ and $0<\varphi<\frac\pi 2$, $a\in\bbZ$
there are automorphisms $\xi_{\sigma}\in \auz\lmk\caA_{C_\varphi\cap H_\sigma^{c_\sigma}}\rmk$, $\sigma=L,R$
such that
\begin{align}
\omega^{(0)}\alpha\circ\tau^{2a}\simeq \omega^{(0)}\alpha\lmk \xi_L\hat \otimes \xi_R\rmk.
\end{align}
\end{lem}

\section{Split property of pure states on Self-dual-CAR-algebras}\label{splitsec}
Having the automorphic equivalence
$\omega_{\Phi}=\omega^{(0)}\alpha$ with $\alpha\in \QAut(\caA)$ and
the $\beta_g$-invariance of $\omega_{\Phi}$,
we expect that the effective excitation caused by $\beta_g^U$
on $\omega_{\Phi}$ is localized around the $x$-axis.
It can be shown so, by observing that 
 $\omega^{(0)}\circ\alpha\beta_g^U\alpha^{-1}$
satisfies the split property (Definition \ref{splitdef}) with respect to the $H_L-H_R$ cut.
In this section, as a preparation of our analysis,
we investigate the split property of pure states on self-dual CAR-algebras.
%
%
%
%
For Fermionic systems, the split property gives some dichotomy (Lemma \ref{lem16}).
The main proposition of this section is Proposition \ref{mainprop},
which states that two pure split states belonging to the same category of the dichotomy
can be connected by automorphisms of graded tensor product form.
We use definitions and facts from Appendix \ref{notasec} and Appendix \ref{gravn} freely.

We start by some basic fact we repeatedly use.
\begin{lem}\label{lem21}
Let $\mkh_1$, $\mkh_2$ be Hilbert spaces
with complex conjugation $\mkC_1$, $\mkC_2$, respectively.
Let $\varphi,\omega$  be
homogeneous pure states on
$\SDC \lmk \mkh_1\oplus \mkh_2, \mkC_1\oplus \mkC_2\rmk$.
Suppose that $\varphi$ and $\omega$ are equivalent.
Then their restrictions $\varphi\vert_{\SDC\lmk \mkh_1,\mkC_1\rmk}$,
$\omega\vert_{\SDC\lmk \mkh_1,\mkC_1\rmk}$
onto $\SDC\lmk \mkh_1,\mkC_1\rmk$
are quasi-equivalent.
\end{lem}
\begin{proof}
We write $\mkA_i:=\SDC\lmk \mkh_i,\mkC_i\rmk$, for $i=1,2$.
Let $(\caH_\varphi,\pi_\varphi,\Omega_\varphi)$
be a GNS representation of $\varphi$.
Because $\varphi$ is homogeneous,
there is a self-adjoint unitary $\Gamma_\varphi$
on $\caH_\varphi$ such that
$\Gamma_\varphi\pi_\varphi(A)\Omega_\varphi
=\pi_\varphi\circ\Theta(A)\Omega_\varphi$,
$A\in\SDC\lmk \mkh_1\oplus \mkh_2, \mkC_1\oplus \mkC_2\rmk$,
for the grading operator $\Theta$ on $\SDC\lmk \mkh_1\oplus \mkh_2, \mkC_1\oplus \mkC_2\rmk$.
Let $p_\varphi$ be an orthogonal projection onto the subspace
$\caK_\varphi:=\overline{\pi_\varphi\lmk\mkA_1\rmk\Omega_\varphi}$.
Because $p_\varphi\in \pi_\varphi(\mkA_1)'$,
$\rho_\varphi(A):=\pi_\varphi\lmk A\rmk p_\varphi$, $A\in\mkA_1$
defines a representation of $\mkA_1$
on $\caK_\varphi$.
Then $(\caK_\varphi, \rho_\varphi,\Omega_\varphi)$
is a GNS representation of $\varphi\vert_{\mkA_1}$.
By the Kaplansky density theorem,
we have $\rho_\varphi(\mkA_1)''=\pi_\varphi\lmk \mkA_1\rmk'' p_\varphi$.
Let us consider the map 
$\tau_\varphi : \pi_\varphi\lmk\mkA_1\rmk''\to \rho_\varphi(\mkA_1)''
$,
defined by $\tau_\varphi(x):=x p_\varphi$, for $x\in \pi_\varphi\lmk\mkA_1\rmk''$.
It is a $*$-homomorphism onto $\rho_\varphi(\mkA_1)''$.
We claim that $\tau_\varphi$ is injective.
To see this, we note $p_\varphi$ and $\Gamma_\varphi$ 
commute because $\caK_\varphi$ is $\Gamma_\varphi$-invariant.
As a result, we see that 
$\ker\tau_\varphi=\Ad\lmk\Gamma_\varphi\rmk\lmk \ker\tau_\varphi\rmk$.
Let $E$ be the central projection
of $\pi_\varphi\lmk\mkA_1\rmk''$ such that 
$\ker\tau_\varphi=\pi_\varphi\lmk\mkA_1\rmk'' E$.
From the above observation, we see that
\begin{align}
\pi_\varphi\lmk\mkA_1\rmk''\Ad\lmk\Gamma_\varphi\rmk(E)
=\Ad\lmk\Gamma_\varphi\rmk \lmk \pi_\varphi\lmk\mkA_1\rmk''\rmk
\Ad\lmk\Gamma_\varphi\rmk(E)
=\Ad\lmk\Gamma_\varphi\rmk\lmk\ker\tau_\varphi\rmk
=\ker\tau_\varphi=\pi_\varphi\lmk\mkA_1\rmk'' E.
\end{align}
From this, we have
$\Ad\lmk\Gamma_\varphi\rmk(E)
=E$.
Hence we have 
\begin{align}
E\in Z\lmk \pi_\varphi\lmk\mkA_1\rmk''\rmk
\cap \lmk \pi_\varphi\lmk\mkA_1^{(0)}\rmk''\rmk
\subset \pi_\varphi\lmk \mkA_1\rmk'\cap \pi_\varphi\lmk\mkA_2\rmk'
=\pi_\varphi\lmk\mkA\lmk \mkh_1\oplus \mkh_2, \mkC_1\oplus \mkC_2\rmk\rmk'=\bbC\unit,
\end{align}
because $\varphi$ is pure.
Then we conclude $E=0$ and
this proves the claim, and our $\tau_\varphi$ is a $*$-isomorphism.

Hence for GNS representations $\pi_\varphi$,
$\rho_\varphi$ of $\varphi$, $\varphi\vert_{\mkA_1}$ respectively,
threre is a $*$-isomorphism $\tau_\varphi: \pi_\varphi\lmk\mkA_1\rmk''
\to \rho_\varphi\lmk\mkA_1\rmk''$
such that $\tau_\varphi\circ\pi_\varphi(A)=\rho_\varphi(A)$,
for each $A\in\mkA_1$.
Similarly,  for GNS representations $\pi_\omega$,
$\rho_\omega$ of $\omega$, $\omega\vert_{\mkA_1}$ respectively,
threre is a $*$-isomorphism $\tau_\omega: \pi_\omega\lmk\mkA_1\rmk''
\to \rho_\omega\lmk\mkA_1\rmk''$
such that $\tau_\omega\circ\pi_\omega(A)=\rho_\omega(A)$,
for each $A\in\mkA_1$.
Because $\varphi$ and $\omega$ are equivalent,
there is a $*$-isomorphism $\tau : \pi_\varphi\lmk \mkA\lmk \mkh_1\oplus \mkh_2, \mkC_1\oplus \mkC_2\rmk\rmk''
\to\pi_\omega\lmk \mkA\lmk \mkh_1\oplus \mkh_2, \mkC_1\oplus \mkC_2\rmk\rmk''$
such that $\tau\circ\pi_\varphi(A)=\pi_\omega(A)$ for
$A\in \mkA\lmk \mkh_1\oplus \mkh_2, \mkC_1\oplus \mkC_2\rmk$.
Restricting this $\tau$ to $\pi_\varphi\lmk\mkA_1\rmk''$,
we obtain
a $*$-isomorphism $\tau_1: \pi_\varphi\lmk\mkA_1\rmk''\to \pi_\omega\lmk\mkA_1\rmk''$
such that $\tau_1\circ\pi_\varphi(A)=\pi_\omega(A)$
for $A\in \mkA_1$.
Then we see that $\tau_\omega\circ\tau_1\circ\tau_{\varphi}^{-1}$
defines a $*$-isomorphism from $\rho_\varphi(\mkA_1)''$
onto $\rho_{\omega}(\mkA_1)''$
such that $\tau_\omega\circ\tau_1\circ\tau_{\varphi}^{-1}\circ
\rho_\varphi(A)=\rho_\omega(A)$, $A\in \mkA_1$.

\end{proof}
We encounter the following situation as well.
\begin{lem}\label{lem26}
Let $\mkk_1$, $\mkk_2$ be Hilbert spaces
with complex conjugation $\mkC_1$, $\mkC_2$, respectively.
Let $\omega$  be
a homogeneous pure states on
$\SDC\lmk \mkk_1\oplus \mkk_2, \mkC_1\oplus \mkC_2\rmk$,
and $\varphi_1,\varphi_2$ homogeneous states on
$\sdckc{1}$, $\sdckc{2}$ respectively.
If $\omega\qe \varphi_1\hat\otimes \varphi_2$,
then we have $\left. \omega\right\vert_{\sdckc{1}}\qe \varphi_1$
and $\left. \omega\right\vert_{\sdckc{2}}\qe \varphi_2$.
\end{lem}
\begin{proof}
By the proof of Lemma \ref{lem21}, 
we have $\left. \pi_\omega\right\vert_{\sdckc{2}}\qe \pi_{\left.\omega\right\vert_{\sdckc{2}}}$.
(Note that $\left. \pi_\omega\right\vert_{\sdckc{2}}$ is the restriction of the GNS representation of $\omega$
while $\pi_{\left.\omega\right\vert_{\sdckc{2}}}$ is the GNS representation of the restriction of $\omega$.)
Because of $\omega\qe \varphi_1\hat\otimes \varphi_2$,
we have $\left. \pi_\omega\right\vert_{\sdckc{2}}\qe \pi_{\varphi_2}$.
Hence we have $ \pi_{\left.\omega\right\vert_{\sdckc{2}}}\qe \pi_{\varphi_2}$.
The same for $\sdckc{1}$.
This proves the Lemma.
\end{proof}

The situation in the previous Lemma has a name.
\begin{defn}\label{splitdef}
Let $\mkk_i$ be an infinite dimensional Hilbert space
with a complex conjugation $\mkC_i$, for $i=1,2$.
We say a homogeneous pure state $\omega$ on 
$\sdc\lmk\mkk_1\oplus \mkk_2,\mkC_1\oplus \mkC_2\rmk$
satisfies the split property if
there are homogeneous states $\varphi_i$
on $\sdc\lmk \mkk_i, \mkC_i\rmk$, $i=1,2$
such that
$\omega\qe \varphi_1\hat\otimes \varphi_2$.
\end{defn}
The following is a refinement of the dichotomy introduced in \cite{Matsui20}.
\begin{lem}\label{lem16}
Let $\mkk_{\sigman}$ be an infinite dimensional Hilbert space
with a complex conjugation $\mkC_{\sigman}$, for $\sigman={\Ln},{\Rn}$.
Let $\omega$ be a homogeneous pure state on 
$\sdc\lmk\mkk_{\Ln}\oplus \mkk_{\Rn},\mkC_{\Ln}\oplus \mkC_{\Rn}\rmk$,
satisfying the split property.
Let $\Theta_{\Ln},\Theta_{\Rn}$ be 
automorphisms on $\sdc^{(0)}\lmk\mkk_{\Ln}\oplus \mkk_{\Rn},\mkC_{\Ln}\oplus \mkC_{\Rn}\rmk$
such that $\Theta_{\Ln}\lmk B(f_{\Ln}\oplus f_{\Rn})\rmk=B\lmk -f_{\Ln}\oplus f_{\Rn}\rmk$,
$\Theta_{\Rn}\lmk B(f_{\Ln}\oplus f_{\Rn})\rmk=B\lmk f_{\Ln}\oplus (-f_{\Rn})\rmk$,
for $f_\sigman\in \mkk_{\sigman}$, $i=1,2$.
Then one of the following occurs.
\begin{description}
\item[(i)]
We have $\left.\omega\right\vert_{\sdc^{(0)}\lmk\mkk_{\Ln}\oplus \mkk_{\Rn},\mkC_{\Ln}\oplus \mkC_{\Rn}\rmk}\Theta_{\Rn}
\qe \left.\omega\right\vert_{\sdc^{(0)}\lmk\mkk_{\Ln}\oplus \mkk_{\Rn},\mkC_{\Ln}\oplus \mkC_{\Rn}\rmk}$.
The state $\omega$ has a GNS representation of the form
$(\caH_{\Ln}\otimes \caH_{\Rn},\pi_{\Ln}\hat \otimes \pi_{\Rn},\Omega)$
with $(\caH_{\sigman},\pi_{\sigman})$ an irreducible representation of $\sdc(\mkk_{\sigman},\mkC_{\sigman})$ for $\sigman={\Ln},{\Rn}$.
There is a self-adjoint unitary $\Gamma_\sigman$ on $\caH_{\sigman}$ implementing $\Theta_{\sigman}$, i.e.,
$\Ad(\Gamma_{\sigman})\pi_\sigman=\pi_{\sigman}\Theta_{\sigman}$ for $\sigman={\Ln},{\Rn}$.
Decomposing $\caH_{\sigman}$ with $\caH_{\sigman\pm}:=\frac{\unit\pm \Gamma_{\sigman}}{{\Rn}}\caH_{\sigman}$, 
as $\caH_\sigman=\caH_{\sigman+}\oplus \caH_{\sigman-}$,
we have $\left.\pi\right\vert_{\sdc^{(0)}\lmk\mkk_{i},\mkC_{i}\rmk}
=\pi_{\sigman+}\oplus \pi_{\sigman-}$
with $\pi_{\sigman\pm}$ mutually singular irreducible representations on $\caH_{i\pm}$
 of $\sdc^{(0)}\lmk\mkk_i,\mkC_i\rmk$.

\item[(ii)]The states
$\left.\omega\right\vert_{\sdc^{(0)}\lmk\mkk_{\Ln}\oplus \mkk_{\Rn},\mkC_{\Ln}\oplus \mkC_{\Rn}\rmk}\Theta_{\Rn}$
and $\left.\omega\right\vert_{\sdc^{(0)}\lmk\mkk_{\Ln}\oplus \mkk_{\Rn},\mkC_{\Ln}\oplus \mkC_{\Rn}\rmk}$
are mutually singular.
The state $\omega$ has a GNS representation of the form
$(\caH_{\Ln}\otimes \caH_{\Rn}\otimes \bbC^{\Rn}, \pi,\Omega)$.
 There are irreducible representations  $\pi_{\sigman}$
 of $\sdc^{(0)}\lmk \mkk_{\sigman},\mkC_{\sigman}\rmk$, on $\caH_{\sigman}$,
 $\sigman={\Ln},{\Rn}$ such that
 \begin{align}
 \pi\lmk a\hat\otimes b\rmk
 =\pi_{\Ln}(a)\otimes\pi_{\Rn}(b)\otimes \unit_{\bbC^{\Rn}},\quad
 a\in \sdc^{(0)}\lmk \mkk_{\Ln},\mkC_{\Ln}\rmk, \quad b\in \sdc^{(0)}\lmk \mkk_{\Rn},\mkC_{\Rn}\rmk.
 \end{align}  
 We have
 \begin{align}
 \pi\lmk\sdc\lmk \mkk_{\Ln},\mkC_{\Ln}\rmk\rmk''
 =\caB(\caH_{\Ln})\otimes\bbC\unit_{\caH_{\Rn}}\otimes\lmk \bbC\sigma_z+\bbC\unit\rmk,\quad
 \pi\lmk\sdc\lmk \mkk_{\Rn},\mkC_{\Rn}\rmk\rmk''
 =\bbC\unit_{\caH_{\Ln}}\otimes \caB(\caH_{\Rn})\otimes\lmk \bbC\sigma_x+\bbC\unit\rmk.
 \end{align}
We also have $\Ad\lmk\unit_{\caH_{\Ln}}\otimes\unit_{\caH_{\Rn}}\otimes
\sigma_y\rmk\circ\pi=\pi\Theta$.
\end{description}
If $\omega$ satisfies (i) (resp. (ii)), then $\omega\circ\lmk\eta_{\Ln}\hat\otimes\eta_{\Rn}\rmk$
also satisfies (i) (resp. (ii)),
for any $\eta_{\Ln}\in \Aut^{(0)}\lmk\sdckc{{\Ln}}\rmk$, 
$\eta_{\Rn}\in \Aut^{(0)}\lmk\sdckc{{\Rn}}\rmk$.
\end{lem}
\begin{proof} 

Let $\varphi_i$
be
 homogeneous states on $\sdc\lmk \mkk_i, \mkC_i\rmk$, $i=1,2$
 such that $\omega\qe \varphi_1\hat\otimes \varphi_2$. 
 Let $\omega_i:=\omega\vert_{\SDC(\mkK_i,\mkC_i)}$, $i=1,2$,
on $\SDC(\mkK_i,\mkC_i)$ be the restriction of $\omega$, which is homogeneous.
 By Lemma \ref{lem26}, we have $\varphi_i\qe \omega_i$ and obtain
%
$\omega\qe \omega_1\hat\otimes \omega_2$. 

Because $\omega$ is a pure state, $\pi_\omega(\SDC(\mkK_i,\mkC_i))''$ $i=1,2$ are central.
(See Appendix \ref{notasec} for the notation.)
They are also balanced because self-dual algebra has a self-adjoint odd unitary.
Hence the von Neumann algebras $\pi_{\omega_i}(\SDC(\mkK_i,\mkC_i))''$ $i=1,2$
which are $*$-isomorphic to them (from the proof of Lemma \ref{lem21}) are also central and balanced.
By $\omega\qe \omega_1\hat\otimes \omega_2$,
$\pi_{\omega_1}(\SDC(\mkK_1,\mkC_1))''\hat\otimes \pi_{\omega_2}(\SDC(\mkK_2,\mkC_2))''$
is $*$-isomorphic to $\pi_\omega(\sdc^{(0)}\lmk\mkk_{\Ln}\oplus \mkk_{\Rn},\mkC_{\Ln}\oplus \mkC_{\Rn}\rmk)''$
hence is a type I factor.
From Lemma \ref{apni},
we conclude that both of $\pi_{\omega_1}(\SDC(\mkK_1,\mkC_1))''$
and $\pi_{\omega_2}(\SDC(\mkK_2,\mkC_2))''$ are type I.
From Lemma \ref{apichi}, 
they are either a type I factor or a direct sum of two type I factors.
If $\pi_{\omega_1}(\SDC(\mkK_1,\mkC_1))''$ is a factor
but not $\pi_{\omega_2}(\SDC(\mkK_2,\mkC_2))''$,
then by Lemma A.2 of \cite{1dFermi} $Z\lmk \pi_{\omega_2}(\SDC(\mkK_2,\mkC_2))''\rmk$ has a self-adjoint odd unitary $b$
while $\pi_{\omega_1}(\SDC(\mkK_1,\mkC_1))''$ includes a even self-adjoint unitary $\theta_1$
implementing the grading on  $\pi_{\omega_1}(\SDC(\mkK_1,\mkC_1))''$.
$\theta_1\hat\otimes b$ belongs to the center of a factor
$\pi_{\omega_1}(\SDC(\mkK_1,\mkC_1))''\hat\otimes \pi_{\omega_2}(\SDC(\mkK_2,\mkC_2))''$,
which contradicts to the fact that the latter algebra is a factor.
Hence if $\pi_{\omega_1}(\SDC(\mkK_1,\mkC_1))''$ is a factor, then 
$\pi_{\omega_2}(\SDC(\mkK_2,\mkC_2))''$ is a factor as well.
Similarly, if $\pi_{\omega_2}(\SDC(\mkK_2,\mkC_2))''$ is a factor
then $\pi_{\omega_1}(\SDC(\mkK_1,\mkC_1))''$ is a factor as well.
As a result, either (i) both of  $\pi_\omega(\SDC(\mkK_i,\mkC_i))''$, $i=1,2$ are type I factors
or (ii) both of  $\pi_\omega(\SDC(\mkK_i,\mkC_i))''$, $i=1,2$ are direct sum of
two type I factors.
Note from Lemma 6.23 of \cite{EK} that 
$\omega\vert_{\sdc^{(0)}\lmk\mkk_{\Ln}\oplus \mkk_{\Rn},\mkC_{\Ln}\oplus \mkC_{\Rn}\rmk}$
is pure.
\\
(i)If both of $\pi_\omega(\SDC(\mkK_i,\mkC_i))''$, $i=1,2$ are type I factors, then from Lemma 5.5 of \cite{1dFermi} and its proof, 
the state $\omega$ has a GNS representation of the form
$(\caH_{\Ln}\otimes \caH_{\Rn},\pi_{\Ln}\hat \otimes \pi_{\Rn},\Omega)$
with $(\caH_{\sigman},\pi_{\sigman})$ an irreducible representation of $\sdc(\mkk_{\sigman},\mkC_{\sigman})$ for $\sigman={\Ln},{\Rn}$.
Because $\pi_\omega(\SDC(\mkK_i,\mkC_i))''$ is a type I factor, there
 is a self-adjoint unitary $\Gamma_\sigman$ on $\caH_{\sigman}$ implementing $\Theta_{\sigman}$, i.e.,
$\Ad(\Gamma_{\sigman})\pi_\sigman=\pi_{\sigman}\Theta_{\sigman}$ for $\sigman={\Ln},{\Rn}$.
From this,  $\Gamma_{\Ln}\otimes\Gamma_{\Rn}$ implements $\Theta$ in $\pi_{\Ln}\hat\otimes \pi_{\Rn}$.
Decomposing $\caH_{\sigman}$ with $\caH_{\sigman\pm}:=\frac{\unit\pm \Gamma_{\sigman}}{{\Rn}}\caH_{\sigman}$, 
as $\caH_\sigman=\caH_{\sigman+}\oplus \caH_{\sigman-}$,
we have $\left.\pi\right\vert_{\sdc^{(0)}\lmk\mkk_{i},\mkC_{i}\rmk}
=\pi_{\sigman+}\oplus \pi_{\sigman-}$
with $\pi_{\sigman\pm}$ mutually singular irreducible representation of $\sdc^{(0)}\lmk\mkk_i,\mkC_i\rmk$.
This last property follows from Lemma 6.24 of \cite{EK}.

Because $\unih {\Ln}\otimes
\Gamma_{\Rn}\in \caB(\caH_{\Ln}\otimes\caH_{\Rn})=
\lmk  \lmk \pi_{\Ln}\hat\otimes \pi_{\Rn}\rmk
\lmk \sdc\lmk\mkk_{\Ln}\oplus \mkk_{\Rn},\mkC_{\Ln}\oplus \mkC_{\Rn}\rmk\rmk\rmk''$ is even with respect to 
$\Ad\lmk\Gamma_{\Ln}\otimes\Gamma_{\Rn}\rmk$,
we have $\unih {\Ln}\otimes \Gamma_{\Rn}\in\lmk  \lmk \pi_{\Ln}\hat\otimes \pi_{\Rn}\rmk
\lmk \sdc\lmk\mkk_{\Ln}\oplus \mkk_{\Rn},\mkC_{\Ln}\oplus \mkC_{\Rn}\rmk\rmk^{(0)}\rmk''$.
Therefore, 
\begin{align}
\omega\circ\Theta_2 \vert_{ \sdc\lmk\mkk_{\Ln}\oplus \mkk_{\Rn},\mkC_{\Ln}\oplus \mkC_{\Rn}\rmk^{(0)}}(\cdot)
=\braket{\lmk \unih {\Ln}\otimes
\Gamma_{\Rn}\rmk\Omega}{
\lmk \pi_{\Ln}\hat\otimes\pi_{\Rn}\rmk\vert_{ \sdc\lmk\mkk_{\Ln}\oplus \mkk_{\Rn},\mkC_{\Ln}\oplus \mkC_{\Rn}\rmk^{(0)}}(\cdot)
\lmk\unih {\Ln}\otimes
\Gamma_{\Rn}\rmk\Omega}
\end{align}
is quasi-equivalent to $\omega\vert_{ \sdc\lmk\mkk_{\Ln}\oplus \mkk_{\Rn},\mkC_{\Ln}\oplus \mkC_{\Rn}\rmk^{(0)}}$.
\\

(ii)If both of  $\pi_\omega(\SDC(\mkK_i,\mkC_i))''$, $i=1,2$ are summation of
two type I factors, then by Lemma \ref{sanap}, 
the state $\omega$ has a GNS representation of the form
$(\caH_{\Ln}\otimes \caH_{\Rn}\otimes \bbC^2, \pi,\Omega)$.
 There are irreducible representations  $\pi_{\sigman}$
 of $\sdc^{(0)}\lmk \mkk_{\sigman},\mkC_{\sigman}\rmk$, on $\caH_{\sigman}$,
 $\sigman={\Ln},{\Rn}$ such that
 \begin{align}
 \pi\lmk a\hat\otimes b\rmk
 =\pi_{\Ln}(a)\otimes\pi_{\Rn}(b)\otimes \unit_{\bbC^{\Rn}},\quad
 a\in \sdc^{(0)}\lmk \mkk_{\Ln},\mkC_{\Ln}\rmk, \quad b\in \sdc^{(0)}\lmk \mkk_{\Rn},\mkC_{\Rn}\rmk.
 \end{align}  
 We have
 \begin{align}
 \pi\lmk\sdc\lmk \mkk_{\Ln},\mkC_{\Ln}\rmk\rmk''
 =\caB(\caH_{\Ln})\otimes\bbC\unit_{\caH_{\Rn}}\otimes\lmk \bbC\sigma_z+\bbC\unit\rmk,\quad
 \pi\lmk\sdc\lmk \mkk_{\Rn},\mkC_{\Rn}\rmk\rmk''
 =\bbC\unit_{\caH_{\Ln}}\otimes \caB(\caH_{\Rn})\otimes\lmk \bbC\sigma_x+\bbC\unit\rmk.
 \end{align}
We also have $\Ad\lmk \unit_{\caH_{\Ln}}\otimes\unit_{\caH_{\Rn}}\otimes
\sigma_y
\rmk\circ\pi=\pi\Theta$.
Set $\Gamma_2:=\unit_{\caH_{\Ln}}\otimes\unit_{\caH_{\Rn}}\otimes
\sigma_z$.
Note that $\Ad\lmk\Gamma_2\rmk\circ\pi=\pi\Theta_2$
and 
\begin{align}
\pi\lmk\lmk \sdc\lmk\mkk_{\Ln}\oplus \mkk_{\Rn},\mkC_{\Ln}\oplus \mkC_{\Rn}\rmk\rmk^{(0)}\rmk''
=\caB(\caH_{\Ln})\otimes \caB(\caH_{\Rn})\otimes \lmk\bbC\unit+\bbC\sigma_y\rmk.
\end{align}
Hence the center of $\pi\lmk\lmk \sdc\lmk\mkk_{\Ln}\oplus \mkk_{\Rn},\mkC_{\Ln}\oplus \mkC_{\Rn}\rmk\rmk^{(0)}\rmk''$
is $\bbC\unih {\Ln} \otimes \bbC \unih {\Rn}\otimes \lmk\bbC r_++\bbC r_-\rmk$
with $r_\pm:=\unih {\Ln} \otimes \unih {\Rn}\otimes \frac{\unit\pm \sigma_y}{2}$,
and $\Ad\lmk\Gamma_2\rmk$ flips $r_+$ and $r_-$.
From this, $\pi_{0,\pm}:=\left.\pi\right\vert_{\lmk \sdc\lmk\mkk_{\Ln}\oplus \mkk_{\Rn},\mkC_{\Ln}\oplus \mkC_{\Rn}\rmk\rmk^{(0)}}(\cdot) r_{\pm} $
defines mutually singular irreducible representations of $\lmk\sdc\lmk\mkk_{\Ln}\oplus \mkk_{\Rn},\mkC_{\Ln}\oplus \mkC_{\Rn}\rmk\rmk^{(0)}$.
Because $\omega\vert_{\lmk \sdc\lmk\mkk_{\Ln}\oplus \mkk_{\Rn},\mkC_{\Ln}\oplus \mkC_{\Rn}\rmk\rmk^{(0)}}$
is pure, there is some $\zeta=\pm$ such that $\Omega= r_\zeta\Omega$.
Note that
\begin{align}
\Gamma_2\Omega
=\Gamma_2r_\zeta\Omega
=\Gamma_2r_\zeta\Gamma_2 \Gamma_2\Omega
=r_{-\zeta} \Gamma_2\Omega.
\end{align}
From these, we have
\begin{align}
\begin{split}
&\omega\vert_{\lmk \sdc\lmk\mkk_{\Ln}\oplus \mkk_{\Rn},\mkC_{\Ln}\oplus \mkC_{\Rn}\rmk\rmk^{(0)}}
=\braket{\Omega}{\pi_{0,\zeta}\lmk\cdot\rmk \Omega},\\
&\omega\vert_{\lmk \sdc\lmk\mkk_{\Ln}\oplus \mkk_{\Rn},\mkC_{\Ln}\oplus \mkC_{\Rn}\rmk\rmk^{(0)}}\circ\Theta_2
=\braket{\Omega}{\left.\pi\right\vert_{\lmk \sdc\lmk\mkk_{\Ln}\oplus \mkk_{\Rn},\mkC_{\Ln}\oplus \mkC_{\Rn}\rmk\rmk^{(0)}}\circ\Theta_2\lmk\cdot \rmk\Omega}\\
&=\braket{\Gamma_2\Omega}{\left.\pi\right\vert_{\lmk \sdc\lmk\mkk_{\Ln}\oplus \mkk_{\Rn},\mkC_{\Ln}\oplus \mkC_{\Rn}\rmk\rmk^{(0)}}\lmk\cdot \rmk\Gamma_2\Omega}
=\braket{\Gamma_2\Omega}{\pi_{0,-\zeta}\lmk\cdot \rmk\Gamma_2\Omega},
\end{split}
\end{align}
Because $\pi_{0,\pm}$ are mutually singular,
$\omega\vert_{\lmk \sdc\lmk\mkk_{\Ln}\oplus \mkk_{\Rn},\mkC_{\Ln}\oplus \mkC_{\Rn}\rmk\rmk^{(0)}}$
 and $\omega\vert_{\lmk \sdc\lmk\mkk_{\Ln}\oplus \mkk_{\Rn},\mkC_{\Ln}\oplus \mkC_{\Rn}\rmk\rmk^{(0)}}\circ\Theta_2$
are disjoint.

The last statement comes from the fact
that $\pi_\omega\circ \lmk\eta_{\Ln}\hat\otimes\eta_{\Rn}\rmk$ 
is a GNS representation of $\omega\circ\lmk\eta_{\Ln}\hat\otimes\eta_{\Rn}\rmk$
and 
\[
\pi_\omega\circ \lmk\eta_{\Ln}\hat\otimes\eta_{\Rn}\rmk\lmk\sdckc{{\Rn}}\rmk''
=\pi_\omega\lmk \sdckc{{\Rn}}\rmk'',
\]
where the right hand side is a factor if and only if
the left hand side is.
\end{proof}
Here is the main Proposition of this section.
\begin{prop}\label{mainprop}
Let $\Lambda_\sigma\subset \bbZ^2$, $\sigma=L,R$ be mutually disjoint
infinite subsets of $\bbZ^2$. 
Set
$\mkk_{\sigma}:=\mkh_{\Lambda_\sigma}$ 
with the complex conjugation $\mkC_{\sigma}:=\mkC_{\Lambda_\sigma}$, for $\sigma=L,R$.
Let $\omega_0,\omega_1$ be homogeneous pure states on 
$\sdc\lmk\mkk_L\oplus \mkk_R,\mkC_L\oplus \mkC_R\rmk$,
satisfying the split property.
Note from Lemma \ref{lem16}, either (i) or (ii) of Lemma \ref{lem16} occurs.
If
\begin{description}
\item[(a)]
(i) occurs for both of $\omega_0$, $\omega_1$, or
\item[(b)]
(ii) occurs for both of $\omega_0$, $\omega_1$,
\end{description}
then there are automorphisms $\eta_\sigma\in \Aut^{(0)}\lmk\sdc\lmk \mkk_\sigma,\mkC_\sigma\rmk\rmk$
$\sigma=L,R$
satisfying
\begin{align}
\omega_1\simeq \omega_0\lmk \eta_L\hat\otimes\eta_R\rmk.
\end{align}
Furthermore, if none of (a), (b) occurs,
then $\omega_1$ and $\omega_0$ are not quasi-equivalent.
\end{prop}

The proof is a deformation of that of \cite{GS}.
As we are considering non-twisted crossed product here, it is even simpler than the case of \cite{GS}.
For the rest of this section, $\Lambda$ indicates an infinite subset of $\bbZ^2$.
We recall notations from \cite{GS} adapted to our current setting.
An irreducible covariant representation of
$\Sigma_\Lambda:=(\bbZ_2,\caA_{\Lambda}, \Theta_\Lambda)$
is a triple $(\caH,\pi, \Gamma)$
where $(\caH,\pi)$ is an irreducible representation of $\caA_\Lambda$
and $\Gamma$ a self-adjoint unitary on $\caH$
satisfying $\Ad \Gamma \circ \pi=\pi\circ \Theta_\Lambda$.

Let $C(\bbZ_2,\caA_{\Lambda})$ be 
the linear space of $\caA_{{\Lambda}}$-valued functions 
on $\bbZ_2$. 
We equip $C(\bbZ_2,\caA_{\Lambda})$ with a product and $*$-operation as follows:
\begin{align}
&f_{1}*f_{2}(h)
:=\sum_{g\in {\bbZ_2}} f_{1}(g)\cdot \Theta_\Lambda^g\lmk f_{2}(g^{-1}h)\rmk,\quad h\in {\bbZ_2},\\
&f^{*}(h):=
\Theta^g_\Lambda\lmk h\rmk \lmk
f(h^{-1})^{*}
\rmk,\quad h\in {\bbZ_2},
\end{align}
for $f_{1},f_{2},f\in C(\bbZ_2,\caA_{\Lambda})$.
The linear space $C(\bbZ_2,\caA_{\Lambda})$ which is a $*$-algebra
with these operations is denoted by $C(\Sigma_{\Lambda})$.

For a covariant representation $(\caH,\pi, \Gamma)$ of 
$\Sigma_{\Lambda}$, we may introduce a $*$-representation 
$(\caH,\pi\times \Gamma)$ of
$C(\Sigma_{\Lambda})$ by
\begin{align}
\lmk \pi\times \Gamma\rmk (f)
:=\pi\lmk f(0)\rmk+\pi\lmk f(1)\rmk \Gamma,\quad f\in C(\Sigma_{\Lambda}).
\end{align}
The full twisted crossed product of $\Sigma_{\Lambda}$, denoted $C^{*}(\Sigma_{\Lambda})$
is the completion of $C(\Sigma_{\Lambda})$ with respect to the norm
\begin{align}
\lV f\rV_{u}:=
\sup\left\{
\lV\lmk \pi\times \Gamma\rmk (f)\rV
\mid
(\pi,\Gamma) : \text{covariant representation}
\right\},\quad f\in C(\Sigma_{\Lambda}).
\end{align}
Because $\bbZ_2$ is finite, we actually have $C(\Sigma_{\Lambda})=C^{*}(\Sigma_{\Lambda})$.
It also coincides with the reduced crossed product $C_r^*(\Sigma_\Lambda)$
and simple because $\Theta$ is properly outer \cite{Elliott}.
 As $A_{\Lambda}$ is unital, we have unitaries $\lambda_{1}\in C^{*}(\Sigma_{\Lambda})$,
 such that
 \begin{align}\label{ldef}
  &\lambda_{1}a \lambda_{1}^{*}=\Theta_\Lambda\lmk a\rmk,\quad a\in \caA_{\Lambda}.
   \end{align}
%

For the rest of this section, we use the following notations.
\begin{nota}
For each finite subset $\Lambda_0$ of $\Lambda$,
let $\left\{e_{IJ}^{(\Lambda_0)}\right\}_{I,J}$ be a system of matrix units spanning
$\caA_{\Lambda_0}$ with 
$\Theta(e_{IJ}^{(\Lambda_0)})=(-1)^{|I|+|J|}e_{IJ}^{(\Lambda_0)}$. 
Because $d$ is even, $\mkh_{\Lambda_0}$ is even-dimensional and
our $\caA_{\Lambda_0}=\SDC(\mathfrak h_{\Lambda_0})$ is 
isomorphic to a finite dimensional CAR-algebra.
Therefore, such system of matrix units exists.
Fix some $I_0$.
 Set 
 $\caG_{\Lambda_0}:=\left\{
 \frac{1}{\sqrt 2} e_{II_0}^{(\Lambda_0)},
 \frac{1}{\sqrt 2} \lambda_1e_{II_0}^{(\Lambda_0)}
 \right\}_{I}$.\\
For each $\varepsilon>0$, 
we fix $\delta_{\rk}(\varepsilon)>0$ from Lemma B.4 of \cite{GS}.\\
For any finite subset $\caF$ of $\lmk \caA_\Lambda\rmk_1$ and 
$\varepsilon>0$,
there exists a finite subset $\Lambda(\varepsilon, \caF)$ of $\Lambda$
such that
\begin{align}
\inf\left\{
\lV a-b\rV\mid b\in \caA_{\Lambda(\varepsilon, \caF)}
\right\}<\frac \varepsilon{16},
\end{align}
for all $a\in\caF$.
For each finite subset $\caF$ of $\lmk \caA_\Lambda\rmk_1$ and 
$\varepsilon>0$, we fix such $\Lambda(\varepsilon, \caF)$.
\end{nota}

First we prepare a proposition needed for case (a) of Proposition \ref{mainprop}.
For  an irreducible covariant representation $(\caH,\pi, \Gamma)$ of
$(\bbZ_2,\caA_{\Lambda}, \Theta_\Lambda)$,
we set
$\caH_{\pm}:=\frac{\unit  \pm \Gamma}{2}\caH$.
Then from Lemma 6.24 of \cite{EK},
there are mutually disjoint irreducible representations $\pi_{\pm}$ of $\caA_{\Gamma}^{(0)}$,
acting on $\caH_{\pm}$ respectively
which decompose 
the restriction $\left.\pi\right\vert_{\caA_{\Gamma}^{(0)}}$
of $\pi$ as
\begin{align}
\left.\pi\right\vert_{\caA_{\Gamma}^{(0)}}
=\pi_{+}\oplus \pi_{-}.
\end{align} 
With a bit abuse of notation, we call $(\caH_\pm,\pi_\pm)$
the decomposition associated to $(\caH,\pi, \Gamma)$.
We show the following Proposition.
\begin{prop}\label{lem6}
Let $(\caH_i,\pi_i, \Gamma_i)$, $i=0,1$ be
irreducible covariant representations of
$\Sigma_\Lambda$.
Let $(\caH_{i\pm},\pi_{i\pm})$
be the decomposition associated to $(\caH_i,\pi_i, \Gamma_i)$,
and $\xi_i\in \caH_{i+}$ a unit vector.
Define a state $\omega_i$ on $\caA_{\Lambda}$
by
\begin{align}
\omega_i:= \braket{\tilde \xi_i}{\pi_i(\cdot )\tilde \xi_i}
\end{align}
with
$\tilde \xi_i:=\xi_i\oplus 0\in \caH_i$.
Then there is an automorphism $\eta\in \Aut^{(0)}\lmk \caA_\Lambda\rmk$
such that
$\omega_1=\omega_0\circ\eta$.
\end{prop}
The point of the Proposition is that we can take $\eta$ graded.
The proof is the same as that of Proposition 4.2 \cite{GS} which relies on the techniques developed
in $C^*$-algebra theory \cite{kos} \cite{Farah}.

We just sketch the proof below.

\begin{lem}\label{lem2}
Let $(\caH_i,\pi_i, \Gamma_i)$, $i=0,1$ be
irreducible covariant representations of
$\Sigma_\Lambda:=\Sigma_\Lambda$.
Let $(\caH_{i\pm},\pi_{i\pm})$
be the decomposition associated to $(\caH_i,\pi_i, \Gamma_i)$,
and $\xi_i\in \caH_{i+}$ a unit vector.
With $\tilde \xi_i:=\xi_i\oplus 0\in \caH_i$,
let $\hat \omega_i$ be pure states on $C^*(\Sigma)$ defined by
\begin{align}
\hat\omega_i:=\braket{\tilde \xi_i}{\lmk \pi_i\times \Gamma_i\rmk(\cdot)\tilde \xi_i}.
\end{align}
Then for any finite set $\caF\subset C^*(\Sigma)$, $\varepsilon>0$,
there is a self-adjoint element $h\in \caA^{(0)}_{\Gamma}$
such that
\begin{align}
\lv
\hat\omega_0(a)-\hat\omega_1\circ\Ad(e^{ih})(a)
\rv<\varepsilon,\quad \text{for all}\quad
a\in \caF.
\end{align}
\end{lem}
\begin{proof}
This is a special case of Lemma 4.5 \cite{GS}.

\end{proof}
\begin{lem}\label{lem4}
For any finite subset $\Lambda_0\subset \Lambda$ and 
$\varepsilon>0$,
there exists
$\tdnh(\epsilon,\Lambda_0)>0$
satisfying the followings:
For any irreducible covariant representation $(\caH,\pi,\Gamma)$ of $\Sigma_\Lambda$
and unit vectors $\xi,\eta\in \caH$
satisfying
\begin{align}\label{go4}
\begin{split}
&\lmk \pi\times \Gamma\rmk(x^*) \xi
\perp
\lmk \pi\times \Gamma\rmk(y^*) \eta,\\
&\lv
\braket{\xi}{\lmk \pi\times \Gamma\rmk (xy^*)\xi}
-\braket{\eta}{\lmk \pi\times \Gamma\rmk (xy^*)\eta}
\rv<\tdnh(\epsilon,\Lambda_0),\\
&\text{for any}\quad x,y\in \caG_{\Lambda_0},
\end{split}
\end{align}
there exists an even positive element $\hat h \in \caA_{\Lambda\setminus{\Lambda_0},1}$
such that 
\begin{align}\label{ichi4}
\lV
e^{i\pi \pi(\hat h)}\xi-\eta
\rV<\frac{\dro(\frac\epsilon 8)}{4\sqrt 2}.
\end{align}
\end{lem}
\begin{proof}
With $\dnh(\varepsilon,n)$ given in Lemma B.6 \cite{GS},
we set $\tdnh(\epsilon,\Lambda_0):=\dnh(\epsilon, 2^{d|\Lambda_0|+1})>0$.
Suppose that
$(\caH,\pi,\Gamma)$, $\xi,\eta$
satisfy the condition (\ref{go4}) for this $\tdnh(\epsilon,\Lambda_0)$.
Then by Lemma B.4 of \cite{GS}, 
there exists $h\in \lmk C^*(\Sigma_\Lambda)\rmk_{+,1}$
such that 
\begin{align}
\lV
\lmk \pi\times \Gamma\rmk(\bar h)\lmk
\xi+\eta
\rmk
\rV,\;\;
\lV
\lmk \unit-\lmk \pi\times \Gamma\rmk(\bar h)\rmk\lmk
\xi-\eta
\rmk
\rV
<\frac{\dro(\frac{\epsilon} 8)e^{-\pi}}{4\sqrt 2},
\end{align}
for
\begin{align}
\bar h:=\frac 12 
\sum_{\sigma=0,1} \sum_I\lambda_1^\sigma \eijz{I}{I_0} h \eijz{I_0}{I}
\lmk \lambda_1^\sigma\rmk^*
\in C^*(\Sigma).
\end{align}
By the definition, we see that $\bar h$
commute with $\lambda_1$ and elements in $\caA_{\Lambda_0}$.
From Lemma 5.4 of \cite{1dFermi},
there are irreducible covariant representations $(\caH_1,\pi_1,\Gamma_1)$, $(\caH_2,\pi_2,\Gamma_2)$
of $\caA_{\Lambda_0}$, $\caA_{\Lambda\setminus \Lambda_0}$
and a unitary $W:\caH\to \caH_1\otimes \caH_2$
such that 
$\Ad W\circ\pi=\pi_1\hat\otimes \pi_2$
and $\Ad\lmk W\rmk (\Gamma)=\Gamma_1\otimes \Gamma_2$.
Because $\Ad W\lmk \pi(\bar h)\rmk$
commutes with $\Ad W\lmk \pi\lmk \caA_{\Lambda_0}\rmk\rmk=\caB(\caH_1)\otimes \unit$
and $\Ad\lmk W\rmk (\Gamma)=\Gamma_1\otimes \Gamma_2$,
$\Ad W\lmk \pi(\bar h)\rmk$ is of the form
$\unih 1\otimes x$
with $x\in \caB(\caH_2)_{+,1}$ such that $\Ad \Gamma_2(x)=x$.
Because $(\caH_2,\pi_2,\Gamma_2)$ is an  irreducible covariant representation,
as in the proof of Lemma 4.9 \cite{GS}, 
there is some $\hat h\in \lmk \caA^{(0)}_{\Lambda\setminus\Lambda_0}\rmk_{1,+}$
such that 
$\lV
\lmk \unih  1\otimes \pi_2 (\hat h)-\unih 1\otimes x
\rmk
W\lmk \xi+\eta\rmk
\rV$ and
$\lV
\lmk \unih  1\otimes \pi_2 (\hat h)-\unih 1\otimes x
\rmk
W\lmk \xi-\eta\rmk
\rV$
are small enough so that 
\begin{align}
\begin{split}
&\lV
\pi(\hat h)(\xi+\eta)
\rV
\le
\lV
\lmk \pi(\hat h)-\lmk\pi\times \Gamma\rmk(\bar h)\rmk
\lmk \xi+\eta\rmk
\rV
+
\lV
\lmk\pi\times \Gamma\rmk(\bar h)(\xi+\eta)\rV\\
&=\lV
\lmk \unih  1\otimes \pi_2 (\hat h)-\unih 1\otimes x
\rmk
W\lmk \xi+\eta\rmk
\rV+
\lV
\lmk\pi\times \Gamma\rmk(\bar h)(\xi+\eta)\rV
<\frac{\dro(\frac{\epsilon} 8)e^{-\pi}}{4\sqrt 2},\\
&\lV
\lmk\unit-\pi(\hat h)\rmk(\xi-\eta)
\rV
\le
\lV
\lmk \pi(\hat h)-\lmk\pi\times \Gamma\rmk(\bar h)\rmk
\lmk \xi-\eta\rmk
\rV
+
\lV
\lmk \unit-\lmk\pi\times \Gamma\rmk(\bar h)\rmk (\xi-\eta)
\rV\\
&=
\lV
\lmk \unih  1\otimes \pi_2 (\hat h)-\unih 1\otimes x
\rmk
W\lmk \xi-\eta\rmk
\rV
+
\lV
\lmk \unit-\lmk\pi\times \Gamma\rmk(\bar h)\rmk (\xi-\eta)
\rV
<\frac{\dro(\frac{\epsilon} 8)e^{-\pi}}{4\sqrt 2},
\end{split}
\end{align}
hold.
Then by the argument in  Lemma 4.9 of \cite{GS} (eq. (91)),
for this $\hat h$,
we obtain (\ref{ichi4}).
\end{proof}
\begin{lem}\label{lem5}
For any $\varepsilon>0$ and a finite subset $\caF$
of $\caA_{\Lambda}$, set
$\tdsj\lmk \varepsilon,\caF\rmk:=\frac 12 \tdnh\lmk \frac\epsilon 4,\Lambda(\varepsilon,\caF)\rmk$
and $\tgsj \lmk \varepsilon,\caF\rmk:=\caG_{\Lambda(\varepsilon,\caF)}$,
where $\tdnh(\varepsilon,\Lambda_0)$ is defined in Lemma \ref{lem4}.
Then for any irreducible covariant representation $(\caH,\pi, \Gamma)$ of $\Sigma_\Lambda$
with the associated decomposition $(\caH_\pm,\pi_\pm)$
and unit vectors $\xi,\eta\in \caH_+$
such that
\begin{align}\label{ichi5}
\lv
\braket{\tilde \xi}{\lmk\pi\times \Gamma\rmk(xy^*)\tilde \xi}-
\braket{\tilde \eta}{\lmk\pi\times \Gamma\rmk(xy^*)\tilde \eta}
\rv
<\tdsj(\varepsilon,\caF),\quad x,y\in \tgsj(\varepsilon,\caF).
\end{align}
for $\tilde \xi:=\xi\oplus 0$, $\tilde \eta:=\eta\oplus 0\in \caH_+\oplus \caH_-$,
there is a continuous map $v :[0,1]\to \caU(\caA_{\Lambda}^{(0)})$
satisfying  
\begin{align}\label{san5}
\begin{split}
&\tilde \eta=\pi\lmk v(1)\rmk\tilde \xi,\quad 
v(0)=\unit,\\
&\sup_{t\in [0,1]}\lV \Ad v(t)(a)-a\rV<\varepsilon,\quad a\in \caF.
\end{split}
\end{align}
\end{lem}
\begin{proof}
Let $\caK$ be a finite dimensional subspace of $\caH$
spanned by
\[
\left\{
\lmk\pi\times \Gamma\rmk(xy^*)\tilde\xi,\;
\lmk\pi\times \Gamma\rmk(xy^*)\tilde\eta\mid
x,y\in \tgsj(\varepsilon,\caF)
\right\}.
\]
Then because $C^*(\Sigma_\Lambda)$ is simple,
by Glimm's Lemma (Lemma 5.2.5 \cite{Farah}), there exists a unit vector $\zeta\in \caK^{\perp}$
such that
\begin{align}
\lv
\braket{\tilde \xi}{\lmk\pi\times \Gamma\rmk(xy^*)\tilde \xi}-
\braket{\zeta}{\lmk\pi\times \Gamma\rmk(xy^*) \zeta}
\rv
<\tdsj(\varepsilon,\caF)< \tdnh\lmk \frac\epsilon 4,\Lambda(\varepsilon,\caF)\rmk,\quad x,y\in \tgsj(\varepsilon,\caF).
\end{align}
Combining this with (\ref{ichi5}),we also have
\begin{align}
\lv
\braket{\tilde \eta}{\lmk\pi\times \Gamma\rmk(xy^*)\tilde \eta}-
\braket{\zeta}{\lmk\pi\times \Gamma\rmk(xy^*) \zeta}
\rv
<2\tdsj(\varepsilon,\caF)= \tdnh\lmk \frac\epsilon 4,\Lambda(\varepsilon,\caF)\rmk,\quad x,y\in \tgsj(\varepsilon,\caF).
\end{align}
Because we have
\begin{align}
\lmk\pi\times \Gamma\rmk(y^*)\tilde\xi,\;
\lmk\pi\times \Gamma\rmk(y^*)\tilde\eta
\perp
\lmk\pi\times \Gamma\rmk(x^*)\zeta
\;\;
\text{for any}\;\;
x,y\in \tgsj(\varepsilon,\caF),
\end{align}
from Lemma \ref{lem4},
there are $h_1,h_2\in \lmk \caA_{\Lambda\setminus\Lambda(\varepsilon,\caF)}^{(0)}\rmk_{1,+}$
such that 
\begin{align}
\lV e^{i\pi \pi(h_1)}\tilde \xi-\zeta\rV,
\lV e^{i\pi \pi(h_2)}\tilde \eta-\zeta\rV
<\frac{1}{4\sqrt 2}\dro\lmk\frac{\varepsilon}{32}\rmk.
\end{align}
From this, recalling that $h_1$ and $h_2$ are even,
 we have
\begin{align}
\lV
e^{i\pi \pi_+(h_1)}\xi-e^{i\pi \pi_+(h_2)}\eta
\rV_{\caH_+}
=\lV
e^{i\pi \pi(h_1)}\tilde \xi- e^{i\pi \pi(h_2)}\tilde \eta
\rV
<\frac{1}{2\sqrt 2}\dro\lmk\frac{\varepsilon}{32}\rmk.
\end{align}
Here $\lV\cdot \rV_{\caH_+}$ means
the norm associated to $\caH_+$.
Because $\pi_+$ is irreducible, by the argument in the proof of Lemma 4.12 in \cite{GS}
(around equation (103)),
there is some self-adjoint $K\in \caA^{(0)}$ such that
\begin{align}
\begin{split}
e^{i\pi_+(K)}e^{i\pi \pi_+(h_1)}\xi=e^{i\pi\pi_+(h_2)}\eta,\\
\lV K\rV<\dn\lmk \frac{\varepsilon}{32}\rmk,\quad
\sup_{\theta\in [0,1]} \lV e^{iK\theta}-\unit\rV\le \frac{\varepsilon}{32}.
\end{split}
\end{align}
Here, $\delta_{\nii}(\varepsilon)$ is given in Notation B.3 of \cite{GS}. 
From this, as in the proof of Lemma 4.12 \cite{GS} (from equation (105)),
we can find $v(t)$
satisfying the condition (\ref{san5}).
\end{proof}
\begin{lem}\label{lem7}
Let $(\caH_i,\pi_i, \Gamma_i)$, $i=0,1$ be
irreducible covariant representations of
$\Sigma_\Lambda$.
Let $(\caH_{i\pm},\pi_{i\pm})$
be the decomposition associated to $(\caH_i,\pi_i, \Gamma_i)$,
and $\xi_i\in \caH_{i+}$ a unit vector.
Set $\tilde \xi_i:=\xi_i\oplus 0\in \caH_{i+}\oplus \caH_{i-}$.
Suppose that for $\varepsilon>0$ and a finite set $\caF\subset \lmk \caA_{\Lambda}\rmk_1$,
\begin{align}\label{ichi7}
\lv
\braket{\tilde \xi_0}{\lmk \pi_0\times \Gamma_0\rmk(x y^*)\tilde\xi_0}
-\braket{\tilde \xi_1}{\lmk \pi_1\times \Gamma_1\rmk(x y^*)\tilde\xi_1}
\rv
<\frac 12 \tdsj(\varepsilon, \caF),\quad
x,y\in \tgsj(\varepsilon,\caF)
\end{align}
hold.
Then for any $\varepsilon'>0$ and finite set $\caF'\subset C^*(\Sigma)$,
there exists a norm continuous path
$v: [0,1]\to \caU(\caA_\Lambda^{(0)})$
such that
\begin{align}
\begin{split}
&\lv
\braket{\tilde \xi_0}{\lmk \pi_0\times \Gamma_0\rmk(a)\tilde\xi_0}
-
\braket{\tilde \xi_1}{\lmk \pi_1\times \Gamma_1\rmk\circ \Ad\lmk v(1)\rmk(a)\tilde \xi_1}
\rv
<\varepsilon',\quad a\in \caF',\\
&
\lV
\Ad\lmk v(t)\rmk(y)-y
\rV<\varepsilon,\quad y\in\caF,\; t\in [0,1].
\end{split}
\end{align}
\end{lem}
\begin{proof}
This corresponds to Lemma 4.14 of \cite{GS}.
By Lemma \ref{lem2}, there exists a self-adjoint $h\in \caA_\Lambda^{(0)}$
such that
\begin{align}
\lv
\braket{\tilde \xi_0}{\lmk \pi_0\times \Gamma_0\rmk(a)\tilde\xi_0}
-
\braket{\tilde \xi_1}{\lmk \pi_1\times \Gamma_1\rmk\circ \Ad\lmk e^{ih}\rmk(a)\tilde \xi_1}
\rv
<\min\{\varepsilon', \frac 12 \tdsj(\varepsilon,\caF)\},\quad
a\in \caF'\cup\tgsj(\varepsilon,\caF)\lmk \tgsj(\varepsilon,\caF)\rmk^*.
\end{align}
Combining this with (\ref{ichi7}),
we obtain
\begin{align}
\lv
\braket{\tilde \xi_1}{\lmk \pi_1\times \Gamma_1\rmk\circ \Ad\lmk e^{ih}\rmk(xy^*)\tilde \xi_1}
-\braket{\tilde \xi_1}{\lmk \pi_1\times \Gamma_1\rmk(x y^*)\tilde\xi_1}
\rv
<\tdsj(\varepsilon,\caF),\quad
x,y\in \tgsj(\varepsilon,\caF).
\end{align}
 Here, because $h$ is even,
 we have
 \begin{align}
 \lmk \pi_1\times \Gamma_1\rmk\lmk e^{-ih}\rmk\tilde \xi_1
 =e^{-i\pi_{1+}(h)}\xi_1\oplus 0.
 \end{align}
 From this, applying Lemma \ref{lem5},
 we obtain the Lemma.
\end{proof}

\begin{proofof}[Proposition \ref{lem6}]
Having Lemma \ref{lem2} and Lemma \ref{lem7}, the proof of Proposition \ref{lem6}
is the same as that of Proposition 4.2 of \cite{GS}.
\end{proofof}

Next we prepare a proposition needed for the case (b) of Proposition \ref{mainprop}.
\begin{prop}\label{lem11}
Let $(\caH_i,\pi_i)$, $i=0,1$ be irreducible representations of
$\caA_\Lambda$ such that $\pi_i\lmk\caA_{\Lambda}^{(0)}\rmk''=\caB(\caH_i)$.
Let $\xi_i\in \caH_i$, $i=0,1$ be unit vectors
and set
\begin{align}
\omega_i:=\braket{\xi_i}{\pi_i(\cdot)\xi_i}.
\end{align}
Then there is an automorphism $\eta\in \auz(\caA_{\Lambda})$
such that 
\begin{align}
\omega_1=\omega_0\eta.
\end{align}
\end{prop}
The proof follows the argument in \cite{kos},
choosing the unitary there each time from $\caA_{\Lambda}^{(0)}$,
which is possible because of $\pi_i\lmk\caA_{\Lambda}^{(0)}\rmk''=\caB(\caH_i)$.
\begin{lem}\label{lem8}
Let $(\caH_i,\pi_i)$, $i=0,1$ be irreducible representations
$\caA_\Lambda$ such that $\pi_i\lmk\caA_{\Lambda}^{(0)}\rmk''=\caB(\caH_i)$.
Let $\xi_i\in \caH_i$, $i=0,1$ be unit vectors.
Then for any finite subset $\caF$ of $\caA_{\Lambda}$ and $\varepsilon>0$,
there exists a self-adjoint element $h\in \caA^{(0)}_{\Lambda}$
such that
\begin{align}
\lv
\braket{\xi_0}{\pi_0(a)\xi_0}
-\braket{\xi_1}{\pi_1\circ\Ad(e^{ih})(a)\xi_1}
\rv<\varepsilon,\quad a\in \caF.
\end{align}
\end{lem}
\begin{proof}
Let $\omega_i:=\braket{\xi_i}{\pi_i(\cdot)\xi_i}$ $i=0,1$
be pure states on $\caA_{\Lambda}$.
By Lemma B.1 \cite{GS},
there exist $f\in \lmk \caA_{\Lambda}\rmk_{+1}$ and 
a unit vector $\xi_1'\in \caH_1$ 
such that
\begin{align}
\begin{split}
\pi_1(f)\xi_1'=\xi_1',\\
\lV
f\lmk a-\omega_0(a)f\rmk
\rV<\varepsilon,\quad a\in \caF.
\end{split}
\end{align}
Because we now have $\pi_1\lmk\caA_{\Lambda}^{(0)}\rmk''=\caB(\caH_1)$,
for unit vectors $\xi_1,\xi_1'\in \caH_1$ there exists
a self-adjoint element $h\in\caA^{(0)}_\Lambda$
such that $e^{i\pi_1(h)}\xi_1'=\xi_1$
by the Kadison transitivity.
For this $h$ we have
\begin{align}
\begin{split}
\lv
\braket{\xi_0}{\pi_0(a)\xi_0}-
\braket{\xi_1}{\pi_1\Ad(e^{ih})(a)\xi_1}
\rv
=\lv
\braket{\xi_1'}
{\pi_1(f) \lmk\omega_0(a)-\pi_1(a)\rmk \pi_1(f)\xi_1'}
\rv<\varepsilon,\quad a\in \caF.
\end{split}
\end{align}
\end{proof}
\begin{lem}\label{lem9}
For any $\varepsilon>0$ and a finite subset $\Lambda_0\subset \Lambda$,
there exists $\tdnhs(\varepsilon,\Lambda_0)>0$
satisfying the following:
For any $(\caH,\pi)$ an irreducible representation of $\caA_{\Lambda}$
with $\pi\lmk\caA_\Gamma^{(0)}\rmk''=\caB(\caH)$
and unit vectors $\xi,\eta\in \caH$
satisfying
\begin{align}
\begin{split}
\pi(x^*)\xi\perp \pi(y^*)\eta,\\
\lv
\braket{\xi}{\pi(xy^*)\xi}-\braket{\eta}{\pi(xy^*)\eta}
\rv<\tdnhs(\varepsilon,\Lambda_0),\\
\text{for any }\quad x,y\in \{\eijz I{I_0}\}_I,
\end{split}
\end{align}
there exist self-adjoint operators $h\in \lmk \caA_{\Lambda}^{(0)}\rmk_{+1}$
and $k\in \caA_{\Lambda}^{(0)}$
such that
\begin{align}
\begin{split}
\bar h:=\sum_I \eijz I{I_0} h \eijz {I_0} I\in \lmk \caA_{\Lambda\setminus \Lambda_0}^{(0)}\rmk_{+1},\\
\lV
e^{i\pi \pi (\bar h)}\xi-\eta
\rV<\frac{1}{\sqrt 2}\dro\lmk \frac{\varepsilon}{8}\rmk\\
\lV
e^{itk}-\unit
\rV\le \frac{\varepsilon}{8}, \quad t\in [0,1],\\
e^{i\pi(k)}e^{i\pi\pi(\bar h)}\xi=\eta.
\end{split}
\end{align}
\end{lem}
\begin{proof}
With $\dnh(\varepsilon,n)$ given in Lemma B.6 \cite{GS},
we set
$\tdnhs(\varepsilon,\Lambda_0):=\dnh(\varepsilon, 2^{d|\Lambda_0|})$.
Suppose that $(\caH,\pi)$ is an irreducible representation of $\caA_{\Lambda}$
satisfying the conditions above with respect to this $\tdnhs(\varepsilon,\Lambda_0)$.
Applying Lemma B.6 \cite{GS}, we obtain
$h\in \lmk\caA_{\Lambda}\rmk_{+,1}$
such that
\begin{align}
\begin{split}
\lV
\bar h\lmk\xi+\eta\rmk
\rV<\frac1{4\sqrt 2}\delta_{2,a}\lmk\frac\varepsilon 8\rmk e^{-\pi},\\
\lV
\lmk\unit-\bar h\rmk\lmk\xi-\eta\rmk
\rV<\frac1{4\sqrt 2}\delta_{2,a}\lmk\frac\varepsilon 8\rmk e^{-\pi}
\end{split}
\end{align}
holds for 
\begin{align}
\bar h:=\sum_I \eijz I {I_0} h \eijz {I_0}I. 
\end{align}
This $h$ is obtained via the Kadison transitivity, using the irreducibility of $\pi$
(see section 3 of \cite{kos}, also 5.6.1 of \cite{Farah}).
Because we have $\pi\lmk\caA_\Lambda^{(0)}\rmk''=\caB(\caH)$,
in fact we can choose this $h$ from $\lmk\caA_{\Lambda}^{(0)}\rmk_{+,1}$.
Doing so, we have
%
%
%
%
\begin{align}
\bar h:=\sum_I \eijz I {I_0} h \eijz {I_0} I
\in\lmk \caA_\Lambda^{(0)}\rmk_{+1}\cap \caA_{\Lambda_0}'\subset \lmk \caA_{\Lambda\setminus \Lambda_0}^{(0)}\rmk_{+1}.
\end{align}
The last inclusion comes from Lemma 4.15 \cite{aramori}. 
For this $\bar h$,  
by the argument in  Lemma 4.9 of \cite{GS} (eq. (91)),
we have
\begin{align}
\lV
e^{i\pi \pi(\bar h)}\xi-\eta
\rV
<\frac 1{\sqrt 2} \dro\lmk\frac \varepsilon 8\rmk.
\end{align}
Because  we have $\pi\lmk\caA_\Lambda^{(0)}\rmk''=\caB(\caH)$,
by the Kadison transitivity, (Theorem B.4 \cite{GS})
we get some self-adjoint $k\in \caA_\Lambda^{(0)}$
such that
\begin{align}
e^{i\pi(k)} e^{i\pi \pi(\bar h)}\xi =\eta,\quad
\sup_{t\in [0,1]}\lV
e^{ikt}-\unit
\rV\le \frac\varepsilon 8.
\end{align}
\end{proof}
\begin{lem}\label{lem10}For any $\varepsilon>0$ and a finite subset $\caF$
of $\lmk \caA_{\Lambda}\rmk_1$, set
$\tdsjs\lmk \varepsilon,\caF\rmk:=\frac 12 \tdnhs\lmk \epsilon,\Lambda(\varepsilon,\caF)\rmk$
and $\tgsj \lmk \varepsilon,\caF\rmk:=\caG_{\Lambda(\varepsilon,\caF)}$,
where $\tdnhs\lmk \epsilon,\Lambda_0\rmk$ is defined in Lemma \ref{lem9}.
Then for any irreducible representation $(\caH,\pi)$ of 
$\caA_\Lambda$ with $\pi(\caA_\Lambda^{(0)})''=\caB(\caH)$,
and unit vectors $\xi,\eta\in \caH$
such that
\begin{align}\label{ichi10}
\lv
\braket{\xi}{\pi(x y^*)\xi}-\braket{\eta}{\pi(x y^*)\eta}
\rv<\tdsjs(\varepsilon,\caF),\quad x,y\in \tgsj\lmk \varepsilon,\caF\rmk,
\end{align}
there exists a norm-continuous path $v:[0,1]\to \caU(\caA_\Lambda^{(0)})$
such that 
\begin{align}
\begin{split}
\eta=\pi\lmk v(1)\rmk\xi,\quad v(0)=\unit,\\
\sup_{t\in [0,1]}\lV \Ad\lmk v(t)\rmk(a)-a\rV<\varepsilon,\quad
a\in \caF.
\end{split}
\end{align}
\end{lem}
\begin{proof}
Because $\caA_\Lambda$ is simple, by Glimm's Lemma, there exists a unit vector $\zeta$
such that
\begin{align}
\begin{split}
\zeta\perp\pi(xy^*)\xi,\pi(xy^*)\eta,\quad x,y\in \tgsj(\varepsilon,\caF),\\
\lv
\braket{\xi}{\pi(xy^*)\xi}- \braket{\zeta}{\pi(xy^*)\zeta}
\rv
<\tdsjs(\varepsilon,\caF)<\tdnhs(\varepsilon,\Lambda(\varepsilon,\caF)),\quad 
x,y\in \tgsj(\varepsilon,\caF).
\end{split}
\end{align}
From the latter inequality combined with (\ref{ichi10}),
we have
\begin{align}
\lv
\braket{\eta}{\pi(xy^*)\eta}- \braket{\zeta}{\pi(xy^*)\zeta}
\rv
<\tdnhs(\varepsilon,\Lambda(\varepsilon,\caF)).
\end{align}
By Lemma \ref{lem9},
there exist $h_1,h_2\in \lmk \caA_{\Lambda\setminus \Lambda(\varepsilon, \caF)}^{(0)}\rmk_{+,1}$
such that
\begin{align}
\begin{split}
\lV
e^{i\pi \pi(h_1)}\xi-\zeta
\rV
<\frac1 {\sqrt 2}\dro\lmk\frac \varepsilon 8\rmk,\quad
\lV
e^{i\pi \pi(h_2)}\eta-\zeta
\rV
<\frac1 {\sqrt 2}\dro\lmk\frac \varepsilon 8\rmk,
\end{split}
\end{align}
and self-adjoint $k_1,k_2\in \caA_\Lambda^{(0)}$
such that
\begin{align}
\begin{split}
\lV e^{it \pi(k_1)}-\unit\rV\le \frac \varepsilon 8,\quad
\lV e^{it \pi(k_2)}-\unit\rV\le \frac \varepsilon 8,\quad t\in [0,1],\\
e^{i\pi(k_1)}e^{i\pi\pi(h_1)}\xi=\zeta,\quad
e^{i\pi(k_2)}e^{i\pi\pi(h_2)}\eta=\zeta.
\end{split}
\end{align}
Out of these $h_1,h_2,k_1,k_2$,
as in the proof of Lemma 4.12 of \cite{GS}, we can obtain the desired $v(t)$.
\end{proof}
\begin{lem}\label{lem12}
Let $\varepsilon>0$ and $\caF$ a finite subset 
of $\lmk \caA_{\Lambda}\rmk_1$.
Let $(\caH_i,\pi_i)$, $i=0,1$ be irreducible representations of
$\caA_\Lambda$ such that $\pi_i\lmk\caA_{\Lambda}^{(0)}\rmk''=\caB(\caH_i)$.
Let $\xi_i\in \caH_i$, $i=0,1$ be unit vectors
satisfying
\begin{align}
\lv
\braket{\xi_0}{\pi_0(xy^*)\xi_0}-\braket{\xi_1}{\pi_1(xy^*)\xi_1}
\rv
<\frac 12 \tdsjs(\varepsilon,\caF),\quad
x,y\in \tgsj(\varepsilon,\caF),
\end{align}
where $ \tdsjs(\varepsilon,\caF)$, $ \tgsj(\varepsilon,\caF)$ are defined in Lemma \ref{lem10}.
Then for any $\varepsilon'>0$, finite subset $\caF'\subset \caA_\Lambda$,
there exists a norm-continuous path $v: [0,1]\to \caU(\caA_\Lambda^{(0)})$
such that
\begin{align}
\begin{split}
v(0)=\unit,\\
\lv
\braket{\xi_0}{\pi_0(a)\xi_0}-
\braket{\xi_1}{\pi_1\circ \Ad\lmk v(1)\rmk(a)\xi_1}
\rv<\varepsilon',\quad a\in \caF',\\
\lV
\Ad\lmk v(t)\rmk(y)-y
\rV<\varepsilon,\quad y\in \caF, t\in[0,1].
\end{split}
\end{align}
\end{lem}
\begin{proof}
The proof of this Lemma is the same as that of
Lemma \ref{lem7}, using Lemma \ref{lem8} and Lemma \ref{lem10}.
\end{proof}
\begin{proofof}[Proposition \ref{lem11}]
Having Lemma \ref{lem8} and Lemma \ref{lem12}, the proof of Proposition \ref{lem11}
is the same as that of Proposition 4.2 of \cite{GS}.
\end{proofof}
Now we are ready to prove the main proposition of this section, Proposition \ref{mainprop}.
\begin{proofof}[Proposition \ref{mainprop}]
Let $\Theta_\sigma:=\Theta_{\Lambda_\sigma}$ be the grading operator on
$\caA_{\Lambda_\sigma}=\sdckc{\sigma}$ for $\sigma=L,R$.\\
\underline{Case (a)}\\
In this case, $\omega_i$, $i=0,1$ has a GNS representation
of the form
$(\caH_{L,i}\otimes\caH_{R,i},\pi_{L,i}\hat\otimes \pi_{R,i}, \Omega_i)$
with  irreducible covariant representations
$(\caH_{\sigma, i},\pi_{\sigma,i},\Gamma_{\sigma,i})$ of $\Sigma_{\sigma,i}:=
(\bbZ_2, \caA_{\Lambda_\sigma}, \Theta_{\Lambda_\sigma})$, $\sigma=L,R$.
Let $(\caH_{\sigma, i, \pm},\pi_{\sigma,i,\pm})$ be the decomposition associated to
$(\caH_{\sigma, i},\pi_{\sigma,i},\Gamma_{\sigma,i})$ .
Choose and fix unit vectors $\xi_{\sigma,i}\in \caH_{\sigma, i, +}$, $i=0,1$, $\sigma=L,R$,
and set $\tilde \xi_{\sigma,i}:=\xi_{\sigma,i}\oplus 0\in \caH_{\sigma, i,+}\oplus \caH_{\sigma, i,-}=
\caH_{\sigma, i}$.
Then 
$\psi_{\sigma,i}:=\braket{\tilde \xi_{\sigma,i}}{\pi_{\sigma,i}(\cdot)\tilde \xi_{\sigma,i}}$, $i=0,1$
defines a pure homogeneous state on $\caA_{\Lambda_\sigma}$, $\sigma=L,R$
such that
\begin{align}
\omega_i\simeq \lmk \psi_{L,i}\hat\otimes \psi_{R,i}\rmk,\quad i=0,1.
\end{align}
 By Proposition \ref{lem6}, 
 there exist
 $\eta_\sigma\in\Aut^{(0)}\lmk\caA_{\Lambda_\sigma}\rmk$
 such that $\psi_{\sigma 1}=\psi_{\sigma 0}\circ\eta_\sigma$, $\sigma=L,R$.
 From this we have 
$\omega_1\simeq \omega_0\circ\lmk \eta_L\hat\otimes \eta_R\rmk$.\\
\underline{Case (b)}\\
In this case 
$\omega_i$, $i=0,1$ has a GNS representation
of the form
$(\caH_{L,i}\otimes \caH_{R,i}\otimes \bbC^2, \pi_i,\Omega_i)$.
There are irreducible representations  $\pi_{\sigma, i}$  
 of $\caA_{\Lambda_\sigma}^{(0)}$, on $\caH_{\sigma,i}$, $\sigma=L,R$, $i=0,1$
 such that
 \begin{align}
 \pi_i\lmk a\hat\otimes b\rmk
 =\pi_{L,i}(a)\otimes\pi_{R,i}(b)\otimes \unit_{\bbC^{2}},\quad
 a\in \caA_{\Lambda_L}^{(0)}, \quad b\in \caA_{\Lambda_R}^{(0)}.
 \end{align}  
 We have
 \begin{align}
 \pi_i\lmk\caA_{\Lambda_L}\rmk''
 =\caB(\caH_{L,i})\otimes\bbC\unit_{\caH_{R,i}}\otimes\lmk \bbC\sigma_z+\bbC\unit\rmk,\quad
 \pi_i\lmk\caA_{\Lambda_R}\rmk''
 =\bbC\unit_{\caH_{L,i}}\otimes \caB(\caH_{R,i})\otimes\lmk \bbC\sigma_x+\bbC\unit\rmk.
 \end{align}
%
%
%
%
%
%
%
%
%
%
%
From this form, there are representations 
$\rho_{i}^\sigma$ of $\caA_{\Lambda_\sigma}$
on $\caH_{\sigma i}\otimes  \bbC^2$
such that
\begin{align}
\pi_i(a)=\rho_i^L(a)\otimes \unit_{\caH_{Ri}},\quad a\in \caA_{\Lambda_L},\quad
\pi_i(b)=\unit_{\caH_{Li}}\otimes \rho_i^R(b),\quad b\in \caA_{\Lambda_R}.
\end{align}
Note that $\unit_{\caH_{Li}}\otimes \unit_{\caH_{Ri}}\otimes r_\pm^L$, 
$\unit_{\caH_{Li}}\otimes \unit_{\caH_{Ri}}\otimes r_\pm^R$
with $r_\pm^L:=\frac{\unit\pm \sigma_z}{2}$, $r_\pm^R:=\frac{\unit\pm \sigma_x}{2}$
 commute with
$\pi_i\lmk\caA_{\Lambda_L}\rmk''$
$\pi_i\lmk\caA_{\Lambda_R}\rmk''$, respectively.
Therefore, 
\begin{align}
\rho_{i\pm}^L(a)\otimes r_\pm^L:=\rho_i^L(a)\lmk \unih{Li}\otimes  r_\pm^L \rmk,\quad a\in \caA_{\Lambda_L},\quad
\rho_{i\pm}^R(b)\otimes r_\pm^R:=\rho_i^R(b)\lmk \unih{Ri}\otimes  r_\pm^R \rmk,\quad b\in \caA_{\Lambda_R}
\end{align}
define irreducible representations $\rho_{i\pm}^\sigma$ of $\caA_{\Lambda_\sigma}$
 on $\caH_{\sigma i }$, for $i=0,1$, $\sigma=L,R$.
Note also that 
\begin{align}
\rho_{i\pm}^\sigma\lmk\caA_{\Lambda_\sigma}^{(0)}\rmk''=\caB\lmk\caH_{\sigma i }\rmk.
\end{align}
 Because $\unih {Li}\otimes \unih {Ri}\otimes \sigma_y$ flips 
 $\unit_{\caH_{Li}}\otimes \unit_{\caH_{Ri}}\otimes r_{\pm}^\sigma$,
 we have
 \begin{align}\label{ppmm}
 \begin{split}
 \rho_{i\pm}^L\circ\Theta_L(a)=
 \rho_{i\mp}^L(a),\quad
 a\in \caA_{\Lambda_L},\\
 \rho_{i\pm}^R
 \circ\Theta_R(b)=
  \rho_{i\mp}^R(b)
  ,\quad
 b\in \caA_{\Lambda_R}.
 \end{split}
 \end{align}
 Choose unit vectors $\xi_{\sigma i}\in  \caH_{\sigma i}$ and define pure states
 \begin{align}
 \psi_{\sigma i}:=\braket{\xi_{\sigma i}}{\rho_{i+}^{\sigma}\lmk \cdot\rmk\xi_{\sigma i}},\quad
 i=0,1,\quad \sigma=L,R.
 \end{align}
 Hence applying Lemma \ref{lem11}, there exist
 automorphisms $\eta_{\sigma}\in \Aut^{(0)}\lmk \caA_{\Lambda_\sigma}\rmk$
 such that $\psi_{\sigma 1}=\psi_{\sigma 0}\circ\eta_\sigma$.
Then both of $(\caH_{\sigma 1}, \rho_{1+}^\sigma,\xi_{\sigma 1})$
and  $(\caH_{\sigma 0 }, \rho_{0+}^\sigma\circ\eta_\sigma,\xi_{\sigma 0})$
 are GNS triple of $\psi_{1\sigma}$.
 Therefore, there are unitaries $W_\sigma : \caH_{\sigma 0 }\to \caH_{\sigma 1 }$
such that $\Ad(W_\sigma)\circ  \rho_{0+}^\sigma\circ\eta_\sigma
=\rho_{1+}^\sigma$, $\sigma=L,R$. 
Combining this and (\ref{ppmm}),
we have $\Ad(W_\sigma)\circ  \rho_{0-}^\sigma\circ\eta_\sigma
=\rho_{1-}^\sigma$. 
Then we have
\begin{align}
\begin{split}
&\Ad\lmk W_L\otimes \unic\rmk\lmk
\rho_0^L\circ\eta_L(a)
\rmk
=\rho_1^L(a),\quad a\in \caA_{\Lambda_L},\\
&\Ad\lmk  W_R\otimes \unic\rmk\lmk
\rho_0^R\circ\eta_R(b)
\rmk
=\rho_1^R(b),\quad b\in \caA_{\Lambda_R}.
\end{split}
\end{align}
We claim $\pi_1=\Ad\lmk W_L\otimes W_R\otimes \unic\rmk\circ\pi _0\lmk
\eta_L\hat\otimes \eta_R\rmk$.
In fact we have
\begin{align}
\begin{split}
\Ad\lmk W_L\otimes W_R\otimes \unic\rmk\circ\pi _0\lmk
\eta_L\hat\otimes \eta_R\rmk(a)
=\rho_1^L(a)\otimes \unih{R1}=\pi_1(a),\quad a\in \caA_{\Lambda_L},\\
\Ad\lmk W_L\otimes W_R\otimes \unic\rmk\circ\pi _0\lmk
\eta_L\hat\otimes \eta_R\rmk(b)
=\unih{L1}\otimes \rho_1^R(b)=\pi_1(b),\quad b\in \caA_{\Lambda_R}.
\end{split}
\end{align}
Hence $\pi_1$ and $\pi _0\lmk
\eta_L\hat\otimes \eta_R\rmk$ are unitarily equivalent, and
$\omega_1\simeq \omega_0\circ \lmk
\eta_L\hat\otimes \eta_R\rmk$.

Suppose that none of (a), (b) occurs.
In this case, one of 
$\pi_{\omega_1}\lmk\caA_{\Lambda_R}\rmk''$, 
$\pi_{\omega_0}\lmk\caA_{\Lambda_R}\rmk''$
is a factor and the other is not.
Therefore, $\omega_1$ and $\omega_0$ are not quasi-equivalent.
\end{proofof}

We conclude this section with examples of the dichotomy of Lemma \ref{lem16}.
It is clear that $\omega^{(0)}$ satisfies (i) of Lemma \ref{lem16}.

Recall the basis projection $p$ (\ref{pxx}) we considered in Section \ref{introsec}, defining $\omega^{(0)}$.
Because $v_\tau$ commutes with $\mkC$,
$q:=v_\tau^* p v_\tau$
is also a basis projection of $(\mkh,\mkC)$.
We denote by $\omega^{(1)}$ the Fock state on $\al$
given by this basis projection $q=v_\tau^* p v_\tau$.
Note that $\omega^{(1)}=\omega^{(0)}\circ\tau$,
and $\omega^{(1)}\circ\tau=\omega^{(0)}$.
To describe $v_\tau^* p v_\tau$,
set for each $x\in \bbZ$
two dimensional space $\caK_{(x,0)}$
spanned by $\delta_{(2x-1,0)}$, $\delta_{(2x,0)}$.
Let $q_x$ be an orthogonal projection on
$\caK_{(x,0)}$ onto the 
one-dimensional space spanned by
$\delta_{(2x-1, 0)}+i\delta_{(2x,0)}$.
By the definition, we have
\begin{align}\label{qdec}
q=v_\tau^* p v_\tau
=Q^{(1)}_L\oplus
Q^{(1)}_R
\oplus
q_{0},
\end{align}
with
\begin{align}
Q^{(1)}_L:=p_{\bbZ_{\le -1}\times \lmk \bbZ\setminus\{0\}\rmk}
\oplus
\bigoplus_{x\in\bbZ_{\le -1}}
q_x,\quad
Q^{(1)}_R
:=p_{\bbZ_{\ge 0}\times \lmk \bbZ\setminus\{0\}\rmk}
\oplus
\bigoplus_{x\in\bbZ_{\ge 1}}
q_x,
\end{align}
which are projections on 
$l^2\lmk \bbZ_{\le -1}\times\lmk \bbZ\setminus\{0\}\rmk\rmk
\oplus l^2\lmk \bbZ_{\le -2}\times \{0\}\rmk
\subset \mkh_{H_L}$ and
$l^2\lmk \bbZ_{\ge 0 }\times\lmk \bbZ\setminus\{0\}\rmk\rmk
\oplus l^2\lmk \bbZ_{\ge 1}\times \{0\}\rmk
\subset\mkh_{H_R}$
respectively.
Noting that $q_0$ is a projection on 
a finite dimensional Hilbert space $\caK_{(0,0)}$,
we see that
there are homogeneous states $\varphi_L$,
$\varphi_R$ on $\al_{H_L}$, $\al_{H_R}$
such that $\omega^{(1)}$
is quasi-equivalent to $\varphi_L\hat\otimes\varphi_R$.
Hence $\omega^{(1)}$ satisfies the split property for $H_L-H_R$-cut.
Now we would like to show (ii) of Lemma \ref{lem16} occurs.
With a unitary $\vartheta:=\unit_{\mkh_{H_L}}\oplus -\unit_{\mkh_{H_R}}$,
we note that $\omega^{(1)}\circ\Theta_R$
is a Fock state on $\al$ given by the basis projection
$\vartheta q\vartheta^*$.
(Note that $\vartheta q\vartheta^*$
is a basis projection because $\vartheta$ and $\mkC$
commute.)
From the definition, this projection $\vartheta q\vartheta^*$
can be decomposed as
\begin{align}
\vartheta q\vartheta^*
=Q^{(1)}_L\oplus
Q^{(1)}_R
\oplus
\lmk\unit_{\caK_{(0,0)}}-q_0
\rmk.
\end{align}
Comparing this and (\ref{qdec}),
\begin{align}
q\wedge (\unit_{\mkh}-\vartheta q\vartheta^*)=q_0.
\end{align}
Noting $q_0$ is a one-rank projection, from Theorem 6.30 of \cite{EK},
$\omega^{(1)}\vert_{\caA^{(0)}}$ and $\omega^{(1)}\vert_{\caA^{(0)}}\Theta_R$
are not equivalent.
Hence (ii) of Lemma \ref{lem16}
holds for $\omega^{(1)}$.

Combining this with Proposition \ref{mainprop}, we conclude that
any homogeneous pure state satisfying the split property with respect to
the $H_L-H_R$-cut can be obtained out of either
$\omega^{(0)}$ or $\omega^{(1)}$ via an automorphism of the form
$\eta_L\hat\otimes \eta_R$.

\section{The $\PDz$-valued index}\label{indexsec}
Now we derive the $\PDz$-valued index out of SPT-phases.
\subsection{An Overview}
A brief description of the derivation is as follows.
From Theorem \ref{mo}, the ground states in SPT-phases are of the form
$\omega=\omega^{(0)}\circ\alpha$
with $\alpha\in \QAut(\caA)$.
Using the factorization property (\ref{xfac}),
we can show that the difference between $\omega^{(0)}\alpha\circ\beta_g^U\alpha^{-1}$
and $\omega^{(0)}$ is localized around $x$-axis, and in particular, 
$\omega^{(0)}\alpha\circ\beta_g^U\alpha^{-1}$ satisfies the split property.
Recall that then there is a dichotomy given by Lemma \ref{lem16}
for this $\omega^{(0)}\alpha\circ\beta_g^U\alpha^{-1}$.
Let us set $\ao g=0$ if (i) of Lemma \ref{lem16} occurs
and $\ao g=1$ if (ii) of Lemma \ref{lem16} occurs.
Applying the result Proposition \ref{mainprop} from the previous section,
it allows us to show
$\omega\circ\beta_g^U\simeq
\omega\circ\tau^{a_\omega(g)\epsilon}\lmk\eta_{g,L}^\epsilon\hat\otimes \eta_{g,R}^\epsilon\rmk
$ for $\epsilon=\pm$, with $\eta_{g,\sigma}^\epsilon$ automorphism localized in $C_\theta\cap H_\sigma$,
$\sigma=L,R$. (Proposition \ref{lem24}, Remark \ref{lem24'}.)

Next we decompose $\alpha$ as (\ref{afactoc}) where $\alpha_{L}$, $\alpha_{R}$ are automorphisms localized to the left and right infinite planes $H_{L}$, $H_{R}$
while $\Upsilon$ is localized in $\lmk C_{\theta}\rmk^{c}$.
We then have $\omega\simeq \lmk \omega_{L}\alpha_{L}\hat
\otimes \omega_{R}\alpha_{R}\rmk\circ\Upsilon$
with pure states $\omega_{L}$, $\omega_{R}$ on the left and right infinite planes.
Hence we obtain
$\lmk \omega_{L}\alpha_{L}\hat
\otimes \omega_{R}\alpha_{R}\rmk\circ\Upsilon
\circ\beta_g^U\simeq
\lmk \omega_{L}\alpha_{L}\hat
\otimes \omega_{R}\alpha_{R}\rmk\circ\Upsilon\circ\tau^{a_\omega(g)\epsilon}
\lmk\eta_{g,L}^\epsilon\hat\otimes \eta_{g,R}^\epsilon\rmk
$.

Setting
$\gamma_g^\epsilon:=\beta_g^U\lmk\eta_{g,L}^\epsilon\hat\otimes \eta_{g,R}^\epsilon\rmk
^{-1}\tau^{-a_\omega(g)\epsilon}$
for $\epsilon=\pm$, $g\in G$,
we have  $\lmk \omega_{L}\alpha_{L}\hat
\otimes \omega_{R}\alpha_{R}\rmk\circ\Upsilon\gamma_g^\epsilon\simeq 
\lmk \omega_{L}\alpha_{L}\hat
\otimes \omega_{R}\alpha_{R}\rmk\circ\Upsilon$.
Repeated use of this gives us
\[
\lmk \omega_{L}\alpha_{L}\hat
\otimes \omega_{R}\alpha_{R}\rmk\circ\Upsilon
\gamma_g^\epsilon\gamma_h^{\eg g\epsilon}{\gamma_{gh}^{\epsilon}}^{-1}
\simeq
\lmk \omega_{L}\alpha_{L}\hat
\otimes \omega_{R}\alpha_{R}\rmk\circ\Upsilon.
\]
But one can see that
\[
\gamma_g^\epsilon\gamma_h^{\eg g\epsilon}{\gamma_{gh}^{\epsilon}}^{-1}
=\widehat\bigotimes_{\sigma=L,R} \zeta_{g,h,\sigma}^\epsilon
\]
with $\zeta_{g,h,\sigma}^\epsilon\in \auz\lmk \cac{\theta}\sigma\rmk$
 (\ref{zetadefs}),
 and it commutes with $\Upsilon$. 
 Hence we obtain $\omega_{L}\alpha_{L}\zeta_{g,h,L}^\epsilon\hat
\otimes \omega_{R}\alpha_{R}\zeta_{g,h,R}^\epsilon
\simeq \omega_{L}\alpha_{L}\hat
\otimes \omega_{R}\alpha_{R}$, which implies 
$\omega_{\sigma}\alpha_{\sigma}\zeta_{g,h,\sigma}^\epsilon\simeq \omega_{\sigma}\alpha_{\sigma}$.
This equivalence gives us a unitary $u_\sigma^\epsilon(g,h)$ in
Lemma \ref{lem39}, implementing $\zeta_{g,h,\sigma}^\epsilon$
in the GNS-representation of $\omega_{\sigma}\alpha_{\sigma}$.
Note also that $\lmk \omega_{L}\alpha_{L}\hat
\otimes \omega_{R}\alpha_{R}\rmk\circ\Upsilon\gamma_g^\epsilon\simeq 
\lmk \omega_{L}\alpha_{L}\hat
\otimes \omega_{R}\alpha_{R}\rmk\circ\Upsilon$
means there is a unitary $W_g^\epsilon$ in the GNS-representation of
$\omega_{L}\alpha_{L}\hat
\otimes \omega_{R}\alpha_{R}$
implementing $\Upsilon\gamma_g^\epsilon\Upsilon^{-1}$.
The unitaries $W_g^\epsilon$ and $u_\sigma^\epsilon(g,h)$
are homogeneous with respect to the grading
and we denote their grade by $b^\epsilon_g$, $\kappa_\sigma^\epsilon(g,h)$.
Using the associativity of automorphisms, it turns out that 
$W_g^\epsilon$ and $u_R^\epsilon(g,h)$
 satisfy some non-trivial relation (\ref{lem41eq})  Lemma \ref{lem41}, with some
 phase factor $c^\epsilon(g,h,k)\in\Uo$.
 Considering the grading of this equation, we obtain some relation between
 $b^\epsilon_g$, $\kappa_\sigma^\epsilon(g,h)$.
 Juggling the relation (\ref{lem41eq}), we obtain
 relation between $\kappa_\sigma^\epsilon(g,h)$ and $c^\epsilon(g,h,k)\in\Uo$.
This gives an element of $\PDz$.

The difference of the choice of objects we introduced to define the index
can be implemented by some unitary.
Writing such relation down brings 
us to the equivalence relation we introduced
to define $\PDz$.

As we saw in Lemma \ref{lem29}, our translation $\tau$ shifts supports of automorphism 
which is sometimes inconvenient. However, Proposition \ref{lem32}
proved in Section \ref{sdautosec} allow us to 
modify them suitably. We occasionally use this Proposition \ref{lem32}
in this section.

\subsection{Effective excitation caused by $\beta_g^U$}
In this subsection, we investigate the effective excitation caused by $\beta_g^U$
on our SPT-ground state $\omega\in \SPT$.
We denote by $\lmk\caH_\sigma,\pi_\sigma,\Omega_\sigma\rmk$
a GNS triple of $\omega_{\sigma}$
for $\sigma=L,R$.
Because $\omega_{\sigma}$ is homogeneous,
there is a self-adjoint unitary $\Gamma_\sigma$ on $\caH_\sigma$
such that $\Gamma_\sigma \pi_\sigma(A)\Omega_\sigma=\pi_\sigma
\lmk\Theta_\sigma(A)\rmk\Omega_\sigma$ for all $A\in\al_{H_\sigma}$.
From (\ref{zlrten}), $(\caH_L\otimes \caH_R, \pi_L\hat\otimes\pi_R,
\Omega_L\otimes \Omega_R)$
is a GNS representation of $\omega^{(0)}$.

\begin{lem}\label{lem22}
For $\omega\in\SPT$ $\alpha\in \QAut(\caA)$, 
$0<\theta<\theta'<\frac\pi 2$, $(\alpha_U,\alpha_D,\Xi_L,\Xi_R)\in
\caD^H(\alpha,\theta)$ and $g\in G$,
there are homogeneous states 
$\varphi_{L}$
$\varphi_{R}$
on $\caA_{H_L\cap H_U\cap C_{\theta'}}$,
$\caA_{H_R\cap H_U\cap C_{\theta'}}$
such that
\begin{align}
\omega_{p_{H_U}}\alpha_U\beta_g^{U}\alpha_U^{-1}
\qe
\varphi_L\hat\otimes\varphi_R\hat\otimes 
\omega_{p_{H_U\cap C_{\theta'}^c}},\label{ichi-22}\\
\left.\omega_{p_{H_U}}\alpha_U\beta_g^{U}\alpha_U^{-1}\right\vert_{\caA_{H_U\cap C_{\theta'}}}
\qe \varphi_L\hat\otimes\varphi_R,\label{ni-22}\\
\left.\omega_{p_{H_U}}\alpha_U\beta_g^{U}\alpha_U^{-1}\right\vert_{\caA_{H_U\cap C_{\theta'}^c}}
\qe
\omega_{p_{H_U\cap C_{\theta'}^c}}.\label{san-22}
\end{align}
\end{lem}
\begin{proof}
Set $c_U:=1$ and $c_D:=-1$.
Because $\alpha_\zeta\in \QAut(\caA_{H_\zeta^{c_\zeta}})$
for $\zeta=U,D$ 
we can decompose $\alpha_\zeta^{-1}$ as
\begin{align}\label{dea}
\alpha_\zeta^{-1}
=\lmk\alpha_{C_\theta,L,\zeta}\hat\otimes \alpha_{C_\theta,R,\zeta}
\rmk\circ
\alpha_{C_\theta^c,\zeta}\circ
\lmk \alpha_{\partial C_\theta, L,\zeta}\hat\otimes
\alpha_{\partial C_\theta, R,\zeta}
\rmk\circ\inn,
\end{align}
where
\begin{align}\label{supt}
\alpha_{C_\theta,\sigma,\zeta}\in\QAut
\lmk\caA_{C_\theta\cap H_\sigma\cap H_\zeta}\rmk,\quad
\alpha_{C_\theta^c,\zeta}
\in \QAut\lmk\caA_{C_\theta^c\cap H_\zeta}\rmk,\quad
\alpha_{\partial C_\theta, \sigma,\zeta}
\in \QAut\lmk \caA_{C_{\theta'}\cap H_\sigma\cap H_\zeta}\rmk,
\end{align}
for $\sigma=L,R$, $\zeta=U,D$.
Using the decomposition (\ref{dea})
and the support of the automorphisms there in (\ref{supt}),
we rewrite $\alpha\beta_g\alpha^{-1}$
as
\begin{align}\label{lem2228}
\alpha\beta_g\alpha^{-1}
=\lmk
\alpha_U\hat\otimes \alpha_D
\rmk\circ \beta_g
\lmk
\alpha_U\hat\otimes \alpha_D
\rmk^{-1}
\lmk
Y_{L}\hat\otimes Y_R 
\rmk\circ\inn
\end{align}
with
\begin{align}
\begin{split}
Y_\sigma:=&
\lmk \alpha_{\partial C_\theta, \sigma,U}^{-1}
\alpha_{C_\theta, \sigma,U}^{-1}\rmk
\hat\otimes 
\lmk \alpha_{\partial C_\theta, \sigma,D}^{-1}
\alpha_{C_\theta, \sigma,D}^{-1}\rmk
\lmk \beta_g^{H_\sigma}\rmk^{-1}
\Xi_\sigma \beta_g^{H_\sigma}
\Xi_\sigma^{-1}
\lmk \alpha_{C_\theta, \sigma,U}
\alpha_{\partial C_\theta, \sigma,U}
\rmk
\hat\otimes 
\lmk \alpha_{C_\theta, \sigma,D}\alpha_{\partial C_\theta, \sigma,D}
\rmk\\
&\in \QAut
\lmk\caA_{C_{\theta'}\cap H_\sigma}\rmk,\quad \sigma=L,R.
\end{split}
\end{align}
Substituting this, from the $\beta_g$-invariance of $\omega$, we have
\begin{align}
\begin{split}
\omega^{(0)}
=\omega^{(0)}\alpha\beta_g\alpha^{-1}
=\omega^{(0)} \lmk
\alpha_U\hat\otimes \alpha_D
\rmk\circ \beta_g
\lmk
\alpha_U\hat\otimes \alpha_D
\rmk^{-1}
\lmk
Y_{L}\hat \otimes Y_R 
\rmk\circ\inn.
\end{split}
\end{align}
From this, we have
\begin{align}\label{hajm}
\begin{split}
\lmk \omega_{p_{H_U}}\alpha_U\beta_g^{U}\alpha_U^{-1}\rmk
\hat\otimes 
\lmk \omega_{p_{H_D}}\alpha_D\beta_g^D\alpha_D^{-1}\rmk
\simeq
\lmk \omega_{p_{H_L \cap C_{\theta'}}}{Y_{L}'}^{-1}\rmk
\hat \otimes 
\lmk \omega_{p_{H_R\cap C_{\theta'}}}{Y_R'}^{-1} \rmk
\hat\otimes 
\omega_{p_{ C_{\theta'}^c}}
.
\end{split}
\end{align}
Here, $Y'_{\sigma}\in \Aut^{(0)}\lmk\caA_{H_\sigma\cap C_{\theta'}}\rmk$
are automorphisms 
such that
\begin{align}
{Y_{L}'}^{-1}\hat \otimes {Y_R'}^{-1} 
=\lmk
Y_{L}\hat \otimes Y_R 
\rmk^{-1}
\lmk\alpha_U\hat\otimes \alpha_D\rmk
\circ\lmk \beta_g^{\bbZ\times\{0\}}\rmk^{-1}
\lmk\alpha_U\hat\otimes \alpha_D\rmk^{-1}\inn
\end{align}
Note that
$\varphi_L:=\omega_{p_{H_L \cap C_{\theta'}}}{Y_{L}'}^{-1}\vert_{\caA_{H_U\cap H_L\cap C_{\theta'}}}$
$\varphi_R:=\omega_{p_{H_R \cap C_{\theta'}}}{Y_{R}'}^{-1}\vert_{\caA_{H_U\cap H_R\cap C_{\theta'}}}$
are homogeneous states on $\caA_{H_L \cap H_U\cap C_{\theta'}}$,
$\caA_{H_R \cap H_U\cap C_{\theta'}}$ respectively.

Note that both sides of (\ref{hajm})
are pure and homogeneous state on $\al$.
Restricting (\ref{hajm}) to $\caA_{H_U}$,
from Lemma \ref{lem21} 
we obtain (\ref{ichi-22}).
If we restrict (\ref{hajm}) to $\caA_{H_U\cap C_{\theta'}}$,
then we have
(\ref{ni-22})
from Lemma \ref{lem21}.
If we restrict (\ref{hajm}) to $\caA_{H_U\cap C_{\theta'}^c}$,
then we have
(\ref{san-22})
from Lemma \ref{lem21}.
\end{proof}

The following Proposition gives the effective excitation caused by $\beta_g^U$ on $\omega\in\SPT$.
\begin{prop}\label{lem24}
Let $\omega\in \SPT$.
Then there is a group homomorpshims
$a_\omega : G\to \{0,1\}=\bbZ_2$
satisfying the following.
\begin{description}
\item[(i)]
For any $0<\theta<\frac\pi2$, and $g\in G$
there are
$\eta_{g,L}\in\Aut^{(0)}\lmk\al_{H_L^{c_L}\cap C_\theta}\rmk$,
$\eta_{g,R}\in\Aut^{(0)}\lmk\al_{H_R^{c_R}\cap C_\theta}\rmk$
such that
\begin{align}\label{efac24}
\omega\circ\beta_g^U\simeq
\omega\circ\tau^{a_\omega(g)}\lmk\eta_{g,L}\hat\otimes \eta_{g,R}\rmk.
\end{align}
\item[(ii)]
If there is $g\in G$, $0<\theta<\frac\pi 2$,
$a'\in\{0,1\}$ and 
$\eta_{g,L}'\in\Aut^{(0)}\lmk\al_{H_L^{c_L}\cap C_\theta}\rmk$,
$\eta_{g,R}'\in\Aut^{(0)}\lmk\al_{H_R^{c_R}\cap C_\theta}\rmk$
such that 
\begin{align}
\omega\circ\beta_g^U\simeq
\omega\circ\tau^{a'}\lmk\eta_{g,L}'\hat\otimes \eta_{g,R}'\rmk,
\end{align}
then $a'=a_\omega(g)$.
\end{description}
\end{prop}
\begin{defn}\label{aodef}
Proposition \ref{lem24}
defines
a $*$-homomorphism
$a_\omega : G\to \{0,1\}=\bbZ_2$
for each $\omega\in \SPT$.
We will use this symbol henceforth.
\end{defn}
\begin{proof}
\underline{\it Existence}:  
 Fix any $0<\theta<\frac\pi2$. We show that for each
 $g\in G$,
 there are $a_\omega(g)=0,1$,
$\eta_{g,L}\in\Aut^{(0)}\lmk\al_{H_L^{c_L}\cap C_\theta}\rmk$,
$\eta_{g,R}\in\Aut^{(0)}\lmk\al_{H_R^{c_R}\cap C_\theta}\rmk$
 satisfying (\ref{efac24}).
Choose some $0<\theta'<\theta''<\theta$,
 $\alpha\in \EAut(\omega)$,
and $(\alpha_U,\alpha_D,\Xi_L,\Xi_R)\in \caD^H(\alpha,\theta')$.
Then by Lemma \ref{lem22},
there are homogeneous states $\varphi_L$, $\varphi_R$
on
$\al_{H_U\cap H_L\cap C_{{\theta''}}}$, 
$\al_{H_U\cap H_R\cap C_{{\theta''}}}$
such that 
\begin{align}
\omega_{p_{H_U}}\alpha_U\beta_g^{U}\alpha_U^{-1}
\qe
\varphi_L\hat\otimes\varphi_R\hat\otimes 
\omega_{p_{H_U\cap C_{{\theta''}}^c}},\label{ichi-24}\\
\left.\omega_{p_{H_U}}\alpha_U\beta_g^{U}\alpha_U^{-1}\right\vert_{\caA_{H_U\cap C_{{\theta''}}}}
\qe \varphi_L\hat\otimes\varphi_R,\label{ni-24}\\
\left.\omega_{p_{H_U}}\alpha_U\beta_g^{U}\alpha_U^{-1}\right\vert_{\caA_{H_U\cap C_{{\theta''}}^c}}
\qe
\omega_{p_{H_U\cap C_{{\theta''}}^c}}.\label{san-24}
\end{align}
Note that 
$\mkh_{H_U\cap C_{{\theta''}}}$,
$\mkh_{H_U\cap C_{{\theta''}}^c}$,
$\omega_{p_{H_U}}\alpha_U\beta_g^{U}\alpha_U^{-1}$
satisfy the conditions of
$\mkk_1$, $\mkk_2$, $\omega$
in Lemma \ref{lem16}
respectively.
Applying Lemma \ref{lem16},
either (i) or (ii) of Lemma \ref{lem16} occurs.
From the quasi-equivalence (\ref{ichi-24}),
the von Neumann algebras
$\pi_{\omega_{p_{H_U}}\alpha_U\beta_g^{U}\alpha_U^{-1}}
\lmk\al_{H_U\cap C_{{\theta''}}^c}\rmk''$
and $\pi_{\omega_{p_{H_U\cap C_{{\theta''}}^c}}}
\lmk \al_{H_U\cap C_{{\theta''}}^c}\rmk''$
are $*$-isomorphic.
As $\pi_{\omega_{p_{H_U\cap C_{{\theta''}}^c}}}
\lmk \al_{H_U\cap C_{{\theta''}}^c}\rmk''$ is a type $I$ factor,
(because $\omega_{p_{H_U\cap C_{{\theta''}}^c}}$ is a pure state),
it means 
 $\pi_{\omega_{p_{H_U}}\alpha_U\beta_g^{U}\alpha_U^{-1}}
\lmk\al_{H_U\cap C_{{\theta''}}^c}\rmk''$
is also a type $I$ factor.
It means from Lemma \ref{lem16} that Lemma \ref{lem16} (i)
occurs.
Hence for both of  pure homogeneous states
$\omega_{p_{H_U}}\alpha_U\beta_g^{U}\alpha_U^{-1}$
and $\omega^{(0)}$,
(i) of Lemma \ref{lem16} occurs, and we may apply
Proposition \ref{mainprop}.
Applying Lemma \ref{mainprop},
there are automorphisms $S\in \Aut^{(0)}\lmk\al_{H_U\cap C_{{\theta''}}}\rmk$, $T\in \Aut^{(0)}\lmk\al_{H_U\cap C_{{\theta''}}^c}\rmk$
such that
\begin{align}\label{st24}
\omega_{p_{H_U}}\alpha_U\beta_g^{U}\alpha_U^{-1}
\simeq
\lmk \omega_{p_{H_U\cap C_{\theta''}}}S\rmk\hat\otimes
\lmk \omega_{p_{H_U}\cap C_{{\theta''}}^c}T\rmk.
\end{align}
Note that both sides are pure homogeneous states.
Therefore, applying Lemma \ref{lem26}, and (\ref{san-24}),
we obtain
\begin{align}
\omega_{p_{H_U\cap C_{{\theta''}}^c}}
\qe
\left.\omega_{p_{H_U}}\alpha_U\beta_g^{U}\alpha_U^{-1}\right\vert_{\caA_{H_U\cap C_{{\theta''}}^c}}
\qe
\omega_{p_{H_U}\cap C_{{\theta''}}^c}T.
\end{align}
Because both of $\omega_{p_{H_U\cap C_{{\theta''}}^c}}
$ and $\omega_{p_{H_U}\cap C_{{\theta''}}^c}T$
are pure, we conclude that they are equvalent.
Substituting this to (\ref{st24}),
we obtain
\begin{align}\label{atherome}
\omega_{p_{H_U}}\alpha_U\beta_g^{U}\alpha_U^{-1}
\simeq
\lmk \omega_{p_{H_U\cap C_{\theta''}}}S\rmk\hat\otimes
\lmk \omega_{p_{H_U}\cap C_{{\theta''}}^c}T\rmk
\simeq 
\lmk \omega_{p_{H_U\cap C_{\theta''}}}S\rmk\hat\otimes
\lmk \omega_{p_{H_U}\cap C_{{\theta''}}^c}\rmk.
\end{align}
Because both sides  of this equation are homogeneous and
pure, applying Lemma \ref{lem21}, and combining it with
(\ref{ni-24})
we obtain
\begin{align}
\varphi_L\hat\otimes\varphi_R\qe
\left.\omega_{p_{H_U}}\alpha_U\beta_g^{U}\alpha_U^{-1}\right\vert_{\caA_{H_U\cap C_{{\theta''}}}}
\qe
\omega_{p_{H_U\cap C_{\theta''}}}S.
\end{align}
Hence 
$\mkh_{H_U\cap H_L\cap C_{{\theta''}}}$, 
$\mkh_{H_U\cap H_R\cap C_{{\theta''}}}$,
 $\omega_{p_{H_U\cap C_{\theta''}}}S$ satisfies
the condition of $\mkK_1$,$\mkK_2$,
$\omega$ in Lemma \ref{lem16} holds and
(i) or (ii) of Lemma \ref{lem16} holds for $\omega_{p_{H_U\cap C_{\theta''}}}S$.

If 
$\omega_{p_{H_U\cap C_{\theta''}}}S$ satisfies (i) of Lemma \ref{lem16},
 we set $a_\omega(g):=0$, and
 if $\omega_{p_{H_U\cap C_{\theta''}}}S$ satisfies (ii) of Lemma \ref{lem16},
 we set $a_\omega(g):=1$.
 
 If $a_\omega(g)=0$, because $\omega_{p_{H_U\cap C_{\theta''}}}$ also satisfies
(i) of Lemma \ref{lem16},
from Proposition \ref{mainprop}, there are
$\tilde \eta_{g,L}\in\Aut^{(0)}\lmk\al_{H_U\cap H_L^{c_L}\cap C_{\theta''}}\rmk$,
$\tilde \eta_{g,R}\in\Aut^{(0)}\lmk\al_{H_U\cap H_R^{c_R}\cap C_{\theta''}}\rmk$
such that
\begin{align}
\omega_{p_{H_U\cap C_{\theta''}}}S\simeq 
\omega_{p_{H_U\cap C_{\theta''}}}\circ\lmk\tilde\eta_{g,L}\hat\otimes \tilde\eta_{g,R}\rmk.
\end{align}
Here, we used Proposition \ref{lem32} in section \ref{sdautosec}
to take the support of $\tilde \eta_{g,\sigma}$ from 
$H_U\cap H_\sigma^{c_\sigma}\cap C_{\theta''}$, not just
$H_U\cap H_\sigma\cap C_{\theta''}$.

If $a_\omega(g)=1$, because $\omega_{p_{H_U\cap C_{\theta''}}}\tau$ also satisfies
(ii) of Lemma \ref{lem16},
from the same argument as the last paragraph of Section \ref{splitsec}, there are
$\tilde \eta_{g,L}\in\Aut^{(0)}\lmk\al_{H_U\cap H_L^{c_L}\cap C_{\theta''}}\rmk$,
$\tilde \eta_{g,R}\in\Aut^{(0)}\lmk\al_{H_U\cap H_R^{c_R}\cap C_{\theta''}}\rmk$
such that
\begin{align}
\omega_{p_{H_U\cap C_{\theta''}}}S\simeq 
\omega_{p_{H_U\cap C_{\theta''}}}\circ\tau\circ\lmk\tilde\eta_{g,L}\hat\otimes \tilde\eta_{g,R}\rmk.
\end{align}
(Again we used Proposition \ref{mainprop} and 
Proposition \ref{lem32}.)

Combining these with (\ref{atherome}), we obtain
\begin{align}\label{atherome1}
\omega_{p_{H_U}}\alpha_U\beta_g^{U}\alpha_U^{-1}
\simeq 
\omega_{p_{H_U\cap C_{\theta''}}}\circ\tau^{a_\omega(g)}
\circ\lmk\tilde\eta_{g,L}\hat\otimes \tilde\eta_{g,R}\rmk
\hat\otimes
\lmk \omega_{p_{H_U}\cap C_{{{\theta''}}}^c}\rmk
=\omega_{p_{H_U}}\circ\tau^{a_\omega(g)}
\circ\lmk\tilde \eta_{g,L}\hat\otimes \tilde\eta_{g,R}\rmk.
\end{align}

As in the proof of Lemma \ref{lem22} (\ref{lem2228}),
there are $Z_{\sigma}\in \Aut^{(0)}\lmk\caA_{H_\sigma\cap C_{\theta''}}\rmk$,
$\sigma=L,R$
such that 
\begin{align}
\alpha\beta_g^U\alpha^{-1}
=
\alpha_U
 \beta_g^U
\alpha_U^{-1}
\circ \lmk
Z_{L}\hat \otimes Z_R 
\rmk\circ\inn.
\end{align}
From this and (\ref{atherome1}), we have
\begin{align}\label{nnn24}
\begin{split}
&\omega^{(0)}\alpha\beta_g^U\alpha^{-1}
\simeq 
\lmk\omega_{p_{H_D}}\hat\otimes \omega_{p_{H_U}}
\alpha_U
 \beta_g^U
\alpha_U^{-1}\rmk
\circ \lmk
Z_{L}\hat \otimes Z_R 
\rmk\simeq
\omega^{(0)}
\tau^{a_\omega(g)}
\lmk\tilde\eta_{g,L}Z_{L}\hat\otimes \tilde\eta_{g,R} Z_R \rmk\\
&=
\omega^{(0)}\alpha\circ\alpha^{-1}
\tau^{a_\omega(g)}
\lmk\tilde\eta_{g,L}Z_{L}\hat\otimes \tilde\eta_{g,R} Z_R \rmk
\alpha\circ\alpha^{-1}.
\end{split}
\end{align}
Note from Lemma \ref{lem29}
that there are automorphisms
$\eta_{g,L}\in\Aut^{(0)}\lmk\caA_{C_{\theta}\cap H_L}\rmk$,
$\eta_{g,R}\in\Aut^{(0)}\lmk\caA_{C_{\theta}\cap H_R}\rmk$
such that
\begin{align}
\alpha^{-1}
\tau^{a_\omega(g)}
\lmk\tilde\eta_{g,L}Z_{L}\hat\otimes \tilde\eta_{g,R} Z_R \rmk
\alpha
=\tau^{a_\omega(g)} \lmk\eta_{g,L}\hat\otimes \eta_{g,R}\rmk\circ\inn.
\end{align}
From Lemma \ref{lem32}, we may assume that
$\eta_{g,L}\in\Aut^{(0)}\lmk\al_{H_L^{c_L}\cap C_\theta}\rmk$,
$\eta_{g,R}\in\Aut^{(0)}\lmk\al_{H_R^{c_R}\cap C_\theta}\rmk$.
Substituting this to (\ref{nnn24}),
we obtain
\begin{align}
\begin{split}
\omega\beta_g^U=
\omega^{(0)}\alpha\beta_g^U
\simeq
\omega^{(0)}\alpha
\tau^{a_\omega(g)} \lmk\eta_{g,L}\hat\otimes \eta_{g,R}\rmk
=\omega\tau^{a_\omega(g)} \lmk\eta_{g,L}\hat\otimes \eta_{g,R}\rmk.
\end{split}
\end{align}
\underline{\it Uniqueness :} 
Suppose for $i=1,2$, there are
$0<\theta_i<\frac\pi2$, $a_1(g), a_2(g)=0,1$,
$\eta_{g,L}^{(i)}\in\Aut^{(0)}\lmk\al_{H_L^{c_L}\cap C_{\theta_i}}\rmk$,
$\eta_{g,R}^{(i)}\in\Aut^{(0)}\lmk\al_{H_R^{c_R}\cap C_{\theta_i}}\rmk$
 satisfying 
 \begin{align}\label{Helsinki}
 \omega\beta_g^U\simeq 
 \omega\tau^{a_1}\lmk \eta_{g,L}^{(1)}\hat\otimes \eta_{g,R}^{(1)}\rmk
 \simeq
 \omega\tau^{a_2}\lmk \eta_{g,L}^{(2)}\hat\otimes \eta_{g,R}^{(2)}\rmk.
  \end{align}
  We set $\theta:=\max\{\theta_1,\theta_2\}$.
%
 Let $\alpha\in \EAut(\omega)$ and $\theta<\theta'<\frac \pi 2$.
 Then from (\ref{Helsinki}), Lemma \ref{lem29}
 there are some automorphisms $\xi_{\sigma}\in \Aut^{(0)}\lmk\al_{C_{\theta'}\cap H_\sigma}\rmk$
 such that
  \begin{align}\label{bocha}
 \begin{split}
\omega^{(0)}\simeq
 \omega^{(0)}\circ\alpha \tau^{a_1}\lmk
 \eta_{g,L}^{(1)}\lmk \eta_{g,L}^{(2)}\rmk^{-1}
 \hat\otimes \eta_{g,R}^{(1)}\lmk \eta_{g,R}^{(2)}\rmk^{-1}
\rmk\tau^{-a_2}\alpha^{-1}
  \simeq \omega^{(0)}\tau^{a_1-a_2}\lmk\xi_L\otimes \xi_R\rmk.
   \end{split}
 \end{align}
 By Lemma \ref{lem16}, if $a_1-a_2=-1,1$, then
 $\omega^{(0)}\tau^{a_1-a_2}\lmk\xi_L\otimes \xi_R\rmk$
 satisfies Lemma \ref{lem16} (ii) and from Proposition \ref{mainprop} it
 cannot be equivalent to $\omega^{(0)}$.
  Therefore, from (\ref{bocha}), 
  $a_1=a_2$.\\
  \underline{Group homomorphism:}\\
 Let  $g,h\in G$ and $\eta_{g,\sigma}, \eta_{h,\sigma}\in \Aut^{(0)}\lmk \caA_{C_\theta\cap H_\sigma}\rmk$
satisfying (\ref{efac24})
for $g,h$.
Then we have
\begin{align}
\begin{split}
&\omega\beta_{gh}^U
=\omega\beta_g^U\beta_h^U
\simeq \omega\tau^{\ao g} \lmk\eta_{g,L}\hat\otimes \eta_{g,R}\rmk\beta_h^U\\
&=\omega\beta_h^U\lmk\beta_h^U\rmk^{-1}
\tau^{\ao g}\lmk\eta_{g,L}\hat\otimes \eta_{g,R}\rmk \beta_h^U\\
&\simeq 
\omega\tau^{\ao h}\lmk\eta_{h,L}\hat\otimes \eta_{h,R}\rmk
\lmk\beta_h^U\rmk^{-1}
\tau^{\ao g}\lmk\eta_{g,L}\hat\otimes \eta_{g,R}\rmk \beta_h^U\\
&\simeq\omega\tau^{\ao h+\ao g}
\lmk \eta_{gh,L}\hat\otimes \eta_{gh,R}\rmk,
\end{split}
\end{align}
with some
$\eta_{gh,\sigma}\in\Aut^{(0)}(\al_{H_\sigma^{c_\sigma}\cap C_\theta})$,
using Proposition \ref{lem32}.
If $\ao h+\ao g=0,1$, then
from the uniqueness, we have $\ao {gh}=\ao g+\ao h$.
If $\ao h+\ao g=2$, then combining with Lemma \ref{lem38},
we have
$\omega\beta_{gh}^U\simeq \omega
\lmk \tilde \eta_{gh,L}\hat\otimes \tilde \eta_{gh,R}\rmk,$
with some
$\tilde \eta_{gh,\sigma}\in\Aut^{(0)}(\al_{H_\sigma^{c_\sigma}\cap C_{\theta}})$
and we get $\ao {gh}=0$.
Hence identifying $\{0,1\}=\bbZ_2$, we see that
$a_\omega: G\to \{0,1\}=\bbZ_2$ is a group homomorphism.
\end{proof}
\begin{rem}\label{lem24'}
If $a_\omega(g)=1$, then 
from Lemma \ref{lem38},
for any $0<\theta<\frac{\pi}{2}$
we have
\begin{align}
\omega\beta_g^U
\simeq 
\omega\tau^{-1}\lmk \eta_{g,L}^{(-1)}\hat\otimes \eta_{g,R}^{(-1)}\rmk,
\end{align}
for some $\eta_{g,\sigma}^{(-1)}\in \Aut^{(0)}(\al_{H_\sigma^{c_\sigma}\cap C_{\theta}})$.
\end{rem}
Hence for each $\omega\in \SPT$, the following set is not empty.
\begin{defn}
For each $\omega\in \SPT$ and $0<\theta<\frac\pi 2$,
we set
\begin{align}
I(\omega,\theta)
:=
\left\{
\lmk\eta_{g\sigma}^\epsilon\rmk_{g\in G,\sigma=L,R,\epsilon=\pm}
\middle|
\begin{gathered}
\eta_{g,L}^\epsilon\in\Aut^{(0)}\lmk\al_{H_L^{c_L}\cap C_\theta}\rmk,\quad
\eta_{g,R}^\epsilon \in\Aut^{(0)}\lmk\al_{H_R^{c_R}\cap C_\theta}\rmk\\
\omega\beta_g^U\simeq 
\omega\tau^{a_\omega(g)\epsilon}\lmk
\eta_{g,L}^\epsilon\hat\otimes \eta_{g,R}^\epsilon
\rmk
\end{gathered}
\right\}.
\end{align}
For $\lmk\eta_{g\sigma}^\epsilon\rmk\in I(\omega,\theta)$,
we set $\eta_g^\epsilon:=\eta_{gL}^\epsilon\hat\otimes \eta_{gR}^\epsilon$.

\end{defn}
\subsection{Derivation of an element in $\PD$}
Using the elements in $I(\omega,\theta)$, we derive some element
 in $\widetilde \PD$ out of SPT-ground states.
 We set
 \begin{align}
 \begin{split}
 {c_L^{(1)}}:=-4=c_L+1,\quad{c_R^{(1)}}:=4=c_R-1,\quad
  {c_L^{(2)}}:=-3=c_L+2,\quad{c_R^{(2)}}:=3=c_R-2.
 \end{split}
  \end{align} 
 The conjugation by $\tau$ or $\tau^{-1}$
 maps an automorphism supported at $C_\theta\cap H_\sigma^{c_\sigma}$
 to an automorphism supported at $C_\theta\cap H_\sigma^{{c_\sigma^{(1)}}}$.
  The conjugation by $\tau^2$ or $\tau^{-2}$
 maps an automorphism supported at $C_\theta\cap H_\sigma^{c_\sigma}$
 to an automorphism supported at $C_\theta\cap H_\sigma^{{c_\sigma^{(2)}}}$
 (Lemma \ref{lem29}).
\begin{lem}\label{lem39}
For any $\omega\in\SPT$, $\alpha\in \EAut(\omega)$, 
$0<\theta<\frac\pi 2$, $(\eta_{g\sigma}^\epsilon)\in I(\omega,\theta)$,
$(\alpha_L,\alpha_R, \Upsilon)\in \caD^V(\alpha, \theta)$,
there are unitaries
$u_\sigma^\epsilon(g,h)\in \caU(\caH_\sigma)$,
$g,h\in G$, $\sigma=L,R$, $\epsilon=\pm 1$
such that
\begin{align}\label{u39}
\Ad\lmk u_\sigma^\epsilon (g,h)\rmk\pi_\sigma\alpha_\sigma
=\pi_\sigma\alpha_\sigma
\bgu{g}\sigma\lmk\eta_{g\sigma}^\epsilon\rmk^{-1}
\bgu{h}\sigma\lmk \lmk \eta_{h\sigma}^{\eg{g}\epsilon}\rmk_{-\ao{g}\epsilon}\rmk^{-1}
\eta_{gh,\sigma}^\epsilon \lmk \bgu{gh}\sigma\rmk^{-1}.
\end{align}
\end{lem}
\begin{proof}
Set
\begin{align}\label{zetadefs}
\zeta_{g,h,\sigma}^\epsilon:
=\bgu{g}\sigma\lmk\eta_{g\sigma}^\epsilon\rmk^{-1}
\bgu{h}\sigma\lmk \lmk \eta_{h\sigma}^{\eg{g}\epsilon}\rmk_{-\ao{g}\epsilon}\rmk^{-1}
\eta_{gh,\sigma}^\epsilon \lmk \bgu{gh}\sigma\rmk^{-1}
\in \auz
\lmk \caci{\theta}\sigma\rmk
\end{align}
for $g,h\in G$, $\sigma=L,R$, $\epsilon=\pm 1$.
Note from $(\eta_{g\sigma}^\epsilon)\in I(\omega,\theta)$
and
\begin{align}\label{aaae}
\ao{g}\epsilon+\ao{h}\eg{g}{\epsilon}=\ao{gh}\epsilon
\end{align}
that 
\begin{align}
\omega\circ\lmk \zeta_{g,h,L}^\epsilon\hat\otimes \zeta_{g,h,R}^\epsilon\rmk
\simeq \omega.
\end{align}
Substituting $\omega\simeq \omega^{(0)}\lmk\alpha_L\hat\otimes \alpha_R\rmk \Upsilon$,
from the commutativity of $\Upsilon$ and $\zeta_{g,h,\sigma}^\epsilon$,
we obtain
\begin{align}
\widehat\bigotimes_{\sigma=L,R} \omega_{p_{H_\sigma}} \alpha_\sigma\zeta_{g,h,\sigma}^\epsilon
\simeq \widehat\bigotimes_{\sigma=L,R} \omega_{p_{H_\sigma}} \alpha_\sigma.
\end{align}
Because both sides are pure homogeneous states,
 Lemma \ref{lem21} implies
\begin{align}
\omega_{p_{H_\sigma}} \alpha_\sigma\zeta_{g,h,\sigma}^\epsilon
\simeq \omega_{p_{H_\sigma}} \alpha_\sigma.
\end{align}
Hence we complete the proof.
\end{proof}

\begin{lem}\label{lem40}
For any $\omega\in\SPT$, $\alpha\in \EAut(\omega)$, 
$0<\theta<\frac\pi 2$, $(\eta_{g\sigma}^\epsilon)\in I(\omega,\theta)$,
$(\alpha_L,\alpha_R, \Upsilon)\in \caD^V(\alpha, \theta)$,
there are unitaries
$W_g^\epsilon\in \caU(\caH_L\otimes \caH_R)$
such that 
\begin{align}\label{w40}
\Ad\lmk W_g^\epsilon\rmk \circ\lmk\pi_L\alpha_L\hat \otimes \pi_R\alpha_R\rmk
=\lmk\pi_L\alpha_L\hat \otimes \pi_R\alpha_R\rmk \circ\Upsilon
\beta_g^U\lmk \eta_{gL}^\epsilon \hat\otimes \eta_{gR}^\epsilon \rmk^{-1}
\tau^{-\ao{g}\epsilon} \Upsilon^{-1}.
\end{align}
\end{lem}
\begin{proof}
This is immediate from $\omega\simeq \omega^{(0)}\lmk\alpha_L\hat\otimes \alpha_R\rmk \Upsilon$,
and  $(\eta_{g\sigma}^\epsilon)\in I(\omega,\theta)$.
\end{proof}
\begin{defn}\label{bkrdef}
For any $\omega\in\SPT$, $\alpha\in \EAut(\omega)$, 
$0<\theta<\frac\pi 2$, $(\eta_{g\sigma}^\epsilon)\in I(\omega,\theta)$,
$(\alpha_L,\alpha_R, \Upsilon)\in \caD^V(\alpha, \theta)$,
we denote by $\IP\lmk \omega,\alpha,\theta, (\eta_{g\sigma}^\epsilon), (\alpha_L,\alpha_R,\Upsilon)\rmk$
the set of $\lmk(W_g^\epsilon), (u_\sigma^\epsilon(g,h))\rmk$
satisfying (\ref{w40}) and (\ref{u39}).
From Lemma \ref{lem42}, there are $b_g^\epsilon, \kappa_\sigma^\epsilon(g,h)\in\{0,1\}$ such that 
\begin{align}
\Ad\lmk \Gamma_L\otimes \Gamma_R\rmk\lmk W_g^\epsilon\rmk=(-1)^{b_g^\epsilon}W_g^\epsilon,\quad
\Ad \Gamma_\sigma \lmk u_\sigma^\epsilon(g,h)\rmk=(-1)^{\kappa_\sigma^\epsilon(g,h)}u_\sigma^\epsilon(g,h).
\end{align}
Because $\pi_L,\pi_R$ are irreducible,  $b_g^\epsilon, \kappa_\sigma^\epsilon(g,h)\in\{0,1\}$  are independent of the choice of $\lmk(W_g^\epsilon), (u_\sigma^\epsilon(g,h))\rmk$.
(The independence of them from the choice of $\lmk(W_g^\epsilon), (u_\sigma^\epsilon(g,h))\rmk$
is because the ambiguity of $\lmk(W_g^\epsilon), (u_\sigma^\epsilon(g,h))\rmk$ are just $\Uo$-phases.
)
Define 
$b\in C^1(G, \bbZ_2\oplus \bbZ_2)$, $\kappa_\sigma\in C^2(G, \bbZ_2\oplus \bbZ_2)$ by
\begin{align}\label{bkdef}
b(g):=\lmk b_g^{+1}, b_g^{-1}\rmk,\quad
\kappa_\sigma(g,h):=\lmk \kappa_\sigma^{+1}(g,h), \kappa_\sigma^{-1}(g,h)\rmk\in \bbZ_2\oplus \bbZ_2,\quad g,h\in G.
\end{align}
We denote $b\in C^1(G, \bbZ_2\oplus \bbZ_2)$, $\kappa_\sigma\in C^2(G, \bbZ_2\oplus \bbZ_2)$
defined in (\ref{bkdef})
 by
\begin{align}
\begin{split}
b\lmk \omega, \alpha,\theta, (\eta_{g\sigma}^\epsilon),
(\alpha_L,\alpha_R, \Upsilon)
\rmk,
\kappa_\sigma\lmk \omega, \alpha,\theta, (\eta_{g\sigma}^\epsilon),
(\alpha_L,\alpha_R, \Upsilon)
\rmk,
\end{split}
\end{align}
respectively.

\end{defn}
In the following theorem, we used notation from Definition \ref{bkrdef}.
\begin{lem}\label{lem43}
For any $\omega\in\SPT$, $\alpha\in \EAut(\omega)$, 
$0<\theta<\frac\pi 2$, $(\eta_{g\sigma}^\epsilon)\in I(\omega,\theta)$,
$(\alpha_L,\alpha_R, \Upsilon)\in \caD^V(\alpha, \theta)$,
$\lmk(W_g^\epsilon), (u_\sigma^\epsilon(g,h))\rmk\in \IP\lmk \omega,\alpha,\theta, (\eta_{g\sigma}^\epsilon), (\alpha_L,\alpha_R,\Upsilon)\rmk$ the following holds.
\begin{description}
\item[(i)]
$u_\sigma^\epsilon(g,h)\in \lmk \pi_\sigma\alpha_\sigma\lmk \cacc{\theta}{\sigma}{1}\rmk\rmk'$.
\item[(ii)] For any $x\in \caB(\caH_L\otimes \caH_R)$,
\begin{align}
\Ad\lmk W_g^\epsilon W_h^{\eg{g}\epsilon} {W_{gh}^\epsilon}^*\rmk(x)
=\Ad\lmk u_L^\epsilon(g,h)\otimes u_R^\epsilon (g,h)\Gamma_R^{\kappa_L^\epsilon (g,h)}\rmk(x).
\end{align}
\item[(iii)]
\begin{align}
\Ad\lmk W_g^\epsilon\rmk
\lmk
 \bbC\unih{L}\otimes 
\lmk 
\pa{R}\lmk \cacc{\theta}R2\rmk\rmk'\rmk
\subset  \bbC\unih{L}\otimes \caB(\caH_R)
\end{align}
\item[(iv)]
\begin{align}
\Ad\lmk W_g^\epsilon W_h^{\eg g \epsilon}\rmk\lmk \unih L\otimes u_R^{\epsilon'}(k,f)\rmk
=\Ad\lmk \lmk \unih L \otimes u^\epsilon_R(g,h)\Gamma_R^{\kappa_L^\epsilon (g,h)}\rmk W_{gh}^\epsilon\rmk
\lmk \unih L\otimes u_R^{\epsilon'}(k,f)\rmk
\end{align}

\end{description}
\end{lem}
\begin{proof}
The proof  is the same as that of Lemma 2.3 of \cite{2dSPT},
using Lemma \ref{lem42}.
We just need to care about the fact that $\tau$ can shift the support of automorphism
according to Lemma \ref{lem30}.
\end{proof}

\begin{lem}\label{lem41}
Let $\omega\in\SPT$, $\alpha\in \EAut(\omega)$, 
$0<\theta<\frac\pi 2$, $(\eta_{g\sigma}^\epsilon)\in I(\omega,\theta)$,
$(\alpha_L,\alpha_R, \Upsilon)\in \caD^V(\alpha, \theta)$,
$\lmk(W_g^\epsilon), (u_\sigma^\epsilon(g,h))\rmk
\in \IP\lmk \omega,\alpha,\theta, (\eta_{g\sigma}^\epsilon), (\alpha_L,\alpha_R,\Upsilon)\rmk$.
Then there is some $\css{\sigma}{\epsilon}(g,h,k)\in \Uo$ for each $g,h,k\in G$ and $\epsilon=\pm 1$, $\sigma=L,R$
such that 
\begin{align}\label{lem41eq}
\begin{split}
&W_g^\epsilon\lmk\unih L\otimes \uss R {\eg g \epsilon}(h,k)\rmk{W_g^\epsilon}^*
\lmk \unih L\otimes \uss R \epsilon (g,hk)\rmk
=\css R\epsilon (g,h,k) \lmk\unih L \otimes \uss R\epsilon (g,h) \uss R\epsilon (gh,k)\rmk\\
&W_g^\epsilon \lmk \uss L {\eg{g}\epsilon}(h,k)\otimes \Gamma_R^{\kss L {\eg g \epsilon}(h,k)}\rmk{W_g^\epsilon}^*
\lmk \uss  L\epsilon (g,hk)\otimes \Gamma_R^{\kss L \epsilon (g,hk)}\rmk
=\css L \epsilon (g,h,k) \uss L \epsilon (g, h)\uss L \epsilon (gh, k)\otimes \Gamma_R^{\kss L \epsilon (g,h)+\kss L \epsilon (gh,k)}.
\end{split}
\end{align}
\end{lem}
\begin{defn}\label{crdef}
We define
$c\lmk \omega,\alpha,\theta, (\eta_{g\sigma}^\epsilon), (\alpha_L,\alpha_R,\Upsilon), (W_g^\epsilon), (u_\sigma^\epsilon(g,h))\rmk
\in C^3(G,\Uo\oplus \Uo)$
by a map $c: G^{\times 3}\to \Uo\oplus \Uo$ such that
\begin{align}
c(g,h,k)
:=\lmk \css R {+1} (g,h,k), \css R {-1} (g,h,k)\rmk\in \Uo\oplus \Uo,\quad g,h,k\in G
\end{align}
with $\css R \epsilon (g,h,k)$ in Lemma \ref{lem41}.
\end{defn}
\begin{proof}
The existence of $\css R\epsilon (g,h,k)\in \Uo$ satisfying  first equation of  (\ref{lem41eq}) follows from
\begin{align}
\begin{split}
&\Ad\lmk\unih L \otimes \uss R\epsilon (g,h) \uss R\epsilon (gh,k)\rmk\lmk  \pa L\hat\otimes \pa R\rmk 
\beta_{ghk}^{UR}\lmk \eta_{ghk, R}^\epsilon\rmk^{-1}\tau^{-\ao {ghk}\epsilon}\\
&=
\lmk  \pa L\hat\otimes \pa R\rmk \beta_g^{UR}\lmk \eta_{gR}^\epsilon\rmk^{-1} \tau^{-\ao g\epsilon}
 \beta_h^{UR}\lmk \eta_{hR}^{\eg g \epsilon}\rmk^{-1} \tau^{- \ao h\eg g \epsilon}
 \beta_k^{UR}\lmk \eta_{kR}^{\eg {gh} \epsilon}\rmk^{-1} \tau^{-\ao k\eg {gh} {\epsilon}}\\
 &=\Ad \lmk W_g^\epsilon\lmk\unih L\otimes \uss R {\eg g \epsilon}(h,k)\rmk{W_g^\epsilon}^*
\lmk \unih L\otimes \uss R \epsilon (g,hk)\rmk\rmk
\lmk  \pa L\hat\otimes \pa R\rmk 
\beta_{ghk}^{UR}\lmk \eta_{ghk, R}^\epsilon\rmk^{-1}\tau^{-\ao {ghk}\epsilon}.
\end{split}
\end{align}
%
\end{proof}

\begin{lem}\label{lem44}
Let  $\omega\in\SPT$, $\alpha\in \EAut(\omega)$, 
$0<\theta<\frac\pi 2$, $(\eta_{g\sigma}^\epsilon)\in I(\omega,\theta)$,
$(\alpha_L,\alpha_R, \Upsilon)\in \caD^V(\alpha, \theta)$,
and $\lmk(W_g^\epsilon), (u_\sigma^\epsilon(g,h))\rmk
\in \IP\lmk \omega,\alpha,\theta, (\eta_{g\sigma}^\epsilon), (\alpha_L,\alpha_R,\Upsilon)\rmk$.
Recall $a_\omega\in H^1(G,\bbZ_2)$ and 
\begin{align}
\begin{split}
&c\lmk \omega,\alpha,\theta, (\eta_{g\sigma}^\epsilon), (\alpha_L,\alpha_R,\Upsilon), (W_g^\epsilon), (u_\sigma^\epsilon(g,h))\rmk=: c\in
C^3(G, \Uo\oplus \Uo)\\
&b\lmk \omega, \alpha,\theta, (\eta_{g\sigma}^\epsilon),
(\alpha_L,\alpha_R, \Upsilon)
\rmk=:b\in  C^1(G, \bbZ_2\oplus \bbZ_2),\\
&\kappa_\sigma\lmk \omega, \alpha,\theta, (\eta_{g\sigma}^\epsilon),
(\alpha_L,\alpha_R, \Upsilon)
\rmk=:\kappa_\sigma \in C^2(G, \bbZ_2\oplus \bbZ_2),\quad\sigma=L,R
\end{split}
\end{align}
defined in Definition \ref{aodef} and Definition \ref{crdef},
Definition \ref{bkrdef} respectively.
Then we have
\begin{align}
\lmk
c,
\kappa_R,\kappa_L, b, a_\omega 
\rmk\in \widetilde \PD.
\end{align}
If $\lmk({W'}_g^\epsilon), ({u'}_\sigma^\epsilon(g,h))\rmk$ is another choice
from $ \IP\lmk \omega,\alpha,\theta, (\eta_{g\sigma}^\epsilon), (\alpha_L,\alpha_R,\Upsilon)\rmk$\\
and $c\lmk \omega,\alpha,\theta, (\eta_{g\sigma}^\epsilon), (\alpha_L,\alpha_R,\Upsilon), ({W'}_g^\epsilon), ({u'}_\sigma^\epsilon(g,h))\rmk=: c'$,
then we have
\begin{align}\label{cc44}
\lmk
c,
\kappa_R,\kappa_L, b, a_\omega 
\rmk
\sim_{\PD}
\lmk
c',
\kappa_R,\kappa_L, b, a_\omega 
\rmk.
\end{align}
\end{lem}
\begin{defn}
From Lemma \ref{lem44},
for each  $\omega\in\SPT$, $\alpha\in \EAut(\omega)$, 
$0<\theta<\frac\pi 2$, $(\eta_{g\sigma}^\epsilon)\in I(\omega,\theta)$,
$(\alpha_L,\alpha_R, \Upsilon)\in \caD^V(\alpha, \theta)$,
we may define
\begin{align}
h^{(1)}\lmk\omega,\alpha, \theta,(\eta_{g\sigma}^\epsilon), (\alpha_L,\alpha_R, \Upsilon)\rmk:=
\left[
\lmk c,\kappa_R,\kappa_L,b,a_\omega\rmk
\right]_{\PD}\in \PD,
\end{align}
with 
$\lmk c,\kappa_R,\kappa_L,b,a_\omega\rmk$
in
Lemma \ref{lem44}, independent of the choice of $\lmk(W_g^\epsilon), (u_\sigma^\epsilon(g,h))\rmk
\in \IP\lmk \omega,\alpha,\theta, (\eta_{g\sigma}^\epsilon), (\alpha_L,\alpha_R,\Upsilon)\rmk$.

\end{defn}
\begin{proof}

(ii) of Lemma \ref{lem43} means that $W_g^\epsilon W_h^{\eg{g}\epsilon} {W_{gh}^\epsilon}^*$ and
$u_L^\epsilon(g,h)\otimes u_R^\epsilon (g,h)\Gamma_R^{\kappa_L^\epsilon (g,h)}$ are proportional.
Taking $\Ad\lmk \Gamma_L\otimes \Gamma_R\rmk$ of them, we obtain (\ref{b49}).
Taking $\Ad\lmk \Gamma_L\otimes \Gamma_R\rmk$ of (\ref{lem41eq}), we obtain (\ref{kr49}) and (\ref{kl49}).
The condition (\ref{c49}) can be checked \che{} by the repeated use of (\ref{lem41eq})
as in the proof of Lemma 2.4 in \cite{2dSPT}. 
Hence we get $\lmk
c,
\kappa_R,\kappa_L, b, a_\omega 
\rmk\in \widetilde \PD$.

To see the last claim (\ref{cc44}), recall that
the ambiguity of $\lmk(W_g^\epsilon), (u_\sigma^\epsilon(g,h))\rmk$ are just $\Uo$-phases.
Therefore, the difference of $c$ and $c'$
is of the form $d^2_a\sigma$, with some $\sigma\in C^2(G, \Uo\oplus \Uo)$.
This proves 
 (\ref{cc44}).
\end{proof}

\subsection{Well-defined-ness of the index}
In this subsection, we show that the index we derived in the previous subsection does not depend
on the choice of the objects we used.
\begin{lem}\label{lem45}
Let  $\omega\in\SPT$, $\alpha^{(1)},\alpha^{(2)}\in \EAut(\omega)$, 
$0<\theta<\frac\pi 2$, $(\eta_{g\sigma}^\epsilon)\in I(\omega,\theta)$,
$(\alpha_L^{(i)},\alpha_R^{(i)}, \Upsilon^{(i)})\in \caD^V(\alpha^{(i)}, \theta)$, $i=1,2$.
Then we have
\begin{align}
h^{(1)}\lmk\omega,\alpha^{(1)}, \theta,(\eta_{g\sigma}^\epsilon), (\alpha_L^{(1)},\alpha_R^{(1)}, \Upsilon^{(1)})\rmk
=h^{(1)}\lmk\omega,\alpha^{(2)}, \theta,(\eta_{g\sigma}^\epsilon), (\alpha_L^{(2)},\alpha_R^{(2)}, \Upsilon^{(2)})\rmk
\in\PD.
\end{align}
\begin{defn}
From \ref{lem45},  for each 
$\omega\in\SPT$, 
$0<\theta<\frac\pi 2$, $(\eta_{g\sigma}^\epsilon)\in I(\omega,\theta)$
we may define
\begin{align}
h^{(2)}\lmk\omega, \theta,(\eta_{g\sigma}^\epsilon)\rmk
:=h^{(1)}\lmk\omega,\alpha, \theta,(\eta_{g\sigma}^\epsilon), (\alpha_L,\alpha_R, \Upsilon)\rmk\in \PD,
\end{align}
independent of the choice of $\alpha$, $(\alpha_L,\alpha_R, \Upsilon)$.
\end{defn}
\end{lem}
\begin{proof}
Because
\begin{align}
\omega^{(0)}\lmk\alpha_L^{(1)}\hat\otimes \alpha_R^{(1)}\rmk\Upsilon^{(1)}
\simeq \omega\simeq \omega^{(0)}\lmk\alpha_L^{(2)}\hat\otimes \alpha_R^{(2)}\rmk\Upsilon^{(2)},
\end{align}
there is a unitary $V\in \caU\lmk \caH_L\otimes \caH_R\rmk$
such that
\begin{align}
\Ad V\circ\ptp\circ \lmk\alpha_L^{(1)}\hat\otimes \alpha_R^{(1)}\rmk\Upsilon^{(1)}
=\ptp\circ \lmk\alpha_L^{(2)}\hat\otimes \alpha_R^{(2)}\rmk\Upsilon^{(2)}.
\end{align}
As all the automorphisms in the equations are graded, from Lemma \ref{lem42}, $V$
is graded with respect to $\Gamma_L\otimes\Gamma_R$.
Let
$\lmk(W_g^\epsilon), (u_\sigma^\epsilon(g,h))\rmk
\in \IP\lmk \omega,\alpha^{(1)},\theta, (\eta_{g\sigma}^\epsilon), (\alpha_L^{(1)},\alpha_R^{(1)},\Upsilon^{(1)})\rmk$.
As in Lemma 2.11 of \cite{2dSPT},
\begin{align}\label{w45}
\Ad\lmk V W_g^\epsilon V^*\rmk\circ\lmk\pi_L\alpha_L^{(2)}\hat\otimes \pi_R\alpha_R^{(2)}\rmk
=\lmk\pi_L\alpha_L^{(2)}\hat\otimes \pi_R\alpha_R^{(2)}\rmk
\Upsilon^{(2)} \beta_g^U {\eta_g^\epsilon}^{-1}\tau^{-a_\omega(g)\epsilon}{\Upsilon^{(2)}}^{-1}
\end{align}
We also can check
\begin{align}\label{ur45}
\begin{split}
&\Ad\lmk V\lmk\unih {L} \otimes \uss {R} {\epsilon} (g,h)\rmk V^*\rmk
\ptp\circ \lmk\alpha_L^{(2)}\hat\otimes \alpha_R^{(2)}\rmk\\
&=\pi_L\alpha_L^{(2)}\hat\otimes \pi_R\alpha_R^{(2)}\bebe{R} g h.
\end{split}
\end{align}
From Lemma \ref{lem42}, this means that there is some
$\uss R{(2) \epsilon} (g,h)\in \caU(\caH_R)$ such that
$ V\lmk\unih {L} \otimes \uss {R} {\epsilon} (g,h)\rmk V^*=\unih L\otimes \uss R{(2) \epsilon} (g,h)$.
Similarly, 
there is some
$\uss L{(2) \epsilon} (g,h)\in \caU(\caH_L)$ such that
$ V\lmk\uss {L} {\epsilon} (g,h)\otimes \Gamma_R^{\kss L\epsilon (g,h)}\rmk V^*
=\uss L{(2) \epsilon} (g,h)\otimes \Gamma_R^{\kss L\epsilon (g,h)}$.
From (\ref{w45}), (\ref{ur45}) and its analog for $\uss {L} {\epsilon} (g,h)$,
we can see that $(( V W_g^\epsilon V^*), (\uss \sigma {(2)\epsilon}))$
belongs to $\IP\lmk \omega,\alpha^{(2)},\theta, (\eta_{g\sigma}^\epsilon), (\alpha_L^{(2)},\alpha_R^{(2)},\Upsilon^{(2)})\rmk$.
As in Lemma 2.11 (2.90) of \cite{2dSPT}, one can check {\che{}} that
\begin{align}
\begin{split}
c\lmk \omega,\alpha^{(1)},\theta, (\eta_{g\sigma}^\epsilon), (\alpha_L^{(1)},\alpha_R^{(1)},\Upsilon^{(1)}), (W_g^\epsilon), (u_\sigma^\epsilon(g,h))\rmk
=
c\lmk \omega,\alpha^{(2)},\theta, (\eta_{g\sigma}^\epsilon), (\alpha_L^{(2)},\alpha_R^{(2)},\Upsilon^{(2)}), 
(( V W_g^\epsilon V^*), (\uss \sigma {(2)\epsilon}))
\rmk.
\end{split}
\end{align}
It is also clear that the grading of $V W_g^\epsilon V^*$, $\uss \sigma {(2)\epsilon}$
are equal to that of $W_g^\epsilon$,  $u_\sigma^\epsilon(g,h)$.
Hence we obtain the claim.
\end{proof}
\begin{lem}\label{lem46}
For any  $\omega\in\SPT$, 
$0<\theta<\frac\pi 2$, $(\etas g\sigma {\epsilon(1)}),( \etas g\sigma {\epsilon(2)}) \in I(\omega,\theta)$,
we have
\begin{align}
h^{(2)}\lmk\omega, \theta, (\etas g\sigma {\epsilon(1)})\rmk
=h^{(2)}\lmk\omega, \theta, (\etas g\sigma {\epsilon(2)})\rmk.
\end{align}
\end{lem}
\begin{defn}
From Lemma \ref{lem46}, for each $\omega\in\SPT$, 
$0<\theta<\frac\pi 2$,
we can define
\begin{align}
h^{(3)}(\omega,\theta):=h^{(2)}\lmk\omega, \theta, (\etas g\sigma {\epsilon})\rmk,
\end{align}
independent of the choice of $(\etas g\sigma {\epsilon})$.
\end{defn}
\begin{proof}
Fix some 
 $\alpha\in \EAut(\omega)$ and 
$(\alpha_L,\alpha_R, \Upsilon)\in \caD^V(\alpha, \theta)$.
Note that $K_{g,\sigma}^\epsilon:=\etas g \sigma {\epsilon (2)}\circ{\etas g \sigma {\epsilon (1)}}^{-1} 
\in \auz\lmk\cac   \theta {\sigma}\rmk$.
Set $K_{g}^\epsilon=K_{g,L}^\epsilon\hat\otimes K_{g,R}^\epsilon$.
Because
$\omega\beta_g^U\simeq \omega\tau^{\ao g \epsilon}\etas g{}{\epsilon (1)}\simeq \omega\tau^{\ao g \epsilon}\etas g{}{\epsilon (2)}$,
we have $\omega \lmk K_{g}^\epsilon\rmk_{\ao g\epsilon}\simeq \omega$.
On the other hand, we have
$\omega\simeq \omega^{(0)}\lmk\alpha_L\hat\otimes \alpha_R\rmk\Upsilon$.
Combining these, we obtain 
$\omega^{(0)}\lmk\alpha_L\hat\otimes \alpha_R\rmk\Upsilon\lmk K_{g}^\epsilon\rmk_{\ao g\epsilon}
\simeq \omega^{(0)}\lmk\alpha_L\hat\otimes \alpha_R\rmk\Upsilon$.
Because $K_{g}^\epsilon\in \auz\lmk\caA_{C_\theta}\rmk$ and $\Upsilon$ commute,
we then obtain 
\[
\omega^{(0)}\lmk\alpha_L\hat\otimes \alpha_R\rmk\lmk K_{g}^\epsilon\rmk_{\ao g\epsilon}
\simeq \omega^{(0)}\lmk\alpha_L\hat\otimes \alpha_R\rmk.
\]
Because both sides of this equation are pure and homogeneous,
from Lemma \ref{lem21}, we have
$
\os \sigma \alpha_\sigma  \lmk K_{g\sigma}^\epsilon\rmk_{\ao g\epsilon}\simeq \os \sigma \alpha_\sigma .
$
Hence there is a unitary $V_{g\sigma}^\epsilon$ on $\caH_\sigma$
such that 
\begin{align}
\Ad\lmk V_{g\sigma}^\epsilon\rmk\pi_\sigma
=\pi_\sigma \alpha_\sigma  \lmk K_{g\sigma}^\epsilon\rmk_{\ao g\epsilon} \alpha_\sigma^{-1}.
\end{align}
It is homogeneous because of Lemma \ref{lem42} and
we denote by $\mss g \sigma \epsilon$ the grading of $V_{g\sigma}^\epsilon$,
i.e., $\Ad\Gamma_\sigma\lmk V_{g\sigma}^\epsilon\rmk=(-1)^{\mss g \sigma \epsilon}V_{g\sigma}^\epsilon$.
Fix some $\lmk(W_g^\epsilon), (u_\sigma^\epsilon(g,h))\rmk$
in $\IP\lmk \omega,\alpha,\theta, (\eta_{g\sigma}^{\epsilon (1)}), (\alpha_L,\alpha_R,\Upsilon)\rmk$.
We denote by $\kss \sigma \epsilon(g,h)$ the grading of $ u_\sigma^\epsilon(g,h)$.
Because $\lmk K_{hR}^{\epsilon'}\rmk_{a(h)\epsilon'-a(g)\epsilon}\in \auz\lmk \caA_{C_\theta\cap H_R}\rmk$
and $\Upsilon$ commute, one can show that
\begin{align}\label{ichi46}
\begin{split}
\Ad\lmk W_g^\epsilon \lmk\unih L\otimes V_{hR}^{\epsilon'}\rmk {W_g^\epsilon}^*\rmk\pat 
=\pa L\hat\otimes\pa R \beta_g^{UR} {\etas g R {\epsilon (1)}}^{-1} \lmk K_{hR}^{\epsilon'}\rmk_{a(h)\epsilon'-a(g)\epsilon}
\etas g R{\epsilon (1)}{\beta_g^{UR}}^{-1}.
\end{split}
\end{align}
From this, by Lemma \ref{lem42}, there are unitaries $x_{g,h,R}^{\epsilon,\epsilon'}\in \caU(\caH_R)$
such that
\begin{align}\label{ni46}
W_g^\epsilon \lmk\unih L\otimes V_{hR}^{\epsilon'}\rmk {W_g^\epsilon}^*
=\unih L\otimes x_{g,h,R}^{\epsilon,\epsilon'}.
\end{align}
We also set $v_{gR}^{\epsilon}:=x_{g,g,R}^{\epsilon,\epsilon}$.
Note that this unitary $v_{gR}^{\epsilon}$ have the grading $\mss g R \epsilon$
as $V_{gR}^\epsilon$
with respect to $\Ad \Gamma_R$.
Because $ \beta_g^{UR} {\etas g R {\epsilon (1)}}^{-1} \lmk K_{hR}^{\epsilon'}\rmk_{a(h)\epsilon'-a(g)\epsilon}
\etas g R{\epsilon (1)}{\beta_g^{UR}}^{-1}$ belongs to
$\auz\lmk\caA_{C_\theta\cap H_R^{c_R^{(2)}}}\rmk$,
we have
\begin{align}
 x_{g,h,R}^{\epsilon,\epsilon'}\in 
 \pa R\lmk \cacc \theta R {2}\rmk'.
\end{align}
From Lemma \ref{lem43} (iii), 
we have 
\begin{align}
\Ad\lmk W_k\rmk\lmk \unih L\otimes x_{g,h,R}^{\epsilon,\epsilon'}\rmk
\in \unih L\otimes \caB(\caH_R).
\end{align}
Therefore, from Lemma \ref{lem43} (ii), we have
\begin{align}\label{san46}
\Ad\lmk W_k^\epsilon W_f^{\eg k \epsilon}\rmk\lmk
\unih L\otimes  x_{g,h,R}^{\epsilon,\epsilon'}
\rmk
=
\Ad\lmk
\unih L\otimes  \uss R \epsilon (k,f)\Gamma_R^{\kss L \epsilon (k.f))}
\rmk\Ad\lmk W_{kf}^\epsilon\rmk\lmk
\unih L\otimes  x_{g,h,R}^{\epsilon,\epsilon'}
\rmk.
\end{align}

Similarly, recalling Lemma \ref{lem42}, we have
\begin{align}
\begin{split}
\Ad\lmk W_g^\epsilon \lmk V_{hL}^{\epsilon'}\otimes \Gamma_R^{\mss h L {\epsilon'}}\rmk {W_g^\epsilon}^*\rmk\pat 
=\pa L
\beta_g^{UL} {\etas g L {\epsilon (1)}}^{-1} \lmk K_{hL}^{\epsilon'}\rmk_{a(h)\epsilon'-a(g)\epsilon}
\etas g L{\epsilon (1)}{\beta_g^{UL}}^{-1}
\hat\otimes\pa R .
\end{split}
\end{align}
Furthermore, by the same argument as (\ref{ichi46}), one can also check 
\begin{align}
\begin{split}
\Ad\lmk W_g^\epsilon \lmk V_{gL}^{\epsilon}\otimes V_{gR}^{\epsilon} \Gamma_R^{\mss g L {\epsilon}}\rmk^*\rmk\pat
=\pat \Upsilon \beta_g^U {\etas g{}{\epsilon (2)}}^{-1}\tau^{-\ao g\epsilon}\Upsilon^{-1}
\end{split}
\end{align}
Using (\ref{aaae}), we have
\begin{align}
\begin{split}
&\bebey\sigma g h{(2)}\\
&=\bgu g \sigma { \etas g\sigma{\epsilon (1)}}^{-1}
{K_{g\sigma}^{\epsilon}}^{-1} \etas g \sigma {\epsilon(1)}{\bgu g \sigma}^{-1}\\
&\bebey\sigma g h{(1)}\\
&\bgu {gh}\sigma  {\etas{gh}\sigma {\epsilon (1)}}^{-1}
\lmk \lmk K_{h\sigma}^{\eg  g \epsilon}\rmk_{-\ao g \epsilon}\rmk^{-1}
K_{gh,\sigma}^\epsilon \etas {gh} \sigma {\epsilon(1)}
{\bgu {gh}\sigma}^{-1}.
\end{split}
\end{align}
Note from (\ref{ichi46}) that each lines on the right hand side can be implemented by some unitary
in $\pat$,
and we get
\newcommand{\wvw}[4]{W_{#1}^{#4} \lmk\unih L\otimes V_{#2 R}^{#3}\rmk {{W_{#1}}^{#4}}^*}
\newcommand{\wvws}[4]{W_{#1}^{#4} \lmk\unih L\otimes {V_{#2 R}^{#3}}^*\rmk {{W_{#1}}^{#4}}^*}
\begin{align}
\begin{split}
&\pa L\hat\otimes \pa R\bebey R g h{(2)}\\
&=
\Ad\lmk
\begin{gathered}
\wvws g g \epsilon  \epsilon \lmk\unih L\otimes \uss R\epsilon(g,h)\rmk\\
\wvws {gh} h {\eg g \epsilon} \epsilon
\wvw {gh}{gh} \epsilon \epsilon
\end{gathered}
\rmk
\pat.
\end{split}
\end{align}
By Lemma \ref{lem42}, this means that there is a homogeneous unitary
$\uss R {\epsilon(2)}(g,h)$ on $\caH_R$
such that
\begin{align}
\Ad\lmk\uss R {\epsilon(2)}(g,h) \rmk\pa R=\pa R\bebey R g h{(2)},
\end{align}
and 
\begin{align}\label{go46}
\begin{split}
&\wvws g g \epsilon  \epsilon \lmk\unih L\otimes \uss R\epsilon(g,h)\rmk
\wvws {gh} h {\eg g \epsilon} \epsilon
\wvw {gh}{gh} \epsilon \epsilon\\
&=\unih L\otimes \uss R {\epsilon(2)}(g,h).
\end{split}
\end{align}
\newcommand{\vs}[3]{v_{#1#2}^{#3}}
Using $v_{gR}^{\epsilon}:=x_{g,g,R}^{\epsilon,\epsilon}$ from (\ref{ni46}),
and (\ref{san46}),
the latter equation can be written
\begin{align}\label{yon46}
\begin{split}
&\unih L\otimes \uss R {\epsilon(2)}(g,h)\\
&=(-1)^{\kss L \epsilon (g,h)\kss R \epsilon (g,h)}
\lmk\unih L\otimes {\vs g R \epsilon}^*\Gamma_R^{\kss L \epsilon (g,h))}\rmk\Ad\lmk
\lmk \unih L\otimes  \uss R \epsilon (g,h)\Gamma_R^{\kss L \epsilon (g,h))}\rmk
W_{gh}^\epsilon
\rmk
\lmk
\unih L\otimes {V_{hR}^{\eg g \epsilon}}^*
\rmk\\
&\quad
\lmk \unih L\otimes  \uss R \epsilon (g,h)\Gamma_R^{\kss L \epsilon (g,h))}\rmk
\lmk
\unih L\otimes \vs {gh} R \epsilon
\rmk\\
&=(-1)^{\kss L \epsilon (g,h)\kss R \epsilon (g,h)}
\lmk\unih L\otimes {\vs g R \epsilon}^*\Gamma_R^{\kss L \epsilon (g,h))}\rmk
\Ad\lmk
W_g^\epsilon W_{h}^{\eg g\epsilon}
\rmk
\lmk
\unih L\otimes {V_{hR}^{\eg g \epsilon}}^*
\rmk\\
&\quad
\lmk \unih L\otimes  \uss R \epsilon (g,h)\Gamma_R^{\kss L \epsilon (g,h))}\rmk
\lmk
\unih L\otimes \vs {gh} R \epsilon
\rmk\\
&=
(-1)^{\kss L \epsilon (g,h) \mss h R {\eg g \epsilon}}
\lmk\unih L\otimes {\vs g R \epsilon}^*
\rmk
\Ad\lmk
W_g^\epsilon
\rmk
\lmk
\unih L\otimes {v_{hR}^{\eg g \epsilon}}^*
\rmk\\
&\quad
\lmk \unih L\otimes  \uss R \epsilon (g,h)
\rmk
\lmk
\unih L\otimes \vs {gh} R \epsilon
\rmk.
\end{split}
\end{align}

Similarly we get unitaries $\uss L\epsilon (g,h)$ on $\caH_L$, 
satisfying
\begin{align}
\Ad\lmk\uss L {\epsilon(2)}(g,h) \rmk\pa L=\pa L\bebey L g h{(2)},
\end{align}
and 
\newcommand{\wvwl}[4]{W_{#1}^{#4} \lmk V_{#2 L}^{#3} \otimes \Gamma_R^{\mss{ #2} L {#3}}\rmk {{W_{#1}}^{#4}}^*}
\newcommand{\wvwsl}[4]{W_{#1}^{#4} \lmk {V_{#2 L}^{#3}}^* \otimes \Gamma_R^{\mss {#2} L {#3}}\rmk {{W_{#1}}^{#4}}^*}
\begin{align}\label{roku46}
\begin{split}
&\wvwsl g g \epsilon  \epsilon \lmk\uss L\epsilon(g,h)\otimes \Gamma_R^{\kss L\epsilon (g,h)}\rmk
\wvwsl {gh} h {\eg g \epsilon} \epsilon
\wvwl {gh}{gh} \epsilon \epsilon\\
&=\unih L\otimes \uss L {\epsilon(2)}(g,h).
\end{split}
\end{align}
Hence we obtain
\begin{align}\label{kyu46}
\lmk
\lmk
W_g^{\epsilon (2)}:=W_g^{\epsilon} \lmk V_{gL}^\epsilon\otimes V_{gR}^\epsilon\Gamma_R^{\mss g L \epsilon}\rmk^*
\rmk,
\lmk \uss \sigma {\epsilon (2)}(g,h)\rmk
\rmk
\in \IP \lmk\omega,\alpha,\theta,(\etas g\sigma {\epsilon(2)}),
(\alpha_L,\alpha_R,\Upsilon)\rmk.
\end{align}

\newcommand{\bgn}[3]{b^{#1(#2)}(#3)}
\newcommand{\kssn}[3]{\kss #1 {#2(#3)}}
We set 
\begin{align}
\begin{split}
b^{(i)}:=
b\lmk \omega, \alpha,\theta, (\eta_{g\sigma}^{\epsilon(i)}),
(\alpha_L,\alpha_R, \Upsilon)
\rmk\in C^1(G, \bbZ_2\oplus\bbZ_2),\\
\kappa_\sigma^{(i)}
:=
\kappa_\sigma\lmk \omega, \alpha,\theta, (\eta_{g\sigma}^{\epsilon(i)}),
(\alpha_L,\alpha_R, \Upsilon)
\rmk\in C^2(G, \bbZ_2\oplus\bbZ_2)
\end{split}
\end{align}
for $i=1,2$. (Recall  Definition \ref{crdef} and 
Definition \ref{bkrdef}.)
Note that $\kappa_\sigma^{(1)}(g,h)=\kappa_\sigma(g,h)$.
We also set
\begin{align}
\begin{split}
c^{(i)}(g,h,k):=
c\lmk \omega,\alpha,\theta, (\eta_{g\sigma}^{\epsilon(i)}), (\alpha_L,\alpha_R,\Upsilon), (W_g^{\epsilon(i)}), (u_\sigma^{\epsilon(i)}(g,h))\rmk
\end{split}
\end{align}
for $i=1,2$.

We also define $m\in C^1(G, \bbZ_2\oplus\bbZ_2)$ by
\begin{align}\label{mdef}
m(g):=\lmk \mss g R{+1}, \mss g R{-1}\rmk,\quad g\in G.
\end{align}

%

Note from the definition and (\ref{yon46}), (\ref{go46})
that
\begin{align}\label{jyu46}
\begin{split}
&\bgn \epsilon 2 g
=\bgn \epsilon 1 g+\mss g L\epsilon +\mss g R\epsilon,\\
&\kssn R \epsilon 2(g,h)=\mss g R \epsilon+\mss h R{\eg g \epsilon}+\mss {gh} R\epsilon
+\kssn R \epsilon 1(g,h),\\
&\kssn L\epsilon 2 (g,h)
=m_{gL}^\epsilon+m_{hL}^{\eg g \epsilon}+m_{gh,L}^\epsilon+\kssn L \epsilon 1 (g,h).
\end{split}
\end{align}
in $\bbZ_2$.

Next we would like to derive
$c\lmk \omega,\alpha,\theta, (\eta_{g\sigma}^{\epsilon(2)}), (\alpha_L,\alpha_R,\Upsilon), (W_g^{\epsilon(2)}), (u_\sigma^{\epsilon(2)}(g,h))\rmk
\in C^3(G,\Uo\oplus \Uo)$.
To do so, we first calculate $u_{R}^{\epsilon(2)}(g,h) u_{R}^{\epsilon(2)}(gh,k)$
as follows.
From (\ref{yon46})
\begin{align}\label{nana46}
\begin{split}
&\unit_{L}\otimes u_{R}^{\epsilon(2)}(g,h) u_{R}^{\epsilon(2)}(gh,k)\\
&=\lmk -1\rmk^{\kappa_{L}^{\epsilon(1)}(g,h)\cdot \lmk m_{hR}^{(-1)^{a_{\omega}(g)}\epsilon}\rmk}
\lmk -1\rmk^{\kappa_{L}^{\epsilon(1)}(gh,k)\cdot \lmk m_{kR}^{(-1)^{a_{\omega}(gh)}\epsilon}\rmk}
\lmk 
\unit\otimes \lmk v_{gR}^{\epsilon}\rmk^{*}
\rmk\\
&\Ad\lmk W_{g}^{\epsilon}\rmk\lmk \unit\otimes \lmk v_{hR}^{(-1)^{a_{\omega}(g)}\epsilon}\rmk^{*}\rmk
\lmk \unit\otimes u_{R}^{\epsilon}(g,h)
\rmk\\&
\Ad\lmk W_{gh}^{\epsilon}\rmk\lmk \unit\otimes \lmk v_{kR}^{(-1)^{a_{\omega}(gh)}\epsilon}\rmk^{*}\rmk
\lmk \unit\otimes u_{R}^{\epsilon}(gh,k)
v_{ghk,R}^{\epsilon}\rmk\\
&=\lmk -1\rmk^{\kappa_{L}^{\epsilon(1)}(g,h)\cdot \lmk m_{hR}^{(-1)^{a_{\omega}(g)}\epsilon} +m_{kR}^{(-1)^{a_{\omega}(gh)}\epsilon}
\rmk}
\lmk -1\rmk^{\kappa_{L}^{\epsilon(1)}(gh,k)\cdot \lmk m_{kR}^{(-1)^{a_{\omega}(gh)}\epsilon}\rmk}
\lmk 
\unit\otimes \lmk v_{gR}^{\epsilon}\rmk^{*}
\rmk
\Ad\lmk W_{g}^{\epsilon}\rmk\lmk \unit\otimes \lmk v_{hR}^{(-1)^{a_{\omega}(g)}\epsilon}\rmk^{*}\rmk
\\&
\Ad\lmk \lmk \unit\otimes u_{R}^{\epsilon}(g,h)
\Gamma_{R}^{\kappa_{L}^{\epsilon(1)}(g,h)}\rmk W_{gh}^{\epsilon}\rmk\lmk \unit\otimes \lmk v_{kR}^{(-1)^{a_{\omega}(gh)}\epsilon}\rmk^{*}\rmk\cdot
\left[\lmk \unit\otimes u_{R}^{\epsilon}(g,h)u_{R}^{\epsilon}(gh,k)
\rmk
\right]
\lmk
\unih L\otimes v_{ghk,R}^{\epsilon}
\rmk.
\end{split}
\end{align}
Here in the last equality,
we got an extra sign $\lmk -1\rmk^{\kappa_{L}^{\epsilon(1)}(g,h)\cdot \lmk m_{kR}^{(-1)^{a_{\omega}(gh)}\epsilon}\rmk}$
because we inserted $\Gamma_{R}^{\kappa_{L}^{\epsilon(1)}(g,h)}$ in the middle.
Applyin Lemma \ref{lem41} to $[\cdot]$ part,
\begin{align}\label{hachi46}
\begin{split}
(\ref{nana46})
%
&
=\lmk -1\rmk^{\kappa_{L}^{\epsilon(1)}(g,h)\cdot \lmk m_{hR}^{(-1)^{a_{\omega}(g)}\epsilon} +m_{kR}^{(-1)^{a_{\omega}(gh)}\epsilon}
\rmk}
\lmk -1\rmk^{\kappa_{L}^{\epsilon(1)}(gh,k)\cdot \lmk m_{kR}^{(-1)^{a_{\omega}(gh)}\epsilon}\rmk}
\overline{c_{R}^{\epsilon(1)}(g,h,k)}
\lmk 
\unit\otimes \lmk v_{gR}^{\epsilon}\rmk^{*}
\rmk\\
&\Ad\lmk W_{g}^{\epsilon}\rmk\lmk \unit\otimes \lmk v_{hR}^{(-1)^{a_{\omega}(g)}\epsilon}\rmk^{*}\rmk
\left[
\Ad\lmk \lmk \unit\otimes u_{R}^{\epsilon}(g,h)
\Gamma_{R}^{\kappa_{L}^{\epsilon(1)}(g,h)}\rmk W_{gh}^{\epsilon}\rmk\lmk \unit\otimes \lmk v_{kR}^{(-1)^{a_{\omega}(gh)}\epsilon}\rmk^{*}\rmk\right]\\
&\cdot
\Ad\lmk W_{g}^{\epsilon}\rmk \lmk \unit_{L}\otimes u_{R}^{\epsilon(-1)^{a_{\omega}(g)}}(h,k)\rmk
\cdot \lmk \unit_{L}\otimes u_{R}^{\epsilon}(g,hk)\rmk
\lmk \unit\otimes
v_{ghk,R}^{\epsilon}\rmk
\\&
=\lmk -1\rmk^{\kappa_{L}^{\epsilon(1)}(g,h)\cdot \lmk m_{hR}^{(-1)^{a_{\omega}(g)}\epsilon} +m_{kR}^{(-1)^{a_{\omega}(gh)}\epsilon}
\rmk}
\lmk -1\rmk^{\kappa_{L}^{\epsilon(1)}(gh,k)\cdot \lmk m_{kR}^{(-1)^{a_{\omega}(gh)}\epsilon}\rmk}
\overline{c_{R}^{\epsilon(1)}(g,h,k)}
\lmk 
\unit\otimes \lmk v_{gR}^{\epsilon}\rmk^{*}
\rmk\\
&\Ad\lmk W_{g}^{\epsilon}\rmk\lmk \unit\otimes \lmk v_{hR}^{(-1)^{a_{\omega}(g)}\epsilon}\rmk^{*}\rmk
\Ad\lmk 
W_{g}^{\epsilon} W_{h}^{(-1)^{a_{\omega}(g)}\epsilon}
\rmk
\lmk \unit\otimes \lmk v_{kR}^{(-1)^{a_{\omega}(gh)}\epsilon}\rmk^{*}\rmk\\
&\Ad\lmk W_{g}^{\epsilon}\rmk \lmk \unit_{L}\otimes u_{R}^{\epsilon(-1)^{a_{\omega}(g)}}(h,k)\rmk
\cdot \lmk \unit_{L}\otimes u_{R}^{\epsilon}(g,hk)\rmk
\lmk \unit\otimes
v_{ghk,R}^{\epsilon}\rmk.
\end{split}
\end{align}
In the second equality, we used Lemma \ref{lem43} (ii),(iii)
 to $[\cdot]$ part.
 Now we recall the relation between
 $\uss R\epsilon$ and $\uss R{\epsilon (2)}$
 (\ref{yon46}) and substitute it to the equation above to obtain
\begin{align}
\begin{split}
(\ref{hachi46})
=&\lmk -1\rmk^{\kappa_{L}^{\epsilon(1)}(g,h)\cdot \lmk m_{hR}^{(-1)^{a_{\omega}(g)}\epsilon} +m_{kR}^{(-1)^{a_{\omega}(gh)}\epsilon}
\rmk}
\lmk -1\rmk^{\kappa_{L}^{\epsilon(1)}(gh,k)\cdot \lmk m_{kR}^{(-1)^{a_{\omega}(gh)}\epsilon}\rmk}\lmk -1\rmk^{\kappa_{L}^{\epsilon(1)}(g,hk)\cdot \lmk m_{hkR}^{(-1)^{a_{\omega}(g)}\epsilon}\rmk}\\&
\lmk -1\rmk^{\kappa_{L}^{(-1)^{a_{\omega}(g)}\epsilon}(h,k)\cdot \lmk m_{kR}^{(-1)^{a_{\omega}(h)}(-1)^{a_{\omega}(g)}\epsilon}\rmk}
\overline{c_{R}^{\epsilon(1)}(g,h,k)}
\\&
\Ad\lmk \lmk 
\unit\otimes \lmk v_{gR}^{\epsilon}\rmk^{*}
\rmk W_{g}^{\epsilon}\rmk
\lmk 
\begin{gathered}
\lmk \unit_{L}\otimes u_{R}^{(-1)^{a_{\omega}(g)}\epsilon(2)}(h,k)\rmk
\end{gathered}
\rmk
\\&
\lmk \unit\otimes u_{R}^{\epsilon(2)}(g,hk)\rmk
\\&
=\lmk -1\rmk^{\kappa_{L}^{\epsilon(1)}(g,h)\cdot \lmk m_{hR}^{(-1)^{a_{\omega}(g)}\epsilon} +m_{kR}^{(-1)^{a_{\omega}(gh)}\epsilon}
\rmk}
\lmk -1\rmk^{\kappa_{L}^{\epsilon(1)}(gh,k)\cdot \lmk m_{kR}^{(-1)^{a_{\omega}(gh)}\epsilon}\rmk}\lmk -1\rmk^{\kappa_{L}^{\epsilon(1)}(g,hk)\cdot \lmk m_{hkR}^{(-1)^{a_{\omega}(g)}\epsilon}\rmk}\\&
\lmk -1\rmk^{\kappa_{L}^{(-1)^{a_{\omega}(g)}\epsilon}(h,k)\cdot \lmk m_{kR}^{(-1)^{a_{\omega}(gh)}\epsilon}\rmk}
\lmk -1\rmk^{m_{gL}^{\epsilon}\cdot\kappa_{R}^{(-1)^{a_{\omega}(g)}\epsilon(2)}(h,k)}
\overline{c_{R}^{\epsilon(1)}(g,h,k)}
\\&
\Ad\lmk
 W_{g}^{\epsilon(2)}
\rmk
\lmk 
\begin{gathered}
\lmk \unit_{L}\otimes u_{R}^{(-1)^{a_{\omega}(g)}\epsilon(2)}(h,k)\rmk
\end{gathered}
\rmk
\unit\otimes\lmk u_{R}^{\epsilon(2)}(g,hk)\rmk.
\end{split}
\end{align}
In the second equality, we substituted the relation between
$W_g^{\epsilon}$ and
$W_g^{\epsilon (2)}$ (\ref{kyu46}).
Now seting $\tilde \sigma\in C^2(G, \Uo\oplus \Uo)$ as
\begin{align}\label{stdef}
\tilde\sigma^\epsilon(g,h):=(-1)^{\kss L {\epsilon(1)}(g,h) \mss h R {\eg g \epsilon}},
\quad\epsilon=\pm 1, g,h\in G
\end{align}
with (\ref{jyu46})
the phase factor in the last part of equation
can be written as
\begin{align}
\begin{split}
&\lmk -1\rmk^{\kappa_{L}^{\epsilon(1)}(g,h)\cdot \lmk m_{hR}^{(-1)^{a_{\omega}(g)}\epsilon} +m_{kR}^{(-1)^{a_{\omega}(gh)}\epsilon}
\rmk}
\lmk -1\rmk^{\kappa_{L}^{\epsilon(1)}(gh,k)\cdot \lmk m_{kR}^{(-1)^{a_{\omega}(gh)}\epsilon}\rmk}\lmk -1\rmk^{\kappa_{L}^{\epsilon(1)}(g,hk)\cdot \lmk m_{hkR}^{(-1)^{a_{\omega}(g)}\epsilon}\rmk}
\\&
\lmk -1\rmk^{\kappa_{L}^{(-1)^{a_{\omega}(g)}\epsilon(1)}(h,k)\cdot \lmk m_{kR}^{(-1)^{a_{\omega}(gh)}\epsilon}\rmk}
\lmk -1\rmk^{m_{gL}^{\epsilon}\cdot\kappa_{R}^{(-1)^{a_{\omega}(g)}\epsilon(2)}(h,k)}\\&
=
\frac{\tilde\sigma^{\eg g\epsilon} (h,k)\cdot \tilde\sigma^\epsilon (g,hk)}{\tilde\sigma^\epsilon (g,h) \cdot \tilde\sigma^\epsilon (gh,k)}
(-1)^{\kss L {\epsilon(1)} (g,h)\cdot \mss k R {\eg{gh}\epsilon}}
\lmk -1\rmk^{m_{gL}^{\epsilon}\cdot\kappa_{R}^{(-1)^{a_{\omega}(g)}\epsilon(2)}(h,k)}
\\
&=
\lmk
d^2_{a_\omega}\tilde \sigma (g,h,k)\cdot
(-1)^{\kappa_L^{(1)} (g,h)\cdot\lmk  m^{a_\omega(gh)}(k)\rmk}
(-1)^{\lmk \lmk  
{\kappa_R^{(2)}}^{\ao{g}}(h,k)
\rmk
\rmk\cdot \lmk b^{(2)}(g)- b^{(1)}(g)-m(g)\rmk
}
\rmk_\epsilon
\end{split}
\end{align}
Hence we got
\begin{align}
\begin{split}
&\unit_{L}\otimes u_{R}^{\epsilon(2)}(g,h) u_{R}^{\epsilon(2)}(gh,k)\\
&=
\lmk
d^2_{a_\omega}\tilde \sigma (g,h,k)\cdot
(-1)^{\kappa_L^{(1)} (g,h)\cdot\lmk  m^{a_\omega(gh)}(k)\rmk}
(-1)^{\lmk \lmk  
{\kappa_R^{(2)}}^{\ao{g}}(h,k)
\rmk
\rmk\cdot \lmk b^{(2)}(g)- b^{(1)}(g)-m(g)\rmk
}\overline {c_{R}^{(1)}(g,h,k)}
\rmk_\epsilon
\\
&\Ad\lmk
 W_{g}^{\epsilon(2)}
\rmk
\lmk 
\begin{gathered}
\lmk \unit_{L}\otimes u_{R}^{(-1)^{a_{\omega}(g)}\epsilon(2)}(h,k)\rmk
\end{gathered}
\rmk
\unit\otimes\lmk u_{R}^{\epsilon(2)}(g,hk)\rmk
\end{split}
\end{align}
Hence we have
\begin{align}
\begin{split}
{c_{R}^{(2)}(g,h,k)}
=
d^2_{a_\omega}\tilde \sigma (g,h,k)\cdot
(-1)^{\kappa_L^{(1)} (g,h)\cdot\lmk  m^{a_\omega(gh)}(k)\rmk}
(-1)^{\lmk \lmk  
{\kappa_R^{(2)}}^{\ao{g}}(h,k)
\rmk
\rmk\cdot \lmk b^{(2)}(g)- b^{(1)}(g)-m(g)\rmk
}{c_{R}^{(1)}(g,h,k)}.
\end{split}
\end{align}
Combining this with (\ref{jyu46}), we have
\begin{align}
(c^{(1)}, \kappa_R^{(1)},\kappa_L^{(1)},b^{(1)},a_\omega)\sim_{\PD}(c^{(2)}, \kappa_R^{(2)},\kappa_L^{(2)},b^{(2)},a_\omega)
\end{align}
with $m\in C^1(G, \bbZ_2\oplus \bbZ_2)$ in (\ref{mdef})
and $\tilde \sigma\in C^2(G, \Uo\oplus \Uo)$ in (\ref{stdef}).
This proves the Lemma.
\end{proof}
\begin{defn}
Let $\omega\in \SPT$.
By the same simple argument as Lemma 2.17 in \cite{2dSPT},
one can show that $h^{(3)}(\omega,\theta)$ is independent of the choice of
$0<\theta<\frac\pi 2$.
Hence we can define an index $h(\omega)=h^{(3)}(\omega,\theta)$,
independent of the choice of $\theta$.
\end{defn}

\section{Stability}\label{stabilitysec}
Now, in order to prove Theorem \ref{mainthm}, it suffices to prove the following.
\begin{thm}
For any $\omega\in \SPT$, $\gamma\in \QAut_\beta(\caA_{\bbZ^2})$,
we have $h(\omega)=h(\omega\circ\gamma)$.
\end{thm}
\begin{proof}
Let $\alpha\in \EAut(\omega)$, $0<\theta_0<\theta<\frac\pi 2$,
$(\etas g \sigma \epsilon{})\in I(\omega, \theta_0)$,
$(\alpha_L,\alpha_R,\Upsilon)\in \caD^V(\alpha,\theta)$.
Set $\theta_{2}:=\theta$ and choose 
\begin{align}
0<\theta_{0}<\theta_{0.8}<\theta_1<\theta_{1.2}<\theta_{1.8}<\theta_2<\theta_{2.2}<
\theta_{2.8}<\theta_3<\theta_{3.2}<\frac\pi 2.
\end{align}
Because $\gamma\in \QAut_\beta(\caA_{\bbZ^2})$, by the same proof as 
Theorem 5.2 \cite{2dSPT},
we can decompose 
$\gamma$ as
\begin{align}\label{gcz}
\gamma=\gamma_{C}\circ\gamma_{H}.
\end{align}
Here, 
automorphisms 
$\gamma_{H}$ and 
$\gamma_{C}$
belong to 
$\QAut\lmk \caA_{\bbZ^2}\rmk$
and 
allow the following decompositions:
{\it 1.}$\gamma_H$ is decomposed as 
\begin{align}\label{ght}
\gamma_{H}=\inn\circ\lmk \gamma_{H, {L}}\hat\otimes \gamma_{H,R}\rmk
=\inn\circ\gamma_{0}
\end{align}
with
some $\gamma_{H, {\sigma}}\in \Aut^{(0)}\lmk \caA_{{\lmk C_{\theta_{0}}\rmk\cap H_\sigma^{c_\sigma^{(-1)}}}}\rmk$, $\sigma=L,R$.
Here set $\gamma_{0}:=\gamma_{H, {L}}\hat\otimes \gamma_{H,R}\in 
\auz\lmk \caA_{{ C_{\theta_{0}}}}\rmk$,
and $c_R^{(-1)}=6$, $c_L^{(-1)}:=-6$.
We have chosen the support of 
$\gamma_{H, {\sigma}}$ as $ C_{\theta_{0}}\cap H_\sigma^{c_\sigma^{(-1)}}$
so that $\widetilde {\etas g\sigma \epsilon{}}$
in (\ref{ettt}) below has a support in $C_{\theta_{0}}\cap H_\sigma^{c_\sigma}$.
\\
{\it 2. }The automorphism 
$\gamma_{C}$ allows a
decomposition
\begin{align}\label{sqautg}
&\gamma_{C}=\inn\circ\gamma_{CS}\notag\\
&\gamma_{CS}=\lmk
\gamma_{[0,\theta_1]}{\hat \otimes}\gamma_{(\theta_1,\theta_2]}
{\hat \otimes} \gamma_{(\theta_2,\theta_3]}{\hat \otimes}
\gamma_{(\theta_3,\frac\pi 2]}
\rmk
\circ
\lmk
\gamma_{(\theta_{0.8}, \theta_{1.2}]}{\hat \otimes}
\gamma_{(\theta_{1.8},\theta_{2.2}]}
{\hat \otimes} \gamma_{(\theta_{2.8},\theta_{3.2}]}
\rmk
\end{align}
with  \begin{align}\label{sqaut2g}
  \begin{split}
&  \gamma_X:=\widehat \bigotimes_{\sigma=L,R,\zeta=D,U} \gamma_{X,\sigma,\zeta},\quad
 \gamma_{[0,\theta_{1}]}:=\widehat \bigotimes_{\sigma=L,R}\gamma_{[0,\theta_{1}],\sigma},\quad
 \gamma_{(\theta_3,\frac\pi 2]}:=\widehat\bigotimes_{\zeta=D,U}  \gamma_{(\theta_3,\frac\pi 2],\zeta},\\
 &\gamma_{X,\sigma,\zeta}\in \auz\lmk\caA_{C_{X}
 \cap H_\sigma^{c_\sigma}\cap H_\zeta}\rmk,\quad
 \gamma_{X,\sigma}:=\widehat\bigotimes_{\zeta=U,D}\gamma_{X,\sigma,\zeta},\quad
\gamma_{X,\zeta}:=\widehat\bigotimes_{\sigma=L,R}\gamma_{X,\sigma,\zeta},\\
&\gamma_{[0,\theta_{1}],\sigma}\in \auz\lmk\caA_{C_{[0,\theta_{1}]}\cap H_\sigma^{c_\sigma}}\rmk,\quad
 \gamma_{(\theta_3,\frac\pi 2],\zeta}\in \auz\lmk\caA_{C_{(\theta_3,\frac\pi 2]}\cap H_\zeta}\rmk, 
  \end{split} 
  \end{align}
 for
 \begin{align}\label{sqaut3g}
 X=(\theta_1,\theta_2], (\theta_2,\theta_3],
 (\theta_{0.8},\theta_{1.2}],
 (\theta_{1.8},\theta_{2.2}], 
(\theta_{2.8},\theta_{3.2}],\quad \sigma=L,R,\quad \zeta=D,U.
 \end{align}
 These automorphisms satisfy
 \begin{align}\label{san47}
\gamma_{I}\circ\beta_g^{U}=\beta_g^{U}\circ\gamma_{I}\;\quad\text{for all}\quad g\in G,\quad\text{and}\quad
\gamma_{I}\circ\tau=\tau\circ\gamma_{I}
\end{align}
for any 
\begin{align}\label{sqaut3g}
I=[0,\theta_1],(\theta_1,\theta_2], (\theta_2,\theta_3], \left(\theta_3,\frac\pi 2\right],
(\theta_{0.8}, \theta_{1.2}], (\theta_{1.8},\theta_{2.2}], (\theta_{2.8},\theta_{3.2}].
 \end{align}
\quad
\\\quad

Now we set
\begin{align}\label{athens}
\hat\Upsilon:=\Upsilon\circ
\lmk
\gamma_{(\theta_2,\theta_3]}\otimes
\gamma_{(\theta_3,\frac\pi 2]}
\rmk
\circ
\lmk
\gamma_{(\theta_{1.8},\theta_{2.2}]}
\otimes \gamma_{(\theta_{2.8},\theta_{3.2}]}
\rmk
\in\auz\lmk \caA_{C_{\theta_{1.8}}^{c}}\rmk
\subset \auz\lmk \caA_{C_{\theta_{1.2}}^{c}}\rmk,
\end{align}
and
\begin{align}\label{rome}
\hat\alpha_{\sigma}:=
\alpha_{\sigma}\circ
\lmk
\gamma_{[0,\theta_1],\sigma}\hat\otimes\gamma_{(\theta_1,\theta_2],\sigma}
\rmk
\circ
\gamma_{(\theta_{0.8}, \theta_{1.2}],\sigma}
\circ\gamma_{H,\sigma}\in\auz(\caA_{H_{\sigma}}),\quad\sigma=L,R.
\end{align}
Then as in the proof of Theorem 3.1 of \cite{2dSPT} {\it Step 1}, we can check
\begin{align}\label{algm}
\alpha\circ\gamma=\inn\circ\lmk\hat\alpha_{L}\otimes\hat \alpha_{R}\rmk\circ\hat\Upsilon.
\end{align}
Therefore, we have 
$(\hat\alpha_L,\hat\alpha_R,\hat \Upsilon)\in\caD^V\lmk \alpha\gamma, \theta_{1.2}\rmk$.
Because $\tau$ and $\gamma_{CS}$ commute,
we also have
\begin{align}\label{ichi47}
\gamma^{-1} \tau^a\gamma 
=\tau^a \lmk
\lmk \gamma_{HL}\rmk_{-a}^{-1} \gamma_{HL}
\hat\otimes  \lmk \gamma_{HR}\rmk_{-a}^{-1} \gamma_{HR}
\rmk\circ\inn 
\end{align}
As in the proof of Theorem 3.1 of \cite{2dSPT} {\it Step 2}
we have
\begin{align}\label{ni47}
\begin{split}
&\gamma^{-1}\eta_g^\epsilon \gamma
=\gamma_0^{-1} \gamma_{(\theta_{0.8}, \theta_{1.2}]}^{-1}
\gamma_{[0,\theta_1]}^{-1} \eta_g^\epsilon
\gamma_{[0,\theta_1]}\gamma_{(\theta_{0.8}, \theta_{1.2}]}\gamma_0\circ\inn,\\
&\gamma^{-1}{\beta_g^U}^{-1} \gamma \beta_g^U
=\gamma_0^{-1} {\beta_g^U}^{-1} \gamma_0 \beta_g^U\circ\inn.
\end{split}
\end{align}
Set
\begin{align}\label{ettt}
\widetilde {\etas g\sigma \epsilon{}}:=
\lmk \gamma_{H_\sigma}\rmk_{-\ao g\epsilon}^{-1}
\lmk \gamma_{(\theta_{0.8}, \theta_{1.2}]\sigma}\rmk^{-1}
\lmk \gamma_{[0, \theta_{1}]\sigma}\rmk^{-1}
\etas g\sigma \epsilon{}
\gamma_{[0, \theta_{1}]\sigma}\gamma_{(\theta_{0.8}, \theta_{1.2}]\sigma}
\lmk\beta_g^{U\sigma}\rmk^{-1}
 \gamma_{H_\sigma}\beta_g^{U\sigma}.
\end{align}
We claim $(\widetilde {\etas g\sigma \epsilon{}})\in I(\omega\gamma,\theta_{1.2})$.
It can be seen as follows using (\ref{ichi47}), (\ref{ni47}) :
\begin{align}
\begin{split}
&\omega\gamma\beta_g^U
=\omega\beta_g^U{\beta_g^U}^{-1}\gamma\beta_g^U 
=\omega\tau^{\ao g \epsilon}\lmk {\etas g L\epsilon{}}\hat\otimes {\etas g R \epsilon{}}\rmk{\beta_g^U}^{-1}\gamma\beta_g^U \circ\inn\\
&=
\omega\gamma\gamma^{-1}\tau^{\ao g \epsilon}\gamma
\circ\gamma^{-1}\lmk {\etas g L\epsilon{}}\hat\otimes {\etas g R \epsilon{}}\rmk
\gamma\circ\gamma^{-1}{\beta_g^U}^{-1}\gamma\beta_g^U \circ\inn
=\omega\gamma\tau^{\ao g \epsilon}
\lmk\widetilde {\etas g L\epsilon{}}\hat\otimes \widetilde {\etas g R \epsilon{}}\rmk\circ\inn.
\end{split}
\end{align}
From this equality, we can also see that $\ao g=a_{\omega\gamma}(g)$
for any $g\in G$.

Because $(\hat\alpha_L,\hat\alpha_R,\hat \Upsilon)\in\caD^V\lmk \alpha\gamma, \theta_{1.2}\rmk$
and $(\widetilde {\etas g\sigma \epsilon{}})\in I(\omega\gamma,\theta_{1.2})$,
in order to prove the Theorem,
it suffices to show
\begin{align}
\begin{split}
\IP\lmk\omega,\alpha,\theta_2,\lmk \etas g\sigma \epsilon {}\rmk, (\alpha_L,\alpha_R,\Upsilon)\rmk
=\IP\lmk\omega\gamma,\alpha\gamma,\theta_{1.2},\lmk \widetilde{\etas g\sigma \epsilon {}}\rmk, 
(\hat \alpha_L,\hat \alpha_R,\hat \Upsilon)\rmk.
\end{split}
\end{align}
In other words,
it suffices to show the following:
\begin{align}
\begin{split}
&\lmk\hat \alpha_L\hat\otimes \hat \alpha_R\rmk\hat \Upsilon\beta_g^U
\lmk \widetilde\eta_g^\epsilon\rmk^{-1}\tau^{-a_{\omega\gamma}(g) \epsilon}
\hat \Upsilon^{-1} 
\lmk\hat \alpha_L\hat\otimes \hat \alpha_R\rmk^{-1}
=
\lmk\alpha_L\hat\otimes \alpha_R\rmk \Upsilon\beta_g^U
\lmk \eta_g^\epsilon\rmk^{-1}\tau^{-\ao g \epsilon}
\Upsilon^{-1} 
\lmk\alpha_L\hat\otimes \alpha_R\rmk^{-1},\\
&\hat\alpha_\sigma \beta_{g}^{U\sigma}\lmk \widetilde{\etas g \sigma \epsilon{}}\rmk^{-1}
\beta_h^{U\sigma}
\lmk
\lmk \widetilde {\etas h \sigma{\eg g\epsilon }{}}\rmk_
{-a_{\omega\gamma}(g) \epsilon}
\rmk^{-1}
\widetilde {\etas {gh} \sigma \epsilon{}}
{\beta_{gh}^{U\sigma}}^{-1}
\hat \alpha_\sigma^{-1}\\
&=\alpha_\sigma \beta_{g}^{U\sigma}\lmk {\etas g \sigma \epsilon{}}\rmk^{-1}
\beta_h^{U\sigma}
\lmk
\lmk  {\etas h \sigma{\eg g\epsilon }{}}\rmk_{-\ao g \epsilon}
\rmk^{-1}
 {\etas {gh} \sigma \epsilon{}}
{\beta_{gh}^{U\sigma}}^{-1}
\alpha_\sigma^{-1}
\end{split}
\end{align}
This can be shown as in {\it Step 3.} of Theorem 3.1 \cite{2dSPT},
noting the commutativity (\ref{san47}).

\end{proof}

\section{Automorphisms on $\sdckc{}$}\label{sdautosec}
In this section we show the following proposition.
It allows us to reduce the support of homogeneous automorphisms
by finite portion.
\begin{prop}\label{lem32}
Let $\mkk_1$ be a finite even dimensional Hilbert space with a complex conjugation $\mkC_1$.
Let $\mkk_2$ be an infinite dimensional Hilbert space with a complex conjugation $\mkC_2$.
Set $\mkk:=\mkk_1\oplus\mkk_2$ with a complex conjugation $\mkC=\mkC_1\oplus \mkC_2$.
Then for any $\alpha\in\auz\lmk \sdckc{}\rmk$,
there are unitary $u\in\sdckc{}$ and an automorphism
$\alpha_2\in\auz\lmk \sdckc{2}\rmk$ such that
\begin{align}\label{uaua}
\Ad u\circ\alpha
=\id_{\sdckc{1}}\hat\otimes \alpha_2.
\end{align}
\end{prop}
We start by several Lemmas.
\begin{lem}\label{lem34}
For each $n\in2\bbN$, there is some $\delta_n>0$
satisfying the following condition.:
Let $\mkk_i$, $i=1,2,3$ be Hilbert spaces with complex conjugations $\mkC_i$.
Suppose that $\mkk_1$ is $n$-dimensional.
Suppose that 
$\mkk_3$ is even-finite-dimensional, and $\mkk_2$ is infinite dimensional. 
Set $\mkk:=\mkk_1\oplus\mkk_2\oplus \mkk_3$ 
with a complex conjugation $\mkC=\mkC_1\oplus \mkC_2\oplus \mkC_3$.
Let $\alpha\in \auz\lmk\sdckc{}\rmk$.
Let $\{e_{IJ}\}_{I,J\subset \{1,\ldots,\frac{\dim \mkk_1}{2}\}}$ 
be a system of matrix units spanning $\sdckc{1}$
with  grading of $e_{IJ}$ being $|I|+|J| \mod 2$. 
Suppose that
\begin{align}
\dist\lmk
\alpha(e_{IJ}), \sdcis
\rmk<\delta_n.
\end{align}
Here $\dist(X,Y)$ means the distance between the two subsets
$X,Y$ of $\sdckc{}$, with respect to the $C^*$-norm.
Then there is an even unitary $u\in\sdckc{}$ such that
\begin{align}
\Ad\lmk u\rmk\lmk \alpha\lmk \sdckc{1}\rmk\rmk
\subset \sdcis.
\end{align}
\end{lem}
\begin{proof}
This holds directly from the proof of Lemma III 3.2, 3.2
of \cite{Davidson}.
It is easy to see that the unitary obtained by the argument there 
is even.{\che{}}
\end{proof}
\begin{lem}\label{lem35}
Let $\mkk_i$, $i=1,2$ be finite even dimensional Hilbert spaces with
complex conjugations $\mkC_i$.
Let $\{e_{IJ}\}_{I,J\subset \{1,\ldots, \frac{\dim\mkk_1}{2} \}}$
be a system of matrix units spanning $\sdckc{1}$
with  grading of $e_{IJ}$ being $|I|+|J| \mod 2$. 
Let $\{f_{IJ}\}_{I,J\subset \{1,\ldots, \frac{\dim\mkk_1}{2} \}}$
be a system of matrix units in $\sdcin$
with  grading of $f_{IJ}$ being $|I|+|J| \mod 2$. 
Suppose that there is a self-adjoint unitary $v\in \sdcin^{(1)}$
such that $\Ad(v)(f_{IJ})=(-1)^{|I|+|J|} f_{IJ}$.
Then there is an even unitary $U\in\sdcin$
such that $f_{IJ}=\Ad\lmk U\rmk (e_{IJ})$.
\end{lem}
\begin{proof}
Let $(\caH_1,\pi_1)$ be 
 Fock representations of $\sdckc{1}$ with
 grading unitary $\tilde \Gamma_1$.
 For each $I\subset \{1,\ldots, \frac{\dim\mkk_1}{2}\}$, 
 note that $\pi_1(e_{II})\caH_1$ is one-dimensional
 space because $\pi_1(e_{II})$ is a minimal projection of $\sdckc{1}$
 and $\pi_1$ is irreducible.
 Let
  $e_I$ be a unit vector in $\pi_1(e_{II})\caH_1$.
  Because $\tilde \Gamma_1\pi_1(e_{II})\caH_1=\pi_1(e_{II}) \tilde \Gamma_1\caH_1
  =\pi_1(e_{II})\caH_1$,
  $e_I$ is eigenvector of $\tilde \Gamma_1$ with eigenvalue $\pm 1$.
  If $\tilde\Gamma_1 e_\emptyset=e_\emptyset$/ $\tilde\Gamma_1 e_\emptyset=-e_\emptyset$, set 
  $\Gamma_1:=\tilde \Gamma_1$/ $\Gamma_1:=-\tilde\Gamma_1$.
  so that $\Gamma_1 e_\emptyset=e_\emptyset$.
  For general $I\subset \{1,\ldots, \frac{\dim\mkk_1}{2}\}$, 
  we have $e_I=c_I \pi_1(e_{I\emptyset})e_\emptyset$, with some $c_I\in \Uo$.
  Hence we have
  $\Gamma_1 e_I=c_I \Gamma_1\pi_1(e_{I\emptyset})e_\emptyset
  =(-1)^{|I| }c_I\pi_1(e_{I\emptyset})e_\emptyset=(-1)^{|I| } e_I$.
  
  Let $(\caH_2,\pi_2)$ be 
 Fock representations of $\sdckc{2}$ with
 grading unitary $\Gamma_2$.
 Then $(\caH,\pi):=(\caH_1\otimes \caH_2, \pi_1\hat\otimes \pi_2)$
 is a Fock representation of $\sdcin$
 with a grading unitary $\Gamma:=\Gamma_1\otimes \Gamma_2$.
  
Setting $\caK:=\pi(f_{\emptyset\emptyset})\caH$,
\begin{align}
W\xi:=\sum_{I\subset \{1,\ldots, \frac{\dim\mkk_1}{2}\}}e_I\otimes
\pi(f_{\emptyset I})\xi,\quad \xi\in \caH
\end{align}
defines a unitary from $\caH$ to $\caH_1\otimes\caK$ satisfying
\begin{align}
\Ad\lmk W\rmk\lmk \pi(f_{IJ})\rmk
=e_{IJ}\otimes \unit.
\end{align}
We also have
\begin{align}
W \Gamma=\lmk \Gamma_1\otimes  \Gamma\pi\lmk f_{\emptyset\emptyset}\rmk\rmk W.
\end{align}
Because $f_{\emptyset\emptyset}$ is even,
$\caK$ is invariant under $\Gamma$ and $\Gamma\pi\lmk f_{\emptyset\emptyset}\rmk$
is a self-adjoint unitary on $\caK$.
We have
$\Ad\lmk \pi(v)\rmk\lmk\Gamma\pi\lmk f_{\emptyset\emptyset}\rmk\rmk
=-\Gamma\pi\lmk f_{\emptyset\emptyset}\rmk$,
with the unitary $v$.
It means the eigenvalue $1$ of $\pi(\Gamma)\pi\lmk f_{\emptyset\emptyset}\rmk$
and that of $-1$ have the same degeneracy.
From this, we may find a unitary
$u : \caH_2\to \caK$
such that $\Ad u\lmk \Gamma_2\rmk=\Gamma\pi\lmk f_{\emptyset\emptyset}\rmk$.
Setting $\tilde U:=\lmk\unih{1}\otimes u^*\rmk W\in\caU(\caH)$,
we have 
$
\tilde U\Gamma=\Gamma \tilde U.
$
Then there is an even unitary $U\in\sdcin$ such that $\tilde U=\pi(U^*)$.
From the definition, we can check that $f_{IJ}=\Ad\lmk U\rmk (e_{IJ})$.

\end{proof}

\begin{lem}\label{lem31}
Let $\mkk_1$ be a finite even dimensional Hilbert space with a complex conjugation $\mkC_1$.
Let $\mkk_2$ be an infinite dimensional Hilbert space with a complex conjugation $\mkC_2$.
Set $\mkk:=\mkk_1\oplus\mkk_2$ with a complex conjugation $\mkC=\mkC_1\oplus \mkC_2$.
Let $\alpha\in\auz\lmk \sdckc{}\rmk$ such that $\alpha\vert_{\sdckc{1}}=\id_{\sdckc{1}}$.
Then there is an automorphism $\alpha_2\in \auz\lmk \sdckc{2}\rmk$
such that $\alpha=\id_{\sdckc{1}}\hat\otimes \alpha_2$.
\end{lem}
\begin{proof}
It suffices to show $\alpha\lmk\sdckc2\rmk=\sdckc2$.
Because $\mkk_1$ is of finite even dimensional and $\mkk_2$
is of infinite dimensional they are $*$-isomorphic to some CAR-algebras.
Let $v_{\mkk_1}\in \sdckc1$ be the grading operator
of $\sdckc1$. It is an even self-adjoint unitary.
From \cite{aramori} Lemma 4.15,
we have
\begin{align}\label{amf}
\sdckc1'\cap\sdckc{}=\sdckc2^{(0)}+v_{\mkk_1} \sdckc2^{(1)}.
\end{align} 

For any $f\in\mkk_2$, we claim $\alpha\lmk B(f)\rmk$
belongs to $\sdckc2$.
To see this, note because $\alpha$ is graded, that $\alpha\lmk B(f)\rmk$
is odd. 
Furthermore,  any homogeneous $a\in\sdckc1$, we have
\begin{align}
a\alpha\lmk B(f)\rmk-(-1)^{\partial a}\alpha\lmk B(f)\rmk a
=\alpha\lmk
aB(f)-(-1)^{\partial a} B(f)a
\rmk=0,
\end{align}
because $\alpha(a)=a$ from the assumption.
Hence we see that $\alpha\lmk B(f)\rmk v_{\mkk_1}$
is an odd element in $\sdckc1'\cap\sdckc{}$.
From (\ref{amf}), we conclude $\alpha(B(f))\in \sdckc2^{(1)}$, proving the claim.

As this holds for any $f\in\mkk_2$, we conclude
$\alpha\lmk\sdckc2\rmk\subset \sdckc2$.
The same argument for $\alpha^{-1}$ implies the opposite inclusion and we obtain
$\alpha\lmk\sdckc2\rmk=\sdckc2$.
\end{proof}

\begin{proofof}[Proposition \ref{lem32}]
Let $\mkk_0$ be a $2$-dimensional $\mkC_2$-invariant subspace of
$\mkK_2$ with the complex conjugation 
$\mkC_0:=\left.\mkC_2\right\vert_{\mkk_0}$.
Let $\{e_{IJ}\}_{I,J\subset \{1,\ldots, \frac{\dim\mkk_1}{2}+1 \}}$
be a system of matrix units spanning $\SDC\lmk\mkk_1\oplus\mkk_0, \mkC_1\oplus \mkC_0\rmk$
with  grading of $e_{IJ}$ being $|I|+|J| \mod 2$. 
Let $\{E_{IJ}\}_{I,J\subset \{1,\ldots, \frac{\dim\mkk_1}{2} \}}$
be a system of matrix units spanning $\SDC\lmk\mkk_1, \mkC_1\rmk$
with  grading of $E_{IJ}$ being $|I|+|J| \mod 2$. 

We claim that there is an even finite dimensional $\mkC_2$-invariant subspace 
$\mkK_3$ of $\mkk_2\ominus \mkk_0$
with the complex conjugation $\mkC_3:=\mkC_2\vert_{\mkk_3}$,
and a even unitary $U\in\caU\lmk \sdckc{}\rmk$
such that
\begin{align}\label{ueij}
\Ad U\circ \alpha\lmk\sdciz\rmk
\subset \sdcizs.
\end{align}
To show this, let $\delta_{{\dim\mkk_1}+2}$ be the number given in Lemma \ref{lem34}.
Then there exists an even finite dimensional $\mkC_2$-invariant subspace
$\mkk_3$ of $\mkk_2\ominus \mkk_0$
such that
\begin{align}
\dist\lmk \alpha\lmk e_{IJ}\rmk, \sdcizs\rmk<\delta_{{\dim\mkk_1}+2}.
\end{align}
Applying Lemma \ref{lem34}
with $\mkk_1$, $\mkk_2$, $\mkk_3$
replaced by $\mkk_1\oplus \mkk_0$, $\mkk_2\ominus\lmk \mkk_0\oplus\mkk_3\rmk$,
$\mkk_3$
we obtain an even $U\in\caU\lmk \sdckc{}\rmk$ satisfying (\ref{ueij}).

Next we apply Lemma \ref{lem35}
with $\mkk_1$, $\mkk_2$, $e_{IJ}$
$f_{IJ}$
replaced by $\mkk_1$
$\mkk_0\oplus \mkk_3$, $E_{IJ}$
$\Ad U\circ\alpha(E_{IJ})$.
Note that because $U$ is even and $\alpha$ is homogeneous,
$\Ad U\circ\alpha(E_{IJ})$ has a degree $|I|+|J| \mod 2$.
The algebra $\sdckc{0}$ has an odd self-adjoint unitary $v_0$.
Because $U$ is even and $\alpha$ is homogeneous, $v:=\Ad U\circ\alpha(v_0)$
is an odd self-adjoint unitary in $\sdcizs$.
It satisfies
\begin{align}
\Ad v\lmk  \Ad U\circ\alpha(e_{IJ}) \rmk
=\Ad U\circ\alpha\lmk\Ad(v_0)(e_{IJ})\rmk
=(-1)^{|I|+|J|}\Ad U\circ\alpha(e_{IJ}).
\end{align}
Hence the required conditions of Lemma \ref{lem35} is satisfied.
Applying Lemma \ref{lem35}, we obtain
an even unitary $V\in\sdcizs$
such that $\Ad V(E_{IJ})=\Ad U\circ\alpha(E_{IJ})$.
Setting $u:=V^* U\in \sdckc{}$,
$u$ is a even unitary such that $\Ad u\circ\alpha(e_{IJ})=e_{IJ}$
for all $I,J \subset \{1,\ldots, \frac{\dim\mkk_1}{2} \}$.
 It means $\Ad u \circ\alpha\vert_{\sdckc{1}}=\id_{\sdckc{1}}$.
By Lemma \ref{lem31} it means 
an automorphism
$\alpha_2\in\auz\lmk\sdckc {2}\rmk$ satisfying (\ref{uaua}). 
\end{proofof}

{\bf Acknowledgment.}\\
{
The author is grateful to Yuji Tachikawa for a stimulating discussion and helpful comments. 
This work was supported by JSPS KAKENHI Grant Number 16K05171 and 19K03534.
It was also supported by JST CREST Grant Number JPMJCR19T2.
}

\appendix
\section{Basic Notations}\label{notasec}
 
For a Hilbert space $\caH$, $B(\caH)$ denotes the set of all bounded operators on $\caH$,
while $\caU(\caH)$ denotes the set of all unitaries on $\caH$.
If $V:\caH_1\to\caH_2$ is a linear map from a Hilbert space $\caH_1$ to 
another Hilbert space $\caH_2$,
then $\Ad (V):B(\caH_1)\to B(\caH_2)$ denotes the map
$\Ad(V)(x):=V x V^*$, $x\in B(\caH_1)$.
Occasionally we write $\Ad_V$ instead of $\Ad(V)$.
For a $C^{*}$-algebra $\caB$ and $v\in \caB$, we set 
$\Ad(v)(x):=\Ad_{v}(x):=vxv^{*}$, $x\in \caB$.

For a state $\omega$ on a $C^*$-algebra
$\mkB$,
we denote by $(\caH_\omega, \pi_\omega,\Omega_\omega)$
its GNS triple.
For a $C^*$-algebra $\mkB$, we denote by $\caU(\mkB)$
the set of all unitaries in $\mkB$.
For a $C^*$-algebra $\mkB$, $\mkB_{+,1}$
denotes the set of all positive elements of $\mkB$ with
norm less than or equal to $1$.
For states $\omega,\varphi$ on a $C^*$-algebra $\mkB$,
we write $\omega\simeq \varphi$ if they are equivalent
and $\omega\qe\varphi$ if they are quasi-equivalent.
We denote by $\Aut \caB$ the group of automorphisms on a $C^{*}$-algebra $\caB$.
The group of inner automorphisms on  a unital $C^{*}$-algebra $\caB$ is denoted by $\Inn \caB$.
For $\gamma_1,\gamma_2\in\Aut(\caB)$, $\gamma_1=\inn\circ\gamma_2$
means there is some unitary $u$ in $\caB$ such that $\gamma_1=\Ad(u)\circ\gamma_2$.
For a unital $C^{*}$-algebra $\caB$, the unit of $\caB$ is denoted by $\unit_{\caB}$.
For a Hilbert space we write $\unit_{\caH}:=\unit_{\caB(\caH)}$.
For a unital $C^{*}$-algebra $\caB$, by $\caU(\caB)$, we mean
the set of all unitary elements in $\caB$.For a state $\varphi$ on $\caB$ and a $C^{*}$-subalgebra $\caC$ of $\caB$,
$\varphi\vert_{\caC}$ indicates the restriction of $\varphi$ to $\caC$.

The center of a von Neumann algebra $\caM$ is denoted by $Z(\caM)$.

%

To denote the composition of automorphisms $\alpha_1$,
$\alpha_2$, all of
$\alpha_1\circ\alpha_2$,
$\alpha_1\alpha_2$, $\alpha_1\cdot\alpha_2$
are used.
Frequently, the first one serves as a bracket 
to visually separate a group of operators.

\section{Graded von Neumann algebras}\label{gravn}
In this section we collect facts we use about graded von Neumann algebras.
See \cite{1dFermi} for further explanation.
A graded von Neumann algebra is a pair $(\caM,\theta)$ with 
$\caM$ a von Neumann algebra $\theta$ an involutive automorphism on $\caM$, 
$\theta^2 = \mathrm{Id}$.
The even/odd part of $\caM$ with respect to the grading is denoted by $\caM^{(0)}$/$\caM^{(1)}$. 
 If $\caM \subset \caB(\caH)$ and there is a self-adjoint unitary
$\Gamma$ on $\caH$ such that $\mathrm{Ad}_{\Gamma}|_\caM = \theta$, then
we call $(\caM,\theta)$ a spatially graded von Neumann algebra with
grading operator $\Gamma$. 
We say $(\caM,\theta)$
 is balanced if $\caM$ contains an odd self-adjoint unitary. 
If $Z(\caM)\cap \caM^{(0)}=\bbC\unit$ for
 the center $Z(\caM)$ of $\caM$, we say $(\caM, \theta)$ is central.

For a homogeneous state $\omega$ on a graded $C^*$-algebra $\mkB$,
there is a self-adjoint unitary $\Gamma_\omega$ implementing the grading with respect
to $\pi_\omega$.
As a result, the grading extends to the von Neumann algebra
$\pi_\omega(\mkB)''$ by $\Ad\Gamma_\omega$.
We always consider this extension without mentioning explicitly.

\begin{lem}\label{apichi}
Let $(\caM,\Gamma)$ be a balanced central graded von Neumann algebra on a Hilbert space $\caH$.
Then both of $\caM$ and $\caM^{(0)}$ are either a factor or a direct sum of two factors of the same type.
\end{lem}
\begin{proof}The proof is basically the same as part of Proposition 2.9 of \cite{1dFermi}.
 From Lemma A.2 of \cite{1dFermi},
 if $\caM$ is not a factor, 
 it is a direct sum of two factors of the same type
 which maps to each other by $\Ad\Gamma$.
 
Let $U$ be a self-adjoint odd unitary in $\caM$.
Suppose that $\caM^{(0)}$ is not a factor.
Then there exists a projection $z$ in $Z(\caM^{(0)})$ which is not
$0$ nor $\unit$.
For such a projection, we have $z+\Ad_{U}(z)\in\caM^{(0)}\cap \lmk {\caM^{(0)}}\rmk'\cap\{U\}'=Z(\caM)
\cap \caM^{(0)}=\bbC \unit$, 
which then implies that $z+\Ad_{U}(z)=\unit$.
(We note
that for orthogonal projections 
$p,q$ satisfying $p+q=t\unit$ with $t\in\bbR$, either $p+q=\unit$ or $p=0,\,\unit$ holds,
by considering the spectrum of $p=t\unit-q$.)

We claim $Z(\caM^{(0)})=\bbC z+\bbC\unit$. 
Now, for any projection $s$ in $Z(\caM^{(0)})$, 
$zs$ is a projection in $Z(\caM^{(0)})$. Therefore either $zs=0$ or $zs+\Ad_{U}(zs)=\unit$.
The latter is possible only if $zs=z$ because $z+\Ad_{U}(z)=\unit$. 
Similarly, we have $(\unit-z)s=0$ or $(\unit-z)s=\unit-z$.
Hence we have $Z(\caM^{(0)})=\bbC z+\bbC\unit$, proving the claim.

Hence $\caM^{(0)}$ is a summation of two factors $\caM^{(0)}=\caM^{(0)}z\oplus \caM^{(0)}(\unit-z)$.
Because $\Ad U(z)=\unit-z$, 
these factors $\caM^{(0)}z$, $\caM^{(0)}(\unit-z)$ are isomorphic.
In particular, they have the same type.
\end{proof}

Let  $(\caM_1,\Ad_{\Gamma_1})$ and $(\caM_2,\Ad_{\Gamma_2})$ be spatially
graded von Neumann algebras acting on 
on $\caH_1$, $\caH_2$
with grading operators $\Gamma_1$, $\Gamma_2$.
We define a product and involution on the algebraic tensor product $\caM_1\odot \caM_2$
by
\begin{align}
(a_1\hat\otimes b_1)(a_2\hat\otimes b_2)
&=(-1)^{\partial b_1\partial a_2}(a_1a_2\hat\otimes b_1b_2), \nonumber \\
(a\hat\otimes b)^* &=(-1)^{\partial a\partial b} a^*\hat\otimes b^*.
%
\end{align}
for homogeneous elementary tensors.
The algebraic tensor product with this multiplication and involution 
is a $*$-algebra, denoted $\caM_1\,\hat{\odot}\,\caM_2$.
On the Hilbert space $\caH_1\otimes\caH_2$,
\begin{align}\label{pprep}
\pi (a\hat\otimes b)
:=a\Gamma_1^{\partial b}\otimes b
\end{align}
for homogeneous $a\in \caM_1$, $b\in\caM_2$
defines a faithful $*$-representation of $\caM_1\,\hat{\odot}\,\caM_2$.
We call the von Neumann algebra generated by $\pi(\caM_1\,\hat{\odot}\,\caM_2)$
the graded tensor product of 
$(\caM_1,\caH_1,\Gamma_1)$ and $(\caM_2,\caH_2,\Gamma_2)$ 
and denote it by $\caM_1 \hat\otimes \caM_2$. 
It is simple to check that 
$\caM_1 \hat\otimes \caM_2$ is a spatially graded von Neumann algebra with a grading operator
$\Gamma_1\otimes \Gamma_2$.

For  $a\in \caM_1$ and homogeneous $b\in\caM_2$,
we denote $\pi(a\hat\otimes b)$ by $a{\hat\otimes} b$, embedding $\caM_1\,\hat{\odot}\,\caM_2$
in $\caM_1\hat\otimes\caM_2$.
Note that $\partial(a\hat\otimes b) = \partial(a) + \partial(b)$ for homogeneous
$a\in \caM_1$ and $b\in \caM_2$.

\begin{lem}\label{apni}
For each $i=1,2$,
let $(\caM_i,\Gamma_i)$  be 
balanced central graded von Neumann algebras on a Hilbert space $\caH_i$.
Suppose that $\caM_1\hat\otimes \caM_2$ is of type I factor.
Then both of $\caM_1$ and $\caM_2$ are type I.
\end{lem}
\begin{proof}
By Lemma \ref{apichi}, all of
$\caM_1$, $\caM_2$, $\caM_1^{(0)}$, $\caM_2^{(0)}$ 
are either a factor or a direct sum of two factors of the same type.
Then by Lemma A.1 of \cite{1dFermi}, the type of $\caM_i$ and $\caM_i^{(0)}$
are the same for each $i=1,2$.

Because $\caM_1\hat\otimes \caM_2$ is a type I factor,
it has a faithful semifinite normal trace $\tau$ whose restriction
to $\caM_1^{(0)}\otimes \caM_2^{(0)}$ is also a faithful semifinite normal trace.
Therefore, from Theorem 2.15 of \cite{takesaki}, $\caM_1^{(0)}\otimes \caM_2^{(0)}$ is semifinite.
From Theorem 2.30, it means both of $\caM_1^{(0)}$ and $\caM_2^{(0)}$ are semifinite.
Because $\caM_1\hat\otimes \caM_2$ is a type I factor,
for the set of orthogonal projections $\caP\lmk \caM_1\hat\otimes \caM_2\rmk $ of
$\caM_1\hat\otimes \caM_2$,
$\tau\lmk\caP\lmk \caM_1\hat\otimes \caM_2\rmk\rmk$
is a countable set.
It means that   for the set of orthogonal projections $\caP\lmk \caM_1^{(0)}\otimes \caM_2^{(0)}\rmk $ of
$\caM_1^{(0)}\otimes \caM_2^{(0)}$,
$\tau\lmk\caP\lmk \caM_1^{(0)}\hat\otimes \caM_2^{(0)}\rmk\rmk$
is also countable.
It means that $\caM_1^{(0)}\otimes \caM_2^{(0)}$ is not of type II.
It means both of $\caM_1^{(0)}$ and $ \caM_2^{(0)}$ are type I, hence 
 both of $\caM_1$ and $ \caM_2$ are type I.
\end{proof}
\begin{lem}\label{sanap}
Let $\caH$ be a Hilbert space with a self-adjoint unitary $\Gamma$ 
that gives a grading for $\caB(\caH)$. 
Let $\caM_1$, $\caM_2$  be $\Ad_{\Gamma}$-invariant
type I von Neumann subalgebras of $\caB(\caH)$ with $\caM_1 \vee \caM_2 = \caB(\caH)$.
Suppose with respect to the grading given by $\Ad_{\Gamma}$
that both of $\caM_1$ and $\caM_2$ have a center of the form
$Z\lmk \caM_i\rmk=\bbC\unit+\bbC V_i$ with a self-adjoint odd unitary $V_i$.
Suppose further that
\begin{align}\label{anti}
ab- (-1)^{\partial a \partial b} b a = 0,\quad \text{for homogeneous}\quad a\in\caM_1, b \in\caM_2.
\end{align}
Then there are Hilbert spaces $\caK_1,\caK_2$ and a unitary
$W:\caH\to\caK_1\otimes\caK_2\otimes \bbC^2$ such that
\begin{align}\label{apgo}
\begin{split}
\Ad W \lmk \caM_1^{(0)}\rmk=\caB(\caK_1)\otimes \bbC\unit_{\caK_2}\otimes \bbC\unit_{\bbC^2},\\
\Ad W \lmk \caM_2^{(0)}\rmk=\bbC\unit_{\caK_1} \otimes \caB(\caK_2)\otimes \bbC\unit_{\bbC^2},\\
\Ad W (V_1)=\unit_{\caK_1}\otimes \unit_{\caK_2}\otimes \sigma_z,\\
\Ad W(V_2)=\unit_{\caK_1}\otimes \unit_{\caK_2}\otimes \sigma_x.
\end{split}
\end{align}
\end{lem}
\begin{proof}
By Proposition 2.9 of \cite{1dFermi} (with $G$ a trivial group),
$\caM_i^{(0)}$ $i=1,2$ are type I factors.
By the assumption (\ref{anti}),
$\caM_1^{(0)}$ and $\caM_2^{(0)}$ commute.
From (\ref{anti}) and the fact that $V_i$ belongs to center of $\caM_i$,
 we know that
both of $V_1$ and $V_2$ commute with $\caM_1^{(0)}$ and $\caM_2^{(0)}$.
As a result, there are Hilbert spaces 
$\caK_1,\caK_2,\caK_3$ and a unitary $\hat W: \caH\to\caK_1\otimes \caK_2\otimes \caK_3$
such that 
\begin{align}\label{apyon}
\begin{split}
\Ad \hat W \lmk \caM_1^{(0)}\rmk=\caB(\caK_1)\otimes \bbC\unit_{\caK_2}\otimes \bbC\unit_{\caK_3},\\
\Ad \hat W \lmk \caM_2^{(0)}\rmk=\bbC\unit_{\caK_1} \otimes \caB(\caK_2)\otimes \bbC\unit_{\caK_3},\\
\Ad \hat W (V_1)=\unit_{\caK_1}\otimes \unit_{\caK_2}\otimes y_1,\\
\Ad \hat W(V_2)=\unit_{\caK_1}\otimes \unit_{\caK_2}\otimes y_2,
\end{split}
\end{align}
with $y_1,y_2$ self-adjoint unitaries on $\caK_3$.
Because $V_2 V_1 V_2^*=-V_1$,
we have $y_2 y_1y_2^*=-y_1$.
With $y_1=r_{1+}-r_{1-}$ as a spectral projection,
this means that $y_2 r_{1\pm }y_2=r_{1\mp}$,
and $u:=r_{1+} y_2 r_{1-}: r_{1-}\caK_3\to r_{1+}\caK_3$
is a unitary.
Let $\{e_1,e_2\}$ be the standard basis of $\bbC^2$
and set $v:\caK_3\to \bbC^2\otimes r_{1+}\caK$ be a unitary given 
by
\begin{align}
v\xi:= e_1\otimes r_{1+}\xi+e_2\otimes ur_{1-}\xi,\quad \xi\in\caK_3.
\end{align}
It is then straightforward to show that
\begin{align}
\begin{split}
\Ad v(y_1)=\sigma_z\otimes \unit_{r_{1+}\caK},\\
\Ad v(y_2)=\sigma_x\otimes \unit_{r_{1+}\caK}.
\end{split}
\end{align}
From this, (\ref{apyon}) and $\caM_1 \vee \caM_2 = \caB(\caH)$,
we see that $r_{1+}\caK$ is one-dimensional.
Hence $W:=(\unit_{\caK_1}\otimes \unit_{\caK_2}\otimes v)\hat W
: \caH\to \caK_1\otimes \caK_2\otimes \bbC^2$
defines a unitary
satisfying (\ref{apgo}).

\end{proof}

\section{Miscellaneous Lemmas}
It is elementary to show the following Lemma. We omit the proof.
\begin{lem}\label{lem42}
For $\sigma=L,R$,
let $\caB_\sigma$  be a graded $C^*$-algebra 
with a grading automorphism $\Theta_\sigma$.
Let $(\caH_\sigma,\pi_\sigma)$ be an irreducible representation of $\caB_\sigma$
with a self-adjoint unitary $\Gamma_\sigma$ implementing
$\Theta_\sigma$.
Let $\zeta_\sigma\in\Aut^{(0)}(\caB_\sigma)$.
Then the followings hold.
\begin{description}
\item[(i)]
If each $\zeta_\sigma$ is implemented by a unitary $u_\sigma$ on 
$(\caH_\sigma,\pi_\sigma)$, i.e., $\Ad \lmk u_\sigma\rmk\circ\pi_\sigma=\pi_\sigma\circ\zeta_\sigma$,
then $u_\sigma$ is homogeneous with respect to $\Ad(\Gamma_\sigma)$ and
\begin{align}
\begin{split}
&\Ad\lmk \unit\otimes u_R\rmk\circ\lmk\pi_L\hat\otimes \pi_R\rmk
=\lmk\pi_L\hat\otimes \pi_R\rmk\circ\lmk \id_{\caB_L}\hat\otimes \zeta_R\rmk,\\
&\Ad\lmk u_L\otimes \Gamma_R^{\partial u_L}\rmk\circ
\lmk\pi_L\hat\otimes \pi_R\rmk
=\lmk\pi_L\hat\otimes \pi_R\rmk\lmk \zeta_L\hat\otimes \id_{\caB_R}\rmk.
\end{split}
\end{align}
Here $\partial u_L$ denotes the grade of $u_L$ with respect to $\Ad(\Gamma_L)$.
\item[(ii)]
Suppose that there are unitaries $U_\sigma$, $\sigma=L,R$ on $\caH_L\otimes\caH_R$
such that 
\begin{align}
\begin{split}
&\Ad\lmk U_R\rmk\circ\lmk\pi_L\hat\otimes \pi_R\rmk
=\lmk\pi_L\hat\otimes \pi_R\rmk\circ\lmk \id_{\caB_L}\hat\otimes \zeta_R\rmk,\\
&\Ad\lmk U_L\rmk\circ
\lmk\pi_L\hat\otimes \pi_R\rmk
=\lmk\pi_L\hat\otimes \pi_R\rmk\lmk \zeta_L\hat\otimes \id_{\caB_R}\rmk.
\end{split}
\end{align}
Then there are unitaries $u_\sigma\in\caU(\caH_\sigma)$
such that $\Ad \lmk u_\sigma\rmk\circ\pi_\sigma=\pi_\sigma\circ\zeta_\sigma$,
and
\begin{align}
\unit\otimes u_R=U_R,\quad
u_L\otimes \Gamma_R^{\partial u_L}=U_L.
\end{align}
\end{description}
\end{lem}

\end{document}